\documentclass[a4paper,12pt,notitlepage]{article}
\usepackage{geometry,amsfonts}

\usepackage{amsmath}
\usepackage{graphicx,lscape,color}
\usepackage{epstopdf}

\usepackage{amsfonts}
\usepackage{amsthm,bbm}
\usepackage{multirow}

\usepackage[breaklinks,hypertexnames=false]{hyperref}

\hypersetup{
	colorlinks=true,
	citecolor=blue,
	linkcolor=blue,
	anchorcolor=blue
}

\usepackage{rotating}
\usepackage[english]{babel}
\usepackage[latin1]{inputenc}
\usepackage{latexsym}
\usepackage{amssymb}
\usepackage{enumerate}

\newtheorem{theorem}{Theorem}
\newtheorem*{theorem*}{Theorem}

\newtheorem{lemma}{Lemma}

\input{tcilatex}

\newcommand{\argmax}{\operatorname*{argmax}}

\begin{document}

\title{Posterior Average Effects}
\author{St\'ephane
Bonhomme\thanks{%
University of Chicago. Email: sbonhomme@uchicago.edu} \and Martin Weidner\thanks{University of Oxford. Email: martin.weidner@economics.ox.ac.uk}}

\date{$\quad $\\ REVISED DRAFT:  September 2021 }
\vskip 3cm\maketitle

\begin{abstract}
\noindent    
Economists are often interested in estimating averages with respect to distributions of unobservables, such as moments of individual fixed-effects, or average partial effects in discrete choice models. For such quantities, we propose and study posterior average effects (PAE), where the average is computed {conditional} on the sample, in the spirit of empirical Bayes and shrinkage methods. While the usefulness of shrinkage for prediction is well-understood, a justification of posterior conditioning to estimate population averages is currently lacking. We show that PAE have minimum worst-case specification error under various forms of misspecification of the parametric distribution of unobservables. In addition, we introduce a measure of informativeness of the posterior conditioning, which quantifies the worst-case specification error of PAE relative to parametric model-based estimators.  As illustrations, we report PAE estimates of distributions of neighborhood effects in the US, and of permanent and transitory components in a model of income dynamics.

\bigskip


\end{abstract}

\baselineskip21pt

\bigskip

\bigskip

\setcounter{page}{0}\thispagestyle{empty}

\newpage

\section{Introduction\label{Intro_sec}}

In many settings, applied researchers wish to estimate population averages with respect to a distribution of unobservables. This includes moments of individual fixed-effects in panel data, and average partial effects in discrete choice models, which are expectations with respect to some distribution of shocks or heterogeneity. The standard approach in applied work is to assume a parametric form for the distribution of unobservables, and to compute the average effect under that assumption. For example, in binary choice, researchers often assume normality of the error term, and compute average partial effects under normality. This ``model-based'' estimation of average effects is justified under the assumption that the parametric model is {correctly specified}. 

In this paper, we consider a different approach, where the average effect is computed {conditional on the observation sample}. We refer to such estimators as ``posterior average effects'' (PAE). Posterior averaging is appealing for prediction purposes, and it plays a central role in Bayesian and empirical Bayes approaches (e.g., Berger, 1980, Morris, 1983). Here we focus instead on the estimation of population expectations. Our goal is twofold: to propose a novel class of estimators, and to provide a frequentist framework to understand when and why posterior conditioning may be useful in estimation. Our main result will show that PAE have robustness properties when the parametric model is {misspecified}.  

PAE are closely related to empirical Bayes (EB) estimators, which are increasingly popular in applied economics. Consider a fixed-effects model of teacher quality, which is our main example. When the number of observations per teacher is small, the dispersion of teacher fixed-effects is likely to overstate that of true teacher quality, since teacher effects are estimated with noise. An alternative approach is to postulate a prior distribution for teacher quality --- typically, a normal --- and report posterior estimates, holding fixed the values of the mean and variance parameters. The hope is that such EB estimates, which are shrunk toward the prior, are less affected by noise than the teacher fixed-effects (e.g., Kane and Staiger, 2008, Chetty \textit{et al.}, 2014, Angrist \textit{et al.}, 2017). However, while EB estimates are well-justified predictors of the quality of individual teachers, it is not obvious how to aggregate them across teachers when the goal is to estimate a population average such as a moment or a distribution function.   

As an example, suppose we wish to estimate the distribution function of teacher quality evaluated at a point. Since this quantity is an average of indicator functions, the PAE is simply an average of posterior means --- that is, of EB estimates --- {of the indicator functions}. This estimator is available in closed form. However, the PAE differs from the empirical distribution of the EB estimates of teacher effects. In particular, while the  variance of EB estimates is too small relative to that of latent teacher quality, the PAE has the correct variance. Related applications of PAE include settings involving neighborhood/place effects (Chetty and Hendren, 2017, Finkelstein \textit{et al.}, 2017) or hospital quality (Hull, 2018).

Importantly, although posterior averages have desirable properties for predicting individual parameters, their usefulness for estimating {population average quantities} is not evident. For example, suppose that teacher quality is normally distributed. In this case, a model-based normal estimator of the distribution of teacher quality is consistent. Moreover, it is asymptotically efficient when means and variances are estimated by maximum likelihood. Hence, in the correctly specified case, there is no reason to deviate from the standard model-based approach and compute posterior estimators. The main insight of this paper is that, under {misspecification} --- e.g., when teacher quality is not normally distributed --- conditioning on the data using PAE can be beneficial.

To study estimators under misspecification, we focus on specification error, which is the population discrepancy between the probability limit of an estimator and the true parameter value. In our main results, we show that PAE have {minimum worst-case specification error}, where the worst case is computed in a nonparametric neighborhood of the reference parametric distribution (e.g., a normal). Specifically, we show that, when neighborhoods are defined in terms of the Pearson chi-squared divergence, PAE have minimum worst-case specification error within a large class of estimators, for any neighborhood size smaller than a threshold value that we characterize. In addition, when broadening the class of neighborhoods to $\phi$-divergences, we show that, while PAE do not have minimum worst-case specification error in general in fixed-size neighborhoods, they achieve minimum worst-case specification error under {local} misspecification, i.e. when the size of the neighborhood tends to zero.

In our examples and illustrations, we find that the information contained in the posterior conditioning is setting-specific. This is intuitive, since although PAE have minimum worst-case specification error under our conditions, the specification error is not zero in general and it varies between applications. PAE tend to behave better when the realizations of outcome variables (such as test scores) are more informative about the values of the unobservables (such as the quality of a teacher). Consistently with this intuition, our local result suggests quantifying the ``informativeness'' of the posterior conditioning using an easily computable $R^2$ coefficient.

While our theoretical results focus on population specification error, in practice PAE are also affected by sampling error, due to the fact that the sample size --- e.g., the number of teachers --- is not infinite. A common approach to account for both sampling variability and specification error is to focus on mean squared error. In general, PAE do not have minimum mean squared error: indeed, in finite samples, model-based estimators can have smaller mean squared error than PAE. In Bonhomme and Weidner (2018), we show how to construct estimators that minimize mean squared error under local asymptotic misspecification. However, such estimators depend on the neighborhood size. In contrast, PAE do {not} require taking a stand on the degree of misspecification through the size of the neighborhood, and they are simple to implement and do not depend on tuning parameters. To complement the theory, we report the results of a Monte Carlo simulation, where we compare the performance of the PAE to those of a model-based estimator and a nonparametric deconvolution-based estimator. We find that, while the model-based estimator tends to perform best under correct specification, the performance of the PAE appears less sensitive to misspecification than those of the model-based and nonparametric estimators.

To illustrate the scope of PAE for applications, we then consider two empirical settings. In the first one, we study the estimation of neighborhood/place effects in the US. Chetty and Hendren (2017) report estimates of the variance of neighborhood effects, as well as EB estimates of those effects. Our goal is to estimate the distribution of effects across neighborhoods. We find that, when using a normal prior as in Chetty and Hendren (2017), our posterior estimator of the distribution function of neighborhood effects across commuting zones is not normal. However, we also show through simulations and computation of our posterior informativeness measure that the signal-to-noise ratio in the data is not high enough to be confident about the exact shape of the distribution. Hence, in this setting, PAE inform our knowledge of the distribution of neighborhood effects, and motivate future analyses using more flexible model specifications and individual-level data. 

In the second empirical illustration, our goal is to estimate the distributions of latent components in a permanent-transitory model of income dynamics (e.g., Hall and Mishkin, 1982, Blundell \textit{et al.}, 2008), where log-income is the sum of a random-walk component and a component that is independent over time. Researchers often estimate the covariance structure of the latent components in a first step. Then, in order to document distributions or to use the income process in a consumption-saving model, they often assume Gaussianity. However, there is increasing evidence that income components are not Gaussian (e.g., Geweke and Keane, 2000, Hirano, 2002, Bonhomme and Robin, 2010, Guvenen \textit{et al.}, 2016). We estimate posterior distribution functions of permanent and transitory income components using recent waves from the Panel Study of Income Dynamics (PSID). Our PAE estimates suggest some departure from Gaussianity, especially for the transitory income component.

We analyze several extensions. First, we describe the form of PAE in several models, including binary choice and censored regression. Second, we discuss how to construct confidence intervals and specification tests based on PAE. Lastly, we revisit the question of optimality of EB estimates for {predicting} individual parameters. By extending our misspecification analysis from worst-case specification error of sample averages to worst-case mean squared prediction error, we show that EB estimators remain optimal, up to small-order terms, under local deviations from normality.

\paragraph{Related literature and outline.}

PAE are closely related to parametric EB estimators (Efron and Morris, 1973, Morris, 1983). For recent econometric applications of shrinkage methods (James and Stein, 1961, Efron, 2012), see Hansen (2016), Fessler and Kasy (2018), and Abadie and Kasy (2018). Recent contributions to nonparametric EB methods are Koenker and Mizera (2014) and Ignatiadis and Wager (2019). 

Our analysis is also related to deconvolution and other nonparametric approaches. However, in our framework we allow for forms of misspecification under which the quantity of interest is not consistently estimable, and we search for estimators that have the smallest specification error.

In panel data settings, Arellano and Bonhomme (2009) study the asymptotic properties of random-effects estimators of averages of functions of covariates and individual effects. They show that, when the distribution of individual effects is misspecified whereas the other features of the model are correctly specified, PAE are consistent as $n$ and $T$ tend to infinity. By contrast, in our setup, only $n$ tends to infinity, and misspecification may affect the entire joint distribution of unobservables. 

Our analysis also connects to the literature on robustness to model misspecification (e.g., Huber and Ronchetti, 2009, Kitamura \textit{et al.}, 2013, Andrews \textit{et al.}, 2017, 2020, Armstrong and Koles\'ar, 2018, Bonhomme and Weidner, 2018, Christensen and Connault, 2019). Here our aim is to propose and justify a class of simple, practical  estimators.

The plan of the paper is as follows. In Section \ref{Sec_mainex} we motivate the analysis by considering a fixed-effects model of teacher quality. In Section \ref{Sec_framework} we present our framework and derive our main theoretical results. In Section \ref{Sec_illustration} we illustrate the use of PAE in two empirical settings. In Section \ref{sec_comp} we describe several extensions. Finally, we conclude in Section \ref{Sec_conclusion}. Replication codes are available as \href{https://sites.google.com/site/stephanebonhommeresearch/}{\color{blue}{online material}}.

\section{Motivating example: a fixed-effects model\label{Sec_mainex}}

To motivate the analysis, we start by considering the following model
\begin{equation}Y_{ij}=\alpha_i+\varepsilon_{ij},\qquad \quad i=1,...,n,\qquad j=1,...,J.\label{FE_mod}\end{equation}
To fix ideas, we will think of $Y_{ij}$ as an average test score of teacher $i$ in classroom $j$, $\alpha_i$ as the quality of teacher $i$, and $\varepsilon_{ij}$ as a classroom-specific shock. There are $n$ teachers and $J$ observations per teacher. For simplicity, we abstract away from covariates (such as students' past test scores), but those will be present in the framework we will introduce in the next section. Although here we focus on teacher effects, this model is of interest in other settings, such as the study of neighborhood effects, school effectiveness, or hospital quality, for example.

Suppose we wish to estimate a feature of the distribution of teacher quality $\alpha$. As an example, here we consider the distribution function of $\alpha$ at a particular point $a$,
$$F_{\alpha}(a)=\mathbb{E}\left[ \boldsymbol{1}\{\alpha\leq a\}\right],$$
which is the percentage of teachers whose quality is below $a$. A first estimator is the empirical distribution of the fixed-effects estimates $\widehat{\alpha}_i=\overline{Y}_i=\frac{1}{J}\sum_{j=1}^JY_{ij}$, for all teachers $i=1,...,n$; that is,
\begin{equation}\widehat{F}^{\rm FE}_{\alpha}(a)=\frac{1}{n}\sum_{i=1}^n \boldsymbol{1}\{\overline{Y}_i\leq a\},\label{eq_cdf_FE}
\end{equation}
where FE stands for ``fixed-effects''. An obvious issue with this estimator is that $\overline{Y}_i=\alpha_i+\overline{\varepsilon}_i$ is a noisy estimate of $\alpha_{i}$, where $\overline{\varepsilon}_i=\frac{1}{J}\sum_{j=1}^J\varepsilon_{ij}$. Indeed, due to the presence of noise, for fixed $J$ and $n$ tends to infinity the distribution $\widehat{F}^{\rm FE}_{\alpha}$ tends to be {too dispersed} relative to $F_{\alpha}$ (although one can show that $\widehat{F}^{\rm FE}_{\alpha}(a)$ is consistent for $F_{\alpha}(a)$ as $J$ tends to infinity jointly with $n$ under mild conditions, see Jochmans and Weidner, 2018).

A different strategy is to model the joint distribution of $\alpha,\varepsilon_{1},...,\varepsilon_{J}$. A simple specification is a multivariate normal distribution with means $\mu_{\alpha}$ and $\mu_{\varepsilon}=0$, and variances $s_{\alpha}^2$ and $s_{\varepsilon}^2$. This specification can easily be made more flexible by allowing for different $s_{\varepsilon_j}^2$'s across $j$, for correlation between the different $\varepsilon_j$'s, or for means and variances being functions of covariates, for example. Under the assumption that all components are uncorrelated, $\mu_\alpha$, $s_{\alpha}^2$ and $s^2_{\varepsilon}$ can be consistently estimated for fixed $J$ as $n$ tends to infinity, using quasi-maximum likelihood or minimum distance based on mean and covariance restrictions.

Given estimates $\widehat{\mu}_{\alpha}$, $\widehat{s}_{\alpha}^2$, $\widehat{s}_{\varepsilon}^2$, we can compute empirical Bayes (EB) estimates (Morris, 1983) of the $\alpha_i$ as 
\begin{equation}\mathbb{E}\, [\alpha\, |\, Y=Y_i]=\widehat{\mu}_{\alpha}+\widehat{\rho}(\overline{Y}_i-\widehat{\mu}_{\alpha}),\quad i=1,...,n,\label{eq_EB}\end{equation}
where the expectation is taken with respect to the posterior distribution of $\alpha$ given $Y=Y_i$ for $\widehat{\mu}_{\alpha}$, $\widehat{s}_{\alpha}^2$, $\widehat{s}_{\varepsilon}^2$ fixed, and $\widehat{\rho}=\frac{\widehat{s}_{\alpha}^2}{\widehat{s}_{\alpha}^2+\widehat{s}_{\varepsilon}^2/J}$ is a shrinkage factor.  Here, $Y_i$ are vectors containing all $Y_{ij}$, $j=1,...,J$. The EB estimates in (\ref{eq_EB}) are well-justified as predictors of the $\alpha_i$, since (when treating $\widehat{\mu}_{\alpha}$, $\widehat{s}_{\alpha}^2$, $\widehat{s}_{\varepsilon}^2$ as fixed) $\widehat{\mu}_{\alpha}+\widehat{\rho}(\overline{Y}_i-\widehat{\mu}_{\alpha})$ is the minimum mean squared error predictor of $\alpha_i$ under normality.

Given their rationale for prediction purposes, it is appealing to try and aggregate the EB estimates in order to estimate our target quantity $F_{\alpha}(a)$. A possible estimator is 
\begin{equation}\widehat{F}^{\rm PM}_{\alpha}(a)=\frac{1}{n}\sum_{i=1}^n \boldsymbol{1}\left\{\widehat{\mu}_{\alpha}+\widehat{\rho}(\overline{Y}_i-\widehat{\mu}_{\alpha})\leq a\right\},\label{eq_cdf_PM}\end{equation}
where PM stands for ``posterior means''. For fixed $J$ as $n$ tends to infinity, the EB estimates tend to be {less dispersed} than the true $\alpha_i$, and $\widehat{F}^{\rm PM}_{\alpha}(a)$ is inconsistent in general. Indeed, while in large samples the variance of the fixed-effects estimates is $\rho^{-1}s_{\alpha}^2>s_{\alpha}^2$, the variance of the EB estimates is $\rho s_{\alpha}^2<s_{\alpha}^2$, where $\rho=\frac{s_{\alpha}^2}{s_{\alpha}^2+s_{\varepsilon}^2/J}$.

Instead of computing the distribution of EB estimates as in (\ref{eq_cdf_PM}), a related idea is to compute the posterior distribution estimator
$$ \widehat{F}^{\rm P}_{\alpha}(a)=\frac{1}{n}\sum_{i=1}^n \mathbb{E}\,\left[ \boldsymbol{1}\{\alpha\leq a\}\,|\, Y=Y_i\right],$$
where P stands for ``posterior''. Using the normality assumption, we obtain
\begin{equation}\widehat{F}^{\rm P}_{\alpha}(a)=\frac{1}{n}\sum_{i=1}^n \Phi\left(\frac{a-\widehat{\mu}_{\alpha}-\widehat{\rho}(\overline{Y}_i-\widehat{\mu}_{\alpha})}{\widehat{s}_{\alpha}\sqrt{1-\widehat{\rho}}}\right),\label{posterior_FE}\end{equation}
where $\Phi$ denotes the distribution function of the standard normal. $\widehat{F}^{\rm P}_{\alpha}(a)$ is an example of a {posterior average effect} (PAE). One can check that it is consistent for fixed $J$ as $n$ tends to infinity, when the distribution of $\alpha,\varepsilon_{1},...,\varepsilon_{J}$ is normal. Under non-normality, $\widehat{F}^{\rm P}_{\alpha}(a)$ is generally inconsistent for fixed $J$ as $n$ tends to infinity. Moreover, the mean and variance of $\widehat{F}^{\rm P}_{\alpha}$ are $(1-\widehat{\rho})\widehat{\mu}_{\alpha}+\widehat{\rho}\frac{1}{n}\sum_{i=1}^n\overline{Y}_i$ and $(1-\widehat{\rho})\widehat{s}_{\alpha}^2+\widehat{\rho}^2\left[\frac{1}{n}\sum_{i=1}^n\overline{Y}_i^2-(\frac{1}{n}\sum_{i=1}^n\overline{Y}_i)^2\right]$, respectively, which are consistent for $\mu_{\alpha}$ and $s_{\alpha}^2$ for fixed $J$ as $n$ tends to infinity.

The last estimator we consider here is directly based on the normal specification for $\alpha$,
\begin{equation} \widehat{F}^{\rm M}_{\alpha}(a)=\Phi\left(\frac{a-\widehat{\mu}_{\alpha}}{\widehat{s}_{\alpha}}\right),\label{model_based_FE}\end{equation}
where M stands for ``model''. This estimator enjoys attractive properties when the distribution of $\alpha,\varepsilon_{1},...,\varepsilon_{J}$ is indeed normal. In this case, $\widehat{F}^{\rm M}_{\alpha}(a)$ is consistent for fixed $J$ as $n$ tends to infinity, and it is efficient when $\widehat{\mu}_{\alpha}$ and $\widehat{s}_{\alpha}^2$ are maximum likelihood estimates. Moreover, the mean and variance of $\widehat{F}^{\rm M}_{\alpha}$ are $\widehat{\mu}_{\alpha}$ and $\widehat{s}_{\alpha}^2$, which are consistent irrespective of normality. However, when $\alpha,\varepsilon_{1},...,\varepsilon_{J}$ is {not} normally distributed, $\widehat{F}^{\rm M}_{\alpha}(a)$ is generally inconsistent for fixed $J$ as $n$ tends to infinity. Moreover, $\widehat{F}^{\rm M}_{\alpha}(a)$ only depends on the data through the mean $\widehat{\mu}_{\alpha}$ and the variance $\widehat{s}^2_{\alpha}$. In particular, $\widehat{F}^{\rm M}_{\alpha}$ is always normal, even when the data show clear evidence of non-normality.

Which one of these estimators should one use? The answer is not obvious, since they are all inconsistent as $n$ tends to infinity for fixed $J$ in general. In a framework that allows for misspecification of the normal distribution of $\alpha,\varepsilon_{1},...,\varepsilon_{J}$, we will show that the PAE $\widehat{F}^{\rm P}_{\alpha}(a)$ has minimum worst-case specification error in certain neighborhoods around the normal reference distribution. To our knowledge, unlike the other three estimators above, posterior estimators of distributions are novel to practitioners. They are easy to implement, and do not depend on additional tuning parameters. Our characterization provides a rationale for reporting them in applications, alongside other parametric and semi-parametric estimators.

Note that one may wish to relax normality by making the specification of $\alpha$, and possibly $\varepsilon_j$, more flexible. Deconvolution and nonparametric maximum likelihood estimators are often used for this purpose (e.g., Delaigle \textit{et al.}, 2008, Bonhomme and Robin, 2010, Koenker and Mizera, 2014). While these estimators may be consistent even when $\alpha$ is not normal, consistency relies on additional restrictions on the model. For example, the assumptions in Kotlarski (1967) require that $\alpha$, $\varepsilon_1$, \ldots, 
$\varepsilon_J$ be mutually {independent}. By contrast, we do {not} impose any such additional conditions in our framework. In Section \ref{Sec_framework}, we will show that asymptotically linear estimators have larger specification error than PAE under the form of misspecification that we consider.

To illustrate that an independence assumption among $\alpha$, $\varepsilon_1$, \ldots, 
$\varepsilon_J$ can be restrictive, consider a situation where  the researcher is concerned that the variance of $\varepsilon_j$ depends on $\alpha$. For instance, the variance of classroom-level shocks may depend on teacher quality. The presence of such conditional heteroskedasticity would invalidate conventional nonparametric deconvolution estimators. By contrast, we will show that $\widehat{F}^{\rm P}_{\alpha}(a)$ has minimum specification error in neighborhoods of distributions that allow for conditional heteroskedasticity. In Section \ref{Sec_illustration} and the appendix, we will compare the finite-sample behavior of the parametric model-based estimator, the PAE, and a nonparametric deconvolution estimator, in data simulated from various specifications of model (\ref{FE_mod}). 

In model (\ref{FE_mod}), the researcher may be interested in estimating other quantities. As an example, consider the coefficient in the population regression of teacher quality $\alpha$ on a vector of covariates $W$; 
that is,
\begin{equation}\overline{\delta}=\left(\mathbb{E}[WW']\right)^{-1}\mathbb{E}[W\alpha].\label{eq_proj_coeff}\end{equation}
In applications, it is common to regress fixed-effects estimates on covariates to help interpret them (as in Dobbie and Fryer, 2013, among many others), and to compute
\begin{equation}	\widehat{\delta}^{\rm FE}=\left(\sum_{i=1}^nW_iW_i'\right)^{-1}\sum_{i=1}^n W_i\overline{Y}_i.\label{eq_proj_coeff_est_FE}\end{equation}
Alternatively, one may regress the EB estimates of $\alpha_i$, as given by (\ref{eq_EB}), on covariates (as in Angrist \textit{et al.}, 2017, and Hull, 2018, for example), and compute
\begin{equation}	\widehat{\delta}^{\rm P}=\left(\sum_{i=1}^nW_iW_i'\right)^{-1}\sum_{i=1}^n W_i\left(\widehat{\mu}_{\alpha}+\widehat{\rho}(\overline{Y}_i-\widehat{\mu}_{\alpha})\right),\label{eq_proj_coeff_est}\end{equation}
which is a PAE based on a normal reference specification for $\alpha$. We will see that, in our framework, the rationale for reporting $\widehat{\delta}^{\rm P}$ or $\widehat{\delta}^{\rm FE}$ depends on the form of misspecification that the researcher is concerned about.

The framework we describe next applies to the estimation of different quantities in a variety of settings. In Section \ref{Sec_illustration} we apply PAE to model (\ref{FE_mod}) and estimate the distribution of neighborhood/place effects in the US (Chetty and Hendren, 2017). In addition, we show that the permanent-transitory model of income dynamics (e.g., Hall and Mishkin, 1982) has a structure similar to model (\ref{FE_mod}), and we report PAE estimates in this context. Lastly, in other models --- such as static or dynamic discrete choice models and models with censored outcomes --- our results motivate the use of PAE as complements to other estimators that researchers commonly report, and we provide examples in Section \ref{sec_comp} and analyze them in the appendix.

\section{Framework and main results\label{Sec_framework}}

In this section we describe our framework to study PAE, and present our main results.

\subsection{Model-based estimators and PAE}

We consider the following class of models,
\begin{equation}
Y_i=g_\beta (U_i,X_i),\label{eq_model}
\end{equation}
where outcomes $Y_i$ and covariates $X_i$ are observed by the researcher, and $U_i$ are unobserved. The function $g_{\beta}$ is known up to the finite-dimensional parameter $\beta$. Our aim is to estimate an average effect of the form
\begin{equation}
\overline{\delta}=\mathbb{E}_{f_0} \left[\delta_{\beta} (U,X)\right],\label{eq_average}
\end{equation}
where $\delta_{\beta}$ is scalar, and known given $\beta$. Here $f_0$ denotes the true density of $U\,|\,X$. The expectation is taken with respect to the product $f_0 f_X$, where $f_X$ is the marginal density of $X$. For conciseness we leave the dependence on $f_X$ implicit. While we focus on a scalar $\delta_{\beta}$, our results continue to hold in the vector-valued case, as we show at the end of this section. In Appendix \ref{App_Ext}, we discuss how to estimate quantities that depend on $f_0$ nonlinearly.

While the researcher does not know the true $f_0$, she has a reference parametric density $f_{\sigma}$ for $U\,|\, X$, which depends on a finite-dimensional parameter $\sigma$. We will allow $f_{\sigma}$ to be misspecified, in the sense that $f_0$ may not belong to $\{f_{\sigma}\}$. However, we will always assume that $g_{\beta}$ is correctly specified. In other words, misspecification will only affect the distribution of $U$ and its dependence on $X$, not the structural link between $(U,X)$ and outcomes.

To estimate $\overline{\delta}$ in (\ref{eq_average}), we assume that the researcher has an estimator $\widehat{\beta}$ that remains consistent for $\beta$ under misspecification of $f_{\sigma}$. More precisely, we will only consider potential true densities $f_0$ such that $\widehat{\beta}$ tends in probability to the true value $\beta$ under $f_0$. For example, in the fixed-effects model (\ref{FE_mod}), consistent estimates of means and variances can be obtained in the absence of normality.
	
To map model (\ref{FE_mod}) to the general notation of this section, note that in this case there are no covariates $X$, and the vector of unobservables $U$ is $$U=\left(\frac{\alpha-\mu_{\alpha}}{s_{\alpha}},\frac{\varepsilon_1}{s_{\varepsilon}},...,\frac{\varepsilon_J}{s_{\varepsilon}}\right)'.$$ 
The vector $\beta$ is $\beta=(\mu_{\alpha},s_{\alpha}^2,s_{\varepsilon}^2)'$. The reference distribution for $U$ is a standard multivariate normal, so the reference density $f_{\sigma}$ is known in this case --- in other words, the parameter $\sigma$ in $f_{\sigma}$ can be omitted. We assume that the researcher has computed an estimator $\widehat{\beta}$, for example by quasi-maximum likelihood or minimum distance, which remains consistent for $\beta$ when $U$ is not normally distributed. 
	
In certain applications, the reference density depends on some parameters $\sigma$ that cannot be consistently estimated absent parametric assumptions. In Appendix \ref{sec_other_ex2}, we describe discrete choice and censored regression models that have this structure. In such settings, we assume that the researcher has an estimator $\widehat{\sigma}$ that tends in probability to some $\sigma_*$ under $f_0$. Unlike $\beta$, the parameter $\sigma_*$ is a model-specific ``pseudo-true value'' that is not assumed to have generated the data. However, in our leading example of model (\ref{FE_mod}), as well as in the model's generalizations that we study in our empirical illustrations in Section \ref{Sec_illustration}, the references to $\widehat{\sigma}$ and $\sigma_*$ can be omitted from all subsequent statements and derivations.

Given $\widehat{\beta}$, $\widehat{\sigma}$, a sample $\{Y_i,X_i,\, i=1,...,n\}$ from $(Y,X)$, and the parametric density $f_{\sigma}$, a {model-based} estimator of $\overline{\delta}$ is 
\begin{equation}
\widehat{\delta}^{\rm M}=\frac{1}{n}\sum_{i=1}^n \mathbb{E}_{f_{\widehat{\sigma}}} \left[\delta_{\widehat{\beta}} (U,X)\,\big|\, X=X_i\right],\label{eq_model_based}
\end{equation}
where, with some abuse of notation, the expectation with respect to $f_{\widehat{\sigma}}$ is computed only over $U$. When not available in closed form, this estimator can be computed by numerical integration or simulation under the parametric density $f_{\widehat{\sigma}}$. It is easy to see that, under standard conditions, $\widehat{\delta}^{\rm M}$ is consistent for $\overline{\delta}$ under correct specification; that is, when $f_{\sigma_*}$ is the true density of $U\,|\, X$.

To construct a posterior estimator, consider the posterior density $p_{\beta,\sigma}$ of $U\,|\, Y,X$. This posterior density is computed using Bayes rule, based on the prior $f_{\sigma}$ on $U\,|\, X$ and the likelihood of $Y\,|\, U,X$ implied by $g_{\beta}$. Formally, let ${\cal{U}}(y,x,\beta)=\{u\,:\, y=g_{\beta}(u,x)\}$. We define, whenever the denominator is non-zero,
\begin{equation}\label{eq_posterior}
p_{\beta,\sigma}(u\,|\, y,x)=\frac{f_{\sigma}(u\,|\,x)\boldsymbol{1}\{u\in {\cal{U}}(y,x,\beta)\}}{\int f_{\sigma}(v\,|\,x)\boldsymbol{1}\{v\in {\cal{U}}(y,x,\beta)\}dv}.
\end{equation}
We will compute $p_{\beta,\sigma}$ analytically in our examples. In Appendix \ref{App_Ext} we describe a simulation-based computational approach when an analytical expression is not available. We define the {posterior average effect} (PAE) as the posterior estimator 
\begin{equation}\label{eq_posterioraverage}
\widehat{\delta}^{\rm P}=\frac{1}{n}\sum_{i=1}^n \mathbb{E}_{p_{\widehat{\beta},\widehat{\sigma}}}\left[\delta_{\widehat{\beta}} (U,X)\,\Big|\, Y=Y_i,X=X_i\right],
\end{equation} 
where, again, the expectation is only taken over $U$. Under standard regularity conditions, it is easy to see that, like $\widehat{\delta}^{\rm M}$, the PAE $\widehat{\delta}^{\rm P}$ is consistent for $\overline{\delta}$ under correct specification.

From a Bayesian perspective, $\widehat{\delta}^{\rm P}$ is a natural estimator to consider when $\beta$ and $\sigma$ are known. Indeed, $\widehat{\delta}^{\rm P}$ is then the posterior mean of $\frac{1}{n}\sum_{i=1}^n\delta_{\beta} (U_i,X_i)$, where the prior on $U_i$ is $f_{\sigma}$, independent across $i$. An alternative Bayesian interpretation is obtained by specifying a nonparametric prior on $f_0$, and computing the posterior mean of $\overline{\delta}$ under this prior, as we discuss in Appendix \ref{App_Ext} in the case where $U$ has finite support. However, a frequentist justification for $\widehat{\delta}^{\rm P}$ appears to be lacking in the literature. Indeed, under correct specification of $f_{\sigma}$, both estimators $\widehat{\delta}^{\rm P}$ and $\widehat{\delta}^{\rm M}$ are consistent, and, as we pointed out in the previous section, $\widehat{\delta}^{\rm P}$ may have a higher variance than $\widehat{\delta}^{\rm M}$. The key difference between model-based and posterior estimators is that $\widehat{\delta}^{\rm P}$ is conditional on the observation sample. An intuitive rationale for the conditioning is the recognition that realizations $Y_i$ may be informative about the values of the unknown $U_i$'s. We next formalize this intuition in a framework that accounts for specification error.

\subsection{Neighborhoods, estimators, and worst-case specification error}

Let $P(\beta,f_0)$ denote the true density of $(Y,U,X)$, where as before we omit the reference to the marginal density of $X$ for conciseness. We assume that, under $P(\beta,f_0)$, $\widehat{\beta}$ is consistent for the true $\beta$, and $\widehat{\sigma}$ is consistent for a model-specific ``pseudo-true'' value $\sigma_{*}$, where $\mathbb{E}_{P(\beta,f_0)} [\psi_{\beta,\sigma_{*}}(Y,X)]=0$ for some moment function $\psi$. For example, $\widehat{\beta}$ and $\widehat{\sigma}$ may be the method-of-moments estimators that solve $\sum_{i=1}^n \psi_{\widehat \beta,\widehat \sigma}(Y_i,X_i) = 0$. In models with no $\sigma$ parameters, such as model (\ref{FE_mod}) and its generalizations, we only assume that $\widehat{\beta}$ is consistent for $\beta$, and that $\mathbb{E}_{P(\beta,f_0)} [\psi_{\beta}(Y,X)]=0$ for some $\psi$. Throughout, we take the estimators $\widehat{\beta}$ (and possibly $\widehat{\sigma}$), and the moment function $\psi$,
as given. In particular, we do not address the question of optimal estimation of $\beta$ under misspecification.

Given a distance measure $d$ and a scalar $\epsilon\geq 0$, we define the following {neighborhood} of the reference density $f_{\sigma}$:
$$\Gamma_{\epsilon}=\left\{f_0\,:\,  d(f_0,f_{\sigma_*})\leq \epsilon,\,\,\, \mathbb{E}_{P(\beta,f_0)} [\psi_{\beta,\sigma_*}(Y,X)]=0\right\}.$$ 
This neighborhood consists of densities of $U\,|\, X$ that are at most $\epsilon$ away from $f_{\sigma_*}$, and under which $\widehat{\beta}$ and $\widehat{\sigma}$ converge asymptotically to $\beta$ and $\sigma_*$, respectively. The case $\epsilon=0$ corresponds to correct specification of the reference density, whereas $\epsilon>0$ corresponds to misspecification.

 For ease of notation we omit the dependence of $\Gamma_{\epsilon}$ on  $\beta$, $\sigma_*$, and $\psi$,
all of which we consider fixed and given in this section. Indeed, 
we assume that the researcher has chosen an estimator $\widehat{\beta}$, and, depending on the setting, an estimator $\widehat{\sigma}$
--- our theory is silent about where these choices come from --- and that she has already observed their realized values
in a large sample. The moment function $\psi$ is determined by this choice
of estimators. Moreover, in large samples, the population values $\beta$ and $\sigma_*$ are arbitrarily close to the observed values $\widehat{\beta}$ and $\widehat{\sigma}$. In our setup, we only consider densities of unobservables $f_0$ that are consistent with those values, in the sense that the moment restriction $\mathbb{E}_{P(\beta,f_0)} [\psi_{\beta,\sigma_*}(Y,X)]=0$ holds. This large-sample logic is consistent with our focus on specification error; see (\ref{DefBias}) below. 

Note that the same logic might suggest imposing that other features of the joint population distribution of the data $(Y,X)$, such as means, covariances, higher-order moments, or even the entire distribution, be kept constant for all $f_0\in\Gamma_{\epsilon}$. Restricting neighborhoods in this way does not affect the results in this section, because those are valid for all possible $\psi$, and one could thus impose additional moment restrictions on $f_0$.

Let us denote the supports of $X$ and $U$ as ${\cal{X}}$ and ${\cal{U}}$, respectively. We assume that $d$ is a $\phi${-divergence} of the form
$$d(f_0,f_{\sigma})=\int_{{\cal{X}}}\int_{{\cal{U}}} \phi \left(\frac{f_0(u\,|\, x)}{f_{\sigma}(u\,|\, x)}\right)f_{\sigma}(u\,|\, x) \, f_X(x) \, du \, dx,$$
where $\phi$ is a convex function that satisfies $\phi(1)=0$ and $\phi''(1)>0$. This family contains as special cases the $\chi^2$ divergence (averaged over $X$), the Kullback-Leibler divergence, the Hellinger distance, and more generally the members of the Cressie-Read family of divergences (Cressie and Read, 1984). It is commonly used to measure misspecification, see Andrews \textit{et al.} (2020) and Christensen and Connault (2019) for recent examples.

We focus on asymptotically linear estimators of $\overline{\delta}$ that satisfy, for a scalar non-stochastic function $\gamma$ and as $n$ tends to infinity,
\begin{equation}\label{ass_aslin}
\widehat{\delta}_{\gamma}=\frac{1}{n}\sum_{i=1}^n \gamma_{\widehat{\beta},\widehat{\sigma}}(Y_i,X_i)+o_{P(\beta,f_0)}(1).\end{equation}
Note that $\widehat{\delta}_{\gamma}$ depends on $\widehat{\beta},\widehat{\sigma}$, but for conciseness we leave the dependence implicit in the notation. Many estimators can be written in this form (see, e.g., Bickel \textit{et al.}, 1993). Given an estimator $\widehat{\delta}_{\gamma}$, we define its $\epsilon$-worst-case {specification error} as
\begin{align}
b_{\epsilon}(\gamma)=\limfunc{sup}_{f_0\in\Gamma_{\epsilon}}\, \left|\mathbb{E}_{P(\beta,f_0)}[\gamma_{{\beta},{\sigma_*}}(Y,X)]-\mathbb{E}_{f_0}[\delta_{\beta}(U,X)]\right|.
    \label{DefBias}
\end{align}
We will take the worst-case specification error $b_{\epsilon}(\gamma)$ to be our measure of how well an estimator $\widehat{\delta}_{\gamma}$ performs under misspecification. It quantifies the maximum discrepancy, under any possible $f_0$ in the neighborhood $\Gamma_{\epsilon}$, between the probability limit of the estimator and the true parameter value. Under suitable regularity conditions, $\mathbb{E}_{P(\beta,f_0)}[\gamma_{{\beta},{\sigma}_{*}}(Y,X)]-\mathbb{E}_{f_0}[\delta_{\beta}(U,X)]$ in (\ref{DefBias}) is the asymptotic bias of $\widehat{\delta}_{\gamma}$ under $P(\beta,f_0)$.

By focusing on the worst-case specification error $b_{\epsilon}(\gamma)$, we abstract from other sources of estimation error. Importantly, we do not account for sampling variability. In Bonhomme and Weidner (2018), we study an alternative approach that consists in minimizing worst-case mean squared error under a local asymptotic --- i.e., as $\epsilon$ tends to zero, $n$ tends to infinity, and $\epsilon n$ tends to a positive constant. Applying this approach to the present case gives estimators that have a smaller worst-case mean squared error than PAE in general. However, unlike PAE, minimum-MSE estimators depend on $\epsilon$, as we will discuss Subsection \ref{subsec_discuss} below. Relative to such estimators, PAE do not require the researcher to take a stand on the degree of misspecification $\epsilon$, and they are easy to implement.

\subsection{Result under small-$\epsilon$ misspecification\label{subsec_local_bias}}

Before stating our first main result, we first characterize the worst-case specification error $b_{\epsilon}(\gamma)$ of estimators $\widehat{\delta}_{\gamma}$ for small $\epsilon$. For conciseness, in the remainder of this section we suppress the reference to $\beta,\sigma_*$ from the notation, and we denote as $\mathbb{E}_*$ and $\limfunc{Var}_*$ expectations and variances that are taken under the reference model $P(\beta,f_{\sigma_*})$. All proofs are in Appendix \ref{App_Proofs}.

\begin{lemma}\label{lem_bias}
     Let $\widetilde \psi(y,x) = \psi(y,x) - \mathbb{E}_* \left[ \psi(Y,X)  \big| X=x \right]$.
      Suppose that one of the following conditions holds:
      \begin{itemize}
          \item[(i)] 
       $\phi(1)=0$, $\phi(r)$ is four times continuously differentiable 
	    with  $\phi''(r)>0$ for all $r >0$, $\mathbb{E}_* [\psi(Y,X)] = 0$,  
    $\mathbb{E}_* \big[ \widetilde \psi(Y,X) \, \widetilde \psi(Y,X)'\big] >0$, and
    $\left| \gamma(y,x) \right|$, $\left| \delta(u,x) \right| $,  $\left| \psi(y,x) \right|$ are bounded over the domain of $Y$, $U$, $X$.
        \item[(ii)] \label{part_ii} Condition (ii) of Lemma~\ref{lem_bias_app} in Appendix~\ref{App_Proofs} holds
        (this alternative condition allows for unbounded $\gamma$, $\delta$, $\psi$, 
        but at the cost of stronger assumptions on $\phi(r)$).
    \end{itemize}
\noindent
    Then, as $\epsilon$ tends to zero we have 
	\begin{align*}
	&b_{\epsilon}(\gamma)=\left|\mathbb{E}_{*}[\gamma(Y,X)-\delta(U,X)]\right| 
	\\ & 
	\,\,\, +\epsilon^{\frac{1}{2}}\bigg\{\frac{2}{\phi''(1)} {\rm Var}_*\Big( 
	 \gamma(Y,X)-  \delta(U,X)  
	 -  \mathbb{E}_{*} \left[\gamma(Y,X)-\delta(U,X)  \, | \, X \right]
	  -\lambda' \widetilde \psi(Y,X) \Big) \bigg\}^{\frac{1}{2}} +{\cal O}(\epsilon),
	\end{align*}
	where $\lambda=\left\{\mathbb{E}_* \big[ \widetilde \psi(Y,X) \, \widetilde \psi(Y,X)'\big]\right\}^{-1}
	 \mathbb{E}_*\left[\left({\gamma}(Y,X)-{\delta}(U,X)\right) \widetilde \psi(Y,X)\right]$.
\end{lemma}

\vskip .3cm

To derive the formula for the worst-case specification error in Lemma \ref{lem_bias}, we maximize the specification error with respect to $f_0$ subject to three contraints: $f_0$ belongs to an $\epsilon$-neighborhood of $f_*$, it is such that the moment condition is satisfied at $(\beta,\sigma_*)$, and it is a density. In part $(i)$ we focus on the case where $\gamma$, $\delta$ and $\psi$ are bounded.
This is satisfied, for example, if those functions and $g(u,x)$ are all continuous, and the domain of $U$ and $X$ is bounded. 
To accommodate situations where supports are unbounded, such as the example of Section \ref{Sec_mainex}, in part $(ii)$ we allow for unbounded functions $\gamma$, $\delta$ and $\psi$,
which only requires existence of third moments under the reference distribution. To guarantee that $b_{\epsilon}(\gamma)$ is well-defined in the unbounded case, we require a regularization of the function $\phi(r)$ for large values of $r$.

Lemma \ref{lem_bias} implies that the small-$\epsilon$ specification error of the PAE is, up to smaller-order terms, proportional to the within-$(Y,X)$ standard deviation of $\delta(U,X)$ under the reference model:
$$b_{\epsilon}(\gamma^{\rm P})=\epsilon^{\frac{1}{2}}\left\{\frac{2}{\phi''(1)}{\limfunc{Var}}_{*}\left(\delta(U,X)-\mathbb{E}_{*}[\delta(U,X)\,|\, Y,X]\right)\right\}^{\frac{1}{2}}+{\cal O}(\epsilon).$$ 
In the fixed-effects model (\ref{FE_mod}) of teacher quality, the worst-case specification error of the PAE $\widehat{F}^{\rm P}_{\alpha}(a)$ is
$$b_{\epsilon}(\gamma^{\rm P})=\epsilon^{\frac{1}{2}}\left\{\frac{4}{\phi''(1)}T\left(\frac{a-{\mu}_{\alpha}}{{\sigma}_{\alpha}},\sqrt{\frac{1-\rho}{1+\rho}}\right)\right\}^{\frac{1}{2}}+{\cal O}(\epsilon),$$
where $T(a,b)=\varphi(a)\int_0^b\frac{\varphi(az)}{1+z^2}dz$ is Owen's T function (Owen, 1956), and $\varphi$ is the standard normal density. The specification error decreases as the number $J$ of observations per teacher increases, and tends to zero as $J$ tends to infinity and the shrinkage factor $\rho$ tends to one.

The next theorem, which holds for all functions $\gamma(y,x)$, subject to regularity conditions, shows that the PAE has minimum worst-case specification error locally.

\begin{theorem}\label{theo_bias}
	Suppose that the conditions of Lemma \ref{lem_bias} hold, and let \begin{equation}\gamma^{\rm P}(y,x)=\mathbb{E}_*[\delta(U,X)\,|\, Y=y,X=x].\label{eq_gamma_P}\end{equation} Then, as $\epsilon$ tends to zero we have 
	$$b_{\epsilon}(\gamma^{\rm P})\leq b_{\epsilon}(\gamma)+{\cal O}(\epsilon).$$
\end{theorem}

\subsection{Result under fixed-$\epsilon$ misspecification\label{subsec_global_bias}}

To show our second main result, let us now focus on the case $\phi(t)=\frac 1 2 (t-1)^2$;
that is, we choose the distance measure $d(f_0,f_{\sigma})$ to be
the Pearson $\chi^2$ divergence.
For this quadratic distance measure, we show that PAE satisfy
a fixed-$\epsilon$ optimality result, which is valid for all values of $\epsilon$ that are smaller than
\begin{align}
   \overline \epsilon \, &= \, \frac{{\rm Var}_*\left[  \delta(U,X)  - \gamma^{\rm P}(Y,X)\right]}
           {2 \, \sup_{u,x}  \Big[ \delta(u,x)  - \gamma^{\rm P}(g(u,x),x) \Big]^2} ,
    \label{MaxEpsilon}       
\end{align} 
where $\gamma^{\rm P}(y,x)$ is given by \eqref{eq_gamma_P}. 
 
\begin{theorem}\label{theo_bias0}
   Assume that $\mathbb{E}_* [\psi(Y,X)] = 0$,  $\phi(t)=\frac 1 2 (t-1)^2$,
   and that $\gamma(Y,X)$ and $\delta(U,X)$ have finite second moments under the reference model.
Then, for $0< \epsilon \leq \overline \epsilon$, we have 
$$b_{\epsilon}(\gamma^{\rm P})\leq b_{\epsilon}(\gamma).$$
\end{theorem}

In Theorem \ref{theo_bias0} we show that  $\gamma^{\rm P}$ is an exact minimizer of the function $ b_{\epsilon}(\gamma)$. This is in contrast with Theorem~\ref{theo_bias}, where we relied on a small-$\epsilon$ approximation. The condition $ \epsilon \leq \overline \epsilon$ guarantees that, for $\gamma=\gamma^{\rm P}$, the constraint $f_0(u\, |\, x) \geq 0$ is non-binding in the optimization problem over $f_0$ in \eqref{DefBias}, implying that the problem has a simple analytic solution. Although, in many settings such as model (\ref{FE_mod}), the parameter of interest $\overline{\delta}$ is not consistently estimable under our assumptions, Theorem \ref{theo_bias0} shows that PAE achieve the smallest possible worst-case specification error when the true distribution $f_0$ lies sufficiently close to the reference distribution $f_{\sigma_*}$, as measured according to the $\chi^2$ divergence.

If the distance measure $d(f_0,f_{\sigma})$ is not a $\chi^2$-divergence,
or if $\epsilon > \overline \epsilon$, then $\gamma^{\rm P}$ is not the exact minimizer of worst-case specification error  $ b_{\epsilon}(\gamma)$. Moreover, in such cases the estimator with minimum worst-case specification error depends on $\epsilon$ in general. However, one can still establish a fixed-$\epsilon$ bound on worst-case specification error, as the next result shows.

\begin{theorem}\label{theo_global}
	Let $\gamma^{\rm P}$ be as in (\ref{eq_gamma_P}),
	and assume that $\phi(r)$ is convex with $\phi(1)=0$.
	Then, for all $\epsilon>0$,
	$$ b_{\epsilon}(\gamma^{\rm P})\leq 2 \,  \limfunc{inf}_{\gamma}\, b_{\epsilon}(\gamma).$$
\end{theorem}

\vskip .3cm

In Theorem \ref{theo_global} we establish a fixed-$\epsilon$ bound on the  worst-case specification error of PAE, which holds for all $\epsilon>0$ and all $\phi$-divergences such that $\phi$ is convex with $\phi(1)=0$. The infimum is taken over all possible functions
	$\gamma(y,x)$, subject to measurability conditions, which we implicitly assume throughout the paper. Although $\widehat{\delta}^{\rm P}$ may not minimize worst-case specification error for finite $\epsilon$,
Theorem \ref{theo_global} shows that its  worst-case specification error is never larger than twice the minimum worst-case specification error. In addition, the factor two in Theorem \ref{theo_global} cannot be improved upon in general, as we show in Appendix \ref{App_Ext} in the context of a simple binary choice model.

\subsection{Discussion\label{subsec_discuss}}

In this subsection, we discuss several features and implications of our main results given by Theorems \ref{theo_bias} and \ref{theo_bias0}.

\paragraph{Uniqueness.}

In the absence of covariates and for known parameters $\beta$, $\sigma_*$, the proof of Theorem \ref{theo_bias} shows that
$\gamma^{\rm P}$ is the  unique minimizer of the first-order worst-case specification error. Likewise, $\gamma^{\rm P}$ is also unique in Theorem \ref{theo_bias0}.  
More generally, if covariates are present and the parameters $\beta$, $\sigma_*$ are estimated, then 
the leading order contribution of $b_{\epsilon}(\gamma)$ is minimized if and only if
  $\gamma(Y,X)  =     \gamma^{\rm P}(Y,X)   + \omega(X) + \lambda'\psi(Y,X)     +  o_{P_*}(1) $,
  for some $\lambda$ and $\omega$ such that $\mathbb{E}_{f_X} [\omega(X)] = 0$ ---
   see part (ii) of Theorem~\ref{theo_bias_app} in Appendix \ref{App_Proofs} for a formal statement. Hence, while the PAE is not the unique minimizer of the local worst-case specification error in this case, any minimizer differs from the PAE by a zero-mean function of $X$ and a linear combination of the moment function $\psi$. In addition, $\widehat{\delta}^{\rm P}$ has smallest variance within the class of minimum worst-case specification error estimators.
  
\paragraph{Form of misspecification.}
Theorems \ref{theo_bias} and \ref{theo_bias0} rely on specific distance measures, $\chi^2$ divergence for the latter and any member of the $\phi$-divergence family for the former. Under other distance measures, the PAE will not have minimum worst-case specification error in general. 
	
	Given a distance measure, the theorems are based on nonparametric neighborhoods that consist of unrestricted distributions of $U\,|\, X$, except for the moment conditions that pin down $\beta$ and $\sigma_*$.  However, if one is willing to make additional assumptions on $f_0$ that further restrict the neighborhood, then one can construct estimators that are more robust than $\widehat{\delta}^{\rm P}$ within a particular class. As an example, consider the fixed-effects model (\ref{FE_mod}). Suppose that, in addition to assuming that $\alpha$, $\varepsilon_1$, ..., $\varepsilon_J$ are mutually uncorrelated, the researcher is willing to assume that they are fully independent. In that case, the distribution of $\alpha$ can be consistently estimated under suitable regularity conditions, provided $J\geq 2$ (Kotlarski, 1967, Li and Vuong, 1998). However, the PAE in (\ref{posterior_FE}) is inconsistent for fixed $J$ as $n$ tends to infinity. As a consequence, the PAE does not minimize worst-case specification error in a semi-parametric neighborhood that consists of distributions with independent marginals.
	
	To elaborate further on this point, consider the coefficient $\overline{\delta}$ in the population regression of $\alpha$ on a covariates vector $W$, see \eqref{eq_proj_coeff}. A possible estimator is the coefficient $\widehat{\delta}^{\rm FE}$ in the regression of the fixed-effects estimates $\overline{Y}_i$ on $W_i$, see (\ref{eq_proj_coeff_est_FE}). Under correct specification of the reference model, $\widehat{\delta}^{\rm FE}$ is consistent for $\overline{\delta}$. However, $\widehat{\delta}^{\rm FE}$ may be inconsistent under the type of misspecification that we allow for, since $\varepsilon_j$ and $W$ may be correlated under $f_0$. For example, $W$ (e.g., teacher absenteeism) may be influenced by $\alpha$ and factors that correlate with $\varepsilon_j$. Theorem \ref{theo_bias} shows that, under such misspecification, the PAE $\widehat{\delta}^{\rm P}$ in (\ref{eq_proj_coeff_est}) has minimum worst-case specification error locally. Nevertheless, if the researcher is confident that $W$ should not enter the outcome equation, and that it is independent of $\varepsilon_j$, then it is natural to report the consistent estimator $\widehat{\delta}^{\rm FE}$.

	\paragraph{Posterior informativeness.}
	Our small-$\epsilon$ calculations can be used to compare the worst-case specification errors of the PAE $\widehat{\delta}^{\rm P}$ to that of the model-based estimator $\widehat{\delta}^{\rm M}$. To see this, let $\gamma^{\rm M}_{\beta,\sigma}(x) =  \mathbb{E}_{f_\sigma}  [ \delta_{\beta}(U,X) \,|\, X=x ]$. Using Lemma~\ref{lem_bias}, the ratio of the two worst-case specification errors satisfies
	\begin{equation}\label{eq_ratio_bias}\underset{\epsilon \rightarrow 0}{\limfunc{lim}}\,\,\,\frac{b_{\epsilon}(\gamma^{\rm P})}{b_{\epsilon}(\gamma^{\rm M})}=\frac{\left\{{\limfunc{Var}}_{*}\left(v(U,X)-\mathbb{E}_{*}[v(U,X)\,|\,Y,X]\right)\right\}^{\frac{1}{2}}}{\left\{{\limfunc{Var}}_{*}\left(v(U,X)\right)\right\}^{\frac{1}{2}}},\end{equation}
	where $v(U,X)$ is the
	population
	residual of $ (\delta(U,X) - {\gamma}^{\rm M}(X))$ on $\widetilde{\psi}(Y,X)$, under the parametric reference model; that is, 
		$v(u,x) =  {\delta}(u,x) - {\gamma}^{\rm M}(x)  + \lambda' \widetilde{\psi}(g(u,x),x)$,
		where
		all functions are evaluated at $\beta,\sigma_*$,
		and $\lambda$ is as defined in Lemma~\ref{lem_bias} for the case $\gamma=\gamma^{\rm M}$. Intuitively, the robustness of $\widehat{\delta}^{\rm P}$ relative to $\widehat{\delta}^{\rm M}$ depends on how informative the outcome values $Y_i$ are for the latent individual parameters $\delta(U_i,X_i)$.

	In practice, we will report an empirical counterpart to the small-$\epsilon$ limit of $1-\frac{b_{\epsilon}^2(\gamma^{\rm P})}{b_{\epsilon}^2(\gamma^{\rm M})}$. This quantity can be simply expressed as the $R^2$ in the population nonparametric regression of $v(U,X)$ on $Y,X$ under the reference model; that is,
	\begin{equation}
	\label{eq_R2}R^2=\frac{\limfunc{Var}_{*}\left(\mathbb{E}_{*}[v(U,X)\,|\,Y,X]\right)}{\limfunc{Var}_{*}\left(v(U,X)\right)},
	\end{equation}
	where with some abuse of notation here $v(U,X)$ denotes the
	sample
	residual of $ ( {\delta}_{\widehat{\beta}}(U,X) - {\gamma}_{\widehat{\beta},\widehat{\sigma}}^{\rm M}(X))$ on $\widetilde{\psi}_{\widehat{\beta},\widehat{\sigma}}(Y,X)$, and expectations and variances are taken with respect to $P(\widehat{\beta},f_{\widehat{\sigma}})$. Using a term from Andrews \textit{et al.} (2020) --- albeit in a different setting --- we refer to $R^2$ in (\ref{eq_R2}) as a measure of the ``informativeness'' of the posterior conditioning, and we will report it in our illustrations. As an example, for $\widehat{F}^{\rm P}_{\alpha}(a)$ in model (\ref{FE_mod}), the informativeness of the posterior conditioning is 
\begin{equation}
\label{eq_R2_FE}
R^2=1-\frac{2T\left(\frac{a-\widehat{\mu}_{\alpha}}{\widehat{s}_{\alpha}},\sqrt{\frac{1-\widehat{\rho}}{1+\widehat{\rho}}}\right)}{\Phi\left(\frac{a-\widehat{\mu}_{\alpha}}{\widehat{s}_{\alpha}}\right)\left[1-\Phi\left(\frac{a-\widehat{\mu}_{\alpha}}{\widehat{s}_{\alpha}}\right)\right]}.\end{equation}
In this case the $R^2$ increases with the number $J$ of observations per teacher, and it tends to one as $J$ tends to infinity.

\paragraph{Multi-dimensional PAE.} 
For simplicity, in this section we have focused on the case where the target parameter $\overline \delta$
in \eqref{eq_average} is scalar. However, our results can be extended to multi-dimensional parameters. 
The definition of worst-case specification error in \eqref{DefBias} is then modified
to 
$$b_{\epsilon}(\gamma)=\limfunc{sup}_{f_0\in\Gamma_{\epsilon}}\, \left\| 
\mathbb{E}_{P(\beta,f_0)}[\gamma(Y,X) ]-\mathbb{E}_{f_0}[ \delta(U,X)]\right\|,$$ where $\| \cdot \|$ is a norm over the vector space in which $\gamma(Y,X)$
and $\delta(U,X)$ take values. 

If $\| \cdot \|_*$ denotes the corresponding dual norm, then we can rewrite
$b_{\epsilon}(\gamma)= \sup_{\|v\|_* = 1} b_{\epsilon}(\gamma,v)$,
where   $b_{\epsilon}(\gamma,v)
= \limfunc{sup}_{f_0\in\Gamma_{\epsilon}}\, \big| \mathbb{E}_{P(\beta,f_0)}[v' \gamma(Y,X)] - \mathbb{E}_{f_0}[v' \delta(U,X)]\big|$. Our minimum worst-case specification error results for PAE for scalar $\overline \delta$ then apply
to $b_{\epsilon}(\gamma,v)$ for every given vector $v$, and the minimum-specification error properties are preserved after taking the supremum over the set of vectors
$v$ with $ \|v\|_* = 1$. Thus, in the multi-dimensional case, PAE minimize worst-case specification error for small $\epsilon$ in the sense of Theorem~\ref{theo_bias}, and for fixed $\epsilon$ under the conditions of Theorem \ref{theo_bias0}. In our leading example of Section~\ref{Sec_mainex}, suppose we are interested in the entire distribution function $F_{\alpha}$. In this case, the average effect is a function indexed by $a$. Taking the supremum norm $\|\cdot\|_\infty$ over distribution functions, we obtain that, as an estimator of $F_{\alpha}$, the PAE minimizes worst-case specification error under suitable conditions.

\paragraph{Mean squared error.}
While we have shown that PAE minimize worst-case specification error locally under the conditions of Theorem \ref{theo_bias}, and for fixed $\epsilon$ under the conditions of Theorem \ref{theo_bias0}, PAE generally do not have minimum mean squared error (MSE). To see this, let us assume that $\beta$ and $\sigma_*$ are known. In a local asymptotic framework where $n$ tends to infinity, $\epsilon$ tends to zero, and $n\epsilon$ tends to a positive constant, and under suitable regularity conditions, we show in Appendix~\ref{App_Ext} that the estimator with minimum worst-case MSE is given by
\begin{align}
\widehat{\delta}^{\rm MMSE}
&=  \left[ 1 - w_{n\epsilon} \right]  \;  \widehat{\delta}^{\rm M}
+   w_{n\epsilon} \,  \widehat{\delta}^{\rm P} ,
&
w_{n\epsilon} :=  \left(1+\frac{\phi''(1)}{2n\epsilon}\right)^{-1} ,
\label{deltaMMSE}
\end{align}
which is a linear combination between the model-based estimator and the PAE. The model-based estimator $\widehat{\delta}^{\rm M}$, which has the smallest asymptotic variance, will be preferred when $\epsilon$ is small relative to $1/n$, while the PAE, which has smallest specification error,
will be preferred when $\epsilon$ is large relative to $1/n$.
However, in order to implement such estimators  $\widehat{\delta}^{\rm MMSE}$ that minimize worst-case MSE, knowledge of $\epsilon$ is required. See Bonhomme and Weidner (2018) for an approach to minimum-MSE estimation.

\section{Simulations and empirical illustrations\label{Sec_illustration}}

In this section, we study two empirical applications: we estimate the distribution of income neighborhood effects in the US, and the distributions of permanent and transitory earnings components in the PSID. We start the section by summarizing the results of a Monte Carlo simulation exercise, in samples generated from various specifications of model (\ref{FE_mod}). 

\subsection{Monte Carlo simulation: summary of results\label{subsec_sim}}

While Theorems \ref{theo_bias} and \ref{theo_bias0} show that PAE minimize worst-case specification error under small-$\epsilon$ and fixed-$\epsilon$ misspecification, respectively, they are silent about other forms of estimation error. In Appendix \ref{sec_App_sim} we report the results of a Monte Carlo simulation exercise, where we compare the performance of PAE and other estimators in finite sample in the fixed-effects model (\ref{FE_mod}), for various specifications. Here we briefly summarize the results from the simulation exercise.

We compare the performance of four estimators: the fixed-effects estimator given by (\ref{eq_cdf_FE}), the PAE given by (\ref{posterior_FE}), the model-based estimator given by (\ref{model_based_FE}), and a nonparametric kernel deconvolution estimator with normal errors (Stefanski and Carroll, 1990). We analyze two sets of data generating processes. When the reference normal distribution for $\alpha_i$ is correctly specified, the model-based estimator performs best, as expected. We find that, while the PAE has both larger bias and variance than the model-based estimator in this case, it is less biased and less variable than both the nonparametric deconvolution estimator and the fixed-effects estimator, especially when the number of measurements $J$ is small (see Appendix Figure \ref{fig_mc}).

We next turn to data generating processes where $\alpha_i$ is not normal, drawn from a skewed Beta distribution. We find that the model-based estimator is substantially biased in this case. The nonparametric deconvolution estimator has smallest bias when errors are normally distributed, but it is heavily biased when errors are non-normal. By contrast, although it has no consistency guarantees in these settings, the PAE tends to perform comparatively well in all situations, for bias and variance  (see Appendix Figure \ref{fig_mc_misp}).  

Overall, the simulations complement our theory by highlighting that, beyond specification error, other sources of estimation error matter in practice. Under correct specification of the reference distribution, the model-based estimator should be preferred. At the same time, our results suggest that, at least in the particular settings we focus on,  the performance of the PAE appears less sensitive to misspecification than those of the model-based and nonparametric deconvolution estimators. Moreover, we find that the robustness gains provided by the PAE depend on the signal-to-noise ratio and the informativeness of the posterior conditioning. We provide details on the simulations in Appendix \ref{sec_App_sim}.

\subsection{Neighborhood effects}

In this subsection and the next, we revisit two applications of models with latent variables. In our first illustration, we focus on a model of neighborhood effects following Chetty and Hendren (2017), using data for the US that these authors made public. In our second illustration, we study a permanent-transitory model of income dynamics (Hall and Mishkin, 1982, Blundell \textit{et al.}, 2008) using the PSID. In both cases, we rely on a normal reference specification and assess how and by how much the posterior conditioning informs the estimates of the parameters of interest. 

Here we start with estimates of neighborhood 
(or ``place'')
effects reported in Chetty and Hendren (2017, CH hereafter). Those were obtained using individuals who moved between different commuting zones at different ages. The outcome variable that we focus on is the causal estimate of the income rank at age 26 of a child whose parents are at the 25 percentile of the income distribution. This is CH's preferred measure of place effect.  

CH report an estimate of the variance of neighborhood effects, corrected for noise. In addition, they report individual predictors. Here we are interested in documenting the entire distribution of place effects. To do so, we consider the model $\widehat{\mu}_{c}=\mu_c+\overline{\varepsilon}_c$, for each commuting zone $c$, where $\widehat{\mu}_{c}$ is a neighborhood-specific fixed-effects reported by CH, $\mu_c$ is the true effect of neighborhood $c$, and $\overline{\varepsilon}_c$ is additive estimation noise. CH also report estimates $\widehat{s}_{c}^2$ of the variances of $\overline{\varepsilon}_c$ for every $c$. When weighted by population, the fixed-effects estimates $\widehat{\mu}_{c}$ have mean zero. We treat neighborhoods as independent observations. The statistics we use for calculations are available at: https://opportunityinsights.org/paper/neighborhoodsii/. Given the aggregate data at hand, we necessarily need to assume that estimates $\widehat{\mu}_{c}$ are independent across neighborhoods $c$, although this might be restrictive in this setting.

We first estimate the variance of place effects $\mu_c$, following CH. We trim the top 1\% percentile of $\widehat{s}_{c}^2$, and weigh all results by population weights. While this differs slightly from CH's approach, which is based on $1/\widehat{s}_{c}^2$ precision weights and no trimming, we replicated the analysis using precision weights in the un-trimmed sample and found similar results. We have information about place effects in $C=590$ commuting zones $c$ in our sample, compared to 595 in the sample without trimming. We estimate a sizable variance of neighborhood fixed-effects: $\limfunc{Var}(\widehat{\mu}_{c})=.077$. In turn, the mean of $ \widehat{s}_{c}^2$ weighted by population is $\widehat{s}_{\overline{\varepsilon}}^2=.047$. Given those, we estimate the variance of place effects as $\widehat{s}_{\mu}^2=\limfunc{Var}(\widehat{\mu}_{c})-\widehat{s}_{\overline{\varepsilon}}^2=.030$. In this setting, the shrinkage factor $\widehat{\rho}_c=\widehat{s}_{\mu}^2/(\widehat{s}_{\mu}^2+\widehat{s}_{c}^2)$ exhibits substantial heterogeneity across commuting zones. Indeed, the mean of $\widehat{\rho}_c$ is .62, and its 10\% and 90\% percentiles are .21 and .93, respectively.

We use a normal with zero mean and variance $\widehat{s}_{\mu}^2$ as a prior for $\mu_c$. Then, we estimate the distribution function of neighborhood effects ${\mu}_c$ using the PAE given by (\ref{posterior_FE}); that is, 
\begin{equation*}\widehat{F}^{\rm P}_{\mu}(a)=\frac{1}{\sum_{c=1}^C\pi_c}\sum_{c=1}^C\pi_c \Phi\left(\frac{a-\widehat{\rho}_c\widehat{\mu}_{c}}{\widehat{s}_{\mu}\sqrt{1-\widehat{\rho}_c}}\right),\end{equation*}
where $\pi_c$ are population weights. In addition, in order to ease the visualization of the results, we will also report estimates of densities, which are the derivatives of the PAE of distribution functions. Note that the density of $\mu$ at $a$ can be approximated for arbitrarily small $h>0$ by the expectation of $\boldsymbol{1}\{|\mu-a|/h\}/2h$. Taking the limit of the corresponding PAE as $h$ tends to zero gives the derivative of $\widehat{F}^{\rm P}_{\mu}$ at $a$. We thus expect derivatives of PAE of distribution functions to enjoy similar minimum-worst-case specification error properties as PAE, but we do not formalize the required assumptions here.

In the top panel of Figure \ref{fig_neighb_FE}, we report several estimates of distribution functions. In the bottom panel, we report the corresponding density estimates. In the left graphs, we show nonparametric kernel estimates of the distribution function (respectively, density) of the fixed-effects $\widehat{\mu}_{c}$, weighted by population (in solid), together with the best-fitting normal (in dashed). The graphs show substantial non-normality of the fixed-effects estimates. In particular, the large variance appears to be driven by some large positive and negative estimates $\widehat{\mu}_{c}$. In the right graphs, we report the PAE $\widehat{F}^{\rm P}_{\mu}$ of the distribution function of true place effects $\mu_c$, with the associated density (in solid). In addition, we show the normal prior, with zero mean and variance $\widehat{s}_{\mu}^2$ (in dashed). The posterior distribution of neighborhood effects differs from the normal prior, although the two estimators have the same variance by construction. In comparison, neighborhood-specific empirical Bayes estimates have a substantially lower dispersion. In Appendix Figure \ref{fig_neighb_PM} we report an estimate of their distribution function $\widehat{F}^{\rm PM}_{\mu}$ and associated density. While $\widehat{s}_{\mu}^2=.030$ and the variance associated with $\widehat{F}^{\rm P}_{\mu}$ is $.030$, the variance of the empirical Bayes estimates is only $.010$. In addition, a specification test that compares model-based estimator and PAE, which we describe in Appendix \ref{App_Ext}, suggests that these differences are statistically significant. Indeed, assuming independence across commuting zones, we obtain p-values below .01 at all deciles except the bottom two.

\begin{figure}[h!]
	\caption{Distribution of neighborhood effects\label{fig_neighb_FE}}
	\begin{center}
		\begin{tabular}{cc}
			Fixed-effects estimates & PAE\\
			\multicolumn{2}{c}{Distribution functions}\\
			\includegraphics[width=80mm, height=60mm]{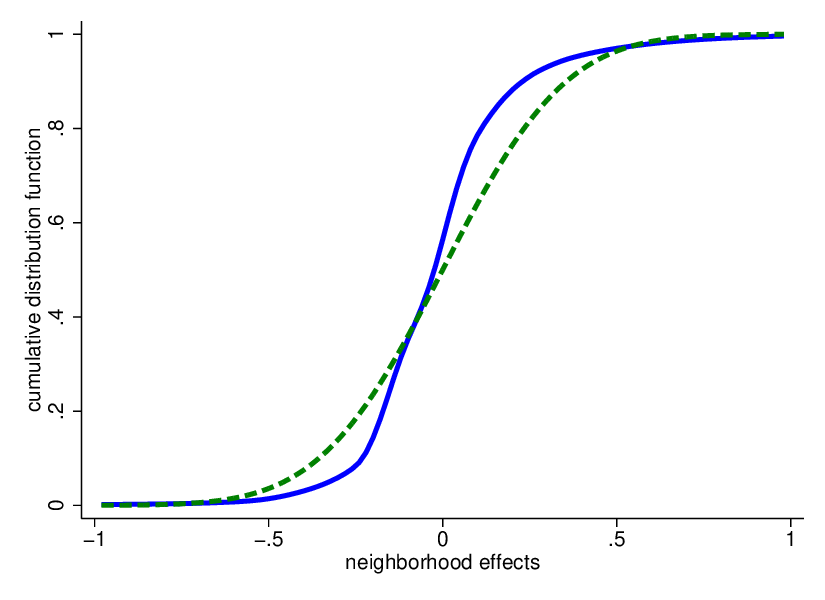}&	\includegraphics[width=80mm, height=60mm]{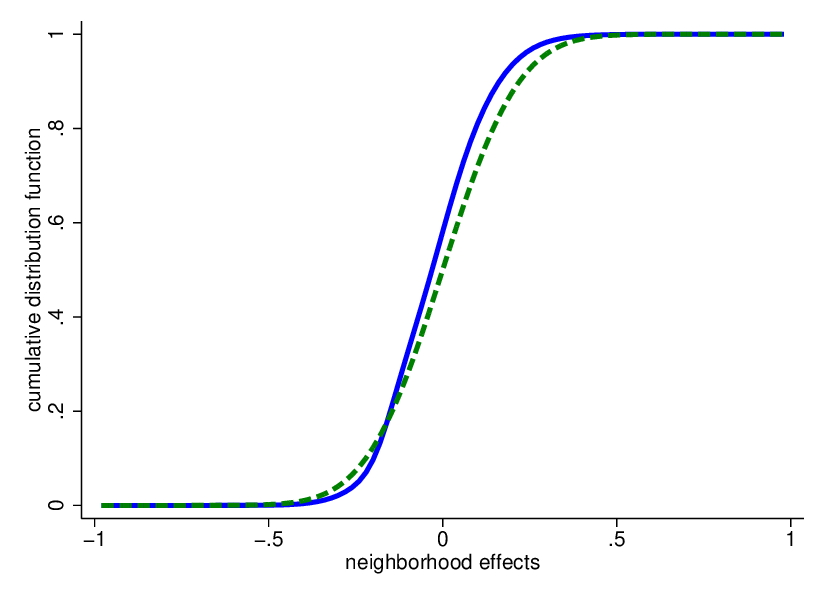}\\
			\multicolumn{2}{c}{Densities}\\
			\includegraphics[width=80mm, height=60mm]{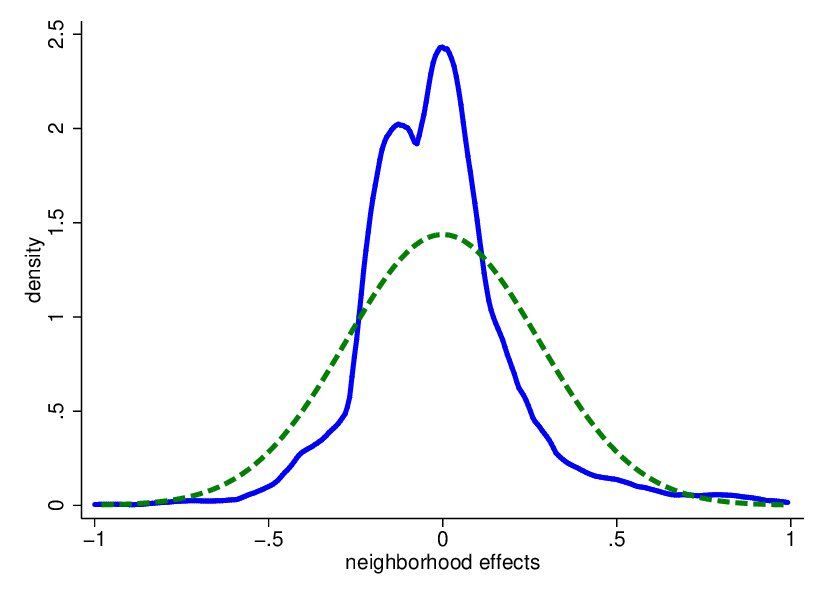}&	\includegraphics[width=80mm, height=60mm]{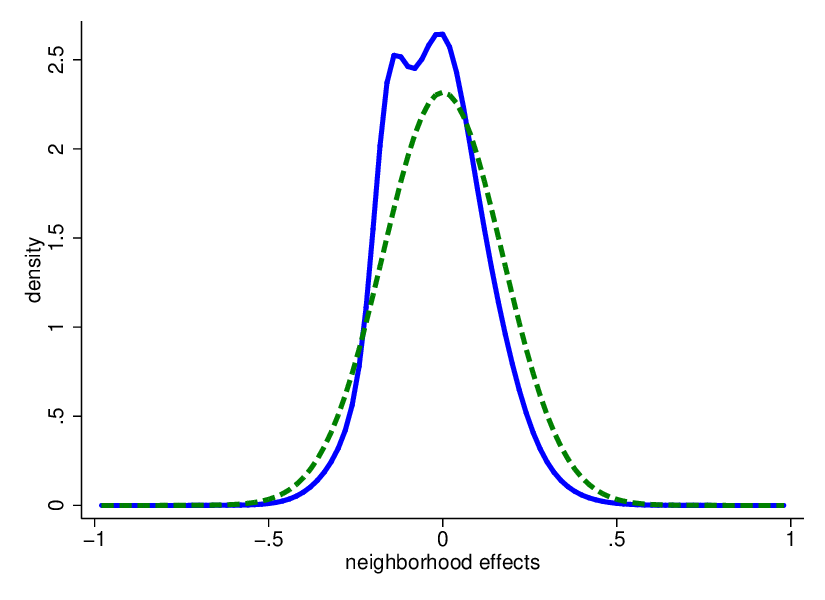}\\\end{tabular}
	\end{center}
	\par
	\textit{{\small Notes: In the left graphs, we show the distribution of fixed-effects estimates $\widehat{\mu}_c$ (solid) and its normal fit (dashed). In the right graphs, we show the posterior distribution of $\mu_c$ (solid) and the prior distribution (dashed). The distribution functions are shown in the top panel, the implied densities are shown in the bottom panel. Calculations are based on statistics available on the Equality of Opportunity website.}}
\end{figure}

To assess how likely it is that the posterior estimator approximates the shape of the distribution of true neighborhood effects, we next perform two different exercises, based on a simulation and on numerical calculations motivated by our theory. We start with a Monte Carlo simulation, where $\mu_c$, for $c=1,...,C_{\rm sim}$, are {log-normally} distributed with zero mean and variance $\widehat{s}_{\mu}^2$, and $\overline{\varepsilon}_c$ are normally distributed independent of $\mu_c$ with zero mean. We consider three scenarios for the noise variances $\widehat{s}_c^2$: the estimates from CH, one-third of those values, and one-tenth of those values. In this exercise we again weigh by population. We show the results for $C_{\rm sim}=100,000$ simulated neighborhoods. In the left graphs of Figure \ref{fig_neighb_sim} we see that, when the noise variances are the ones from the data, the posterior density is more skewed than the normal, yet the posterior shape is quite different from the true log-normal distribution of $\mu_c$. When reducing the noise variances in the middle and right graphs, the posterior distribution function and density estimates get closer to the log-normal ones. In the right graphs, where the shrinkage factor is .90 on average (as opposed to .62 in the data), the posterior distribution function and density approximate the highly non-normal shape of the true distribution of neighborhood effects very well.

\begin{figure}[h!]
	\caption{Simulated data with log-normal $\mu_c$\label{fig_neighb_sim}}
	\begin{center}
		\begin{tabular}{ccc}
			$100\%$ noise variances from data& 33\% noise variances & 10\% noise variances  \\
			\multicolumn{3}{c}{Distribution functions}\\
			\includegraphics[width=40mm, height=35mm]{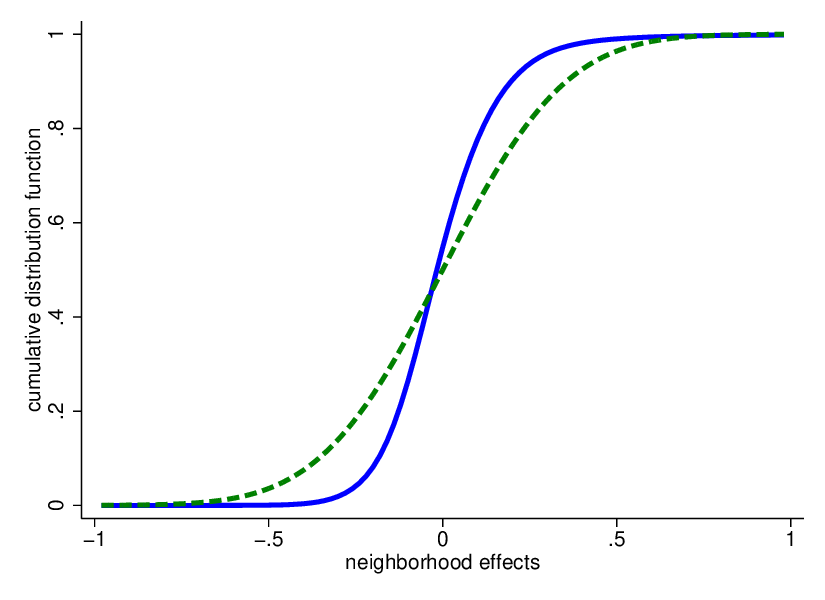}&	\includegraphics[width=40mm, height=35mm]{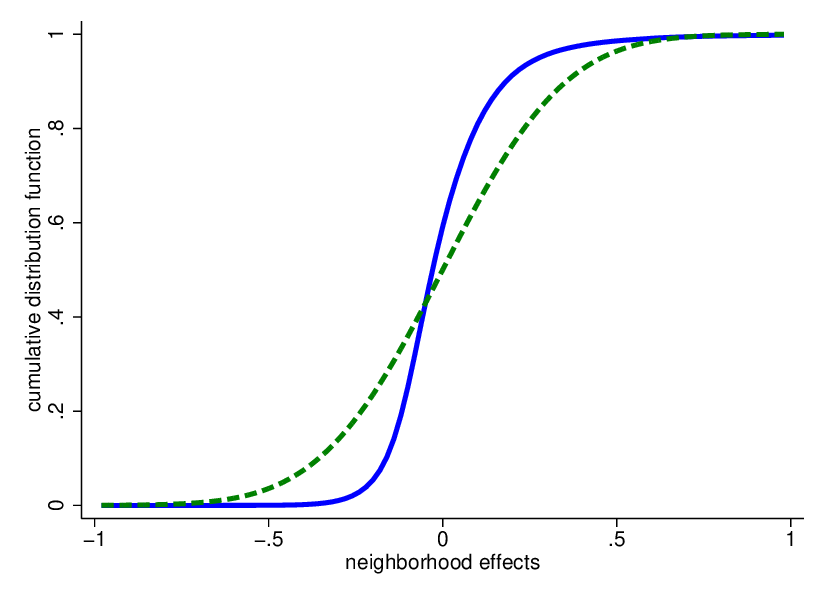}&	\includegraphics[width=40mm, height=35mm]{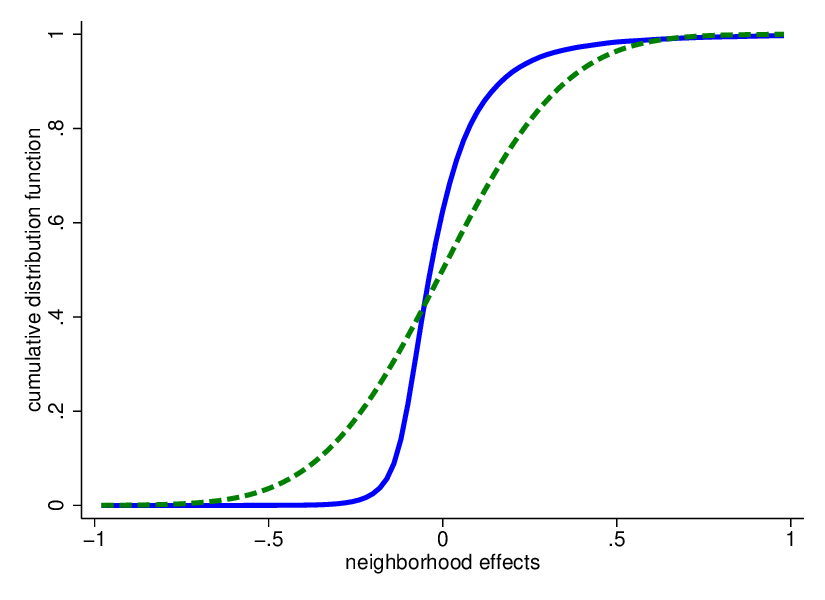}\\
				\multicolumn{3}{c}{Densities}\\
				\includegraphics[width=40mm, height=35mm]{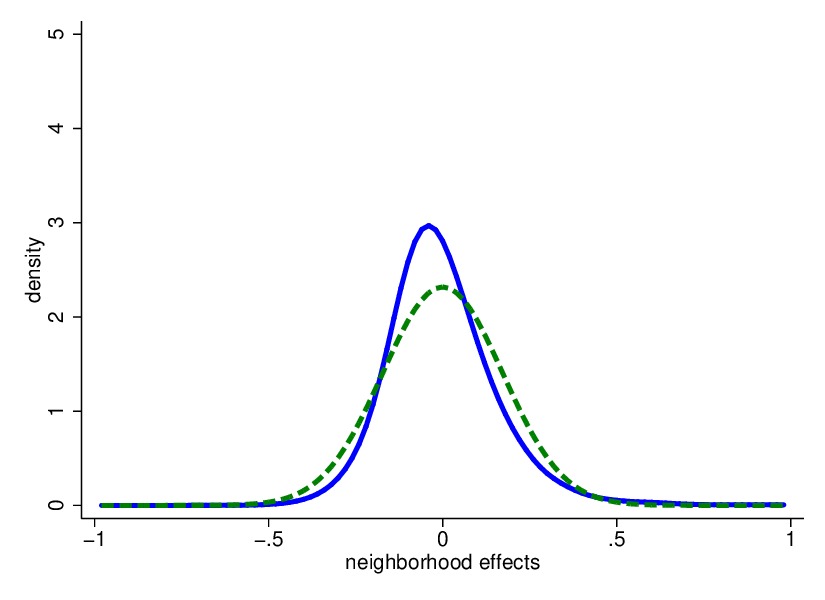}&	\includegraphics[width=40mm, height=35mm]{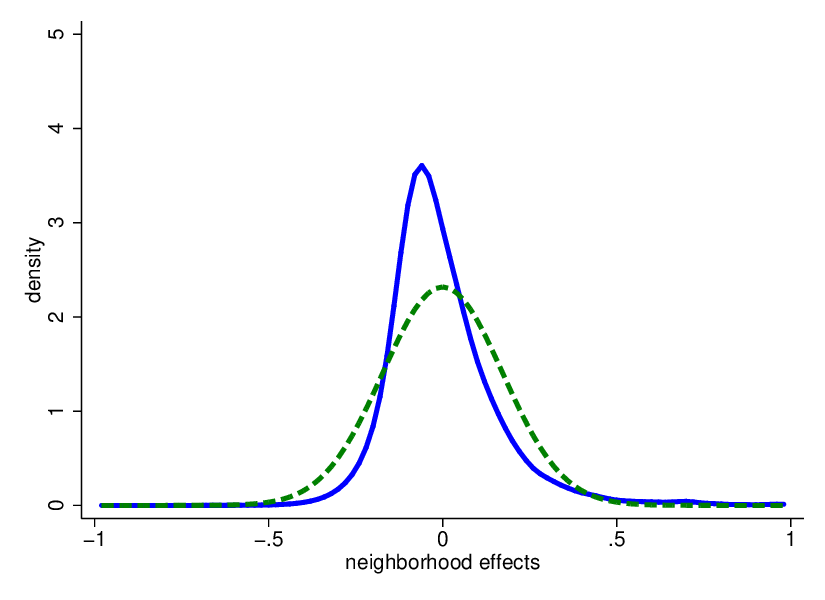}&	\includegraphics[width=40mm, height=35mm]{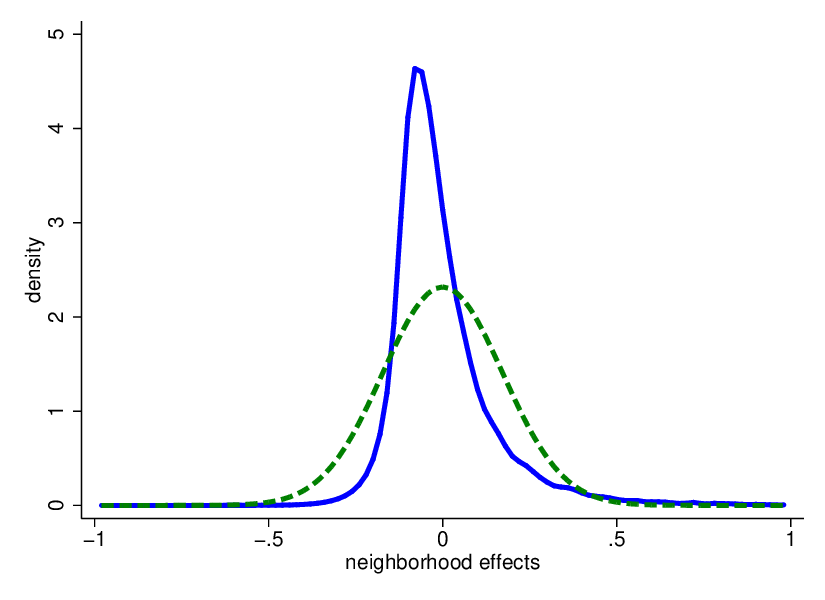}\\\end{tabular}
	\end{center}
	\par
	\textit{{\small Notes: Simulation with $\mu_c$ log-normal and $\overline{\varepsilon}_c$ normal. The posterior distribution is shown in solid, the prior distribution is shown in dashed. The distribution functions are shown in the top panel, the implied densities are shown in the bottom panel. The left graphs correspond to the noise variances $\widehat{s}_c^2$ of the data, the middle ones correspond to the noise variances divided by $3$, and the right graphs correspond to the noise variances divided by $10$.}}
\end{figure}

We next turn to our posterior informativeness measure, which is given by equation (\ref{eq_R2_FE}). Note the $R^2$ coefficient varies along the distribution. We find that the weighted average $R^2$ across values of $a$ is 28\%, where we weigh across cutoff values $a$ by the reference distribution for $\alpha$. This value is consistent with the message of Figure \ref{fig_neighb_sim}, since it suggests that, while the posterior conditioning informs the shape of the distribution of neighborhood effects, the signal-to-noise ratio is not high enough to be confident about the exact shape.

Lastly, we perform two additional exercises as robustness checks. Firstly, we incorporate the mean income $\overline{y}_c$ of permanent residents in county $c$ at the 25\% percentile as a covariate. CH rely on information on permanent residents' income to improve the accuracy of individual predictions. Here we use it to refine the reference distribution and to improve the estimation of the distribution of neighborhood effects. Specifically, our reference model for $\mu_c$ is then a correlated random-effects specification, where the mean depends on $\overline{y}_c$ linearly. Appendix Figure \ref{fig_neighb_FE_correl} shows small differences with our baseline estimates. Secondly, we re-do our main analysis at the county level, instead of the commuting zone level. In that case, the signal-to-noise ratio is lower, our posterior informativeness $R^2$ measure is 17\% on average, and Appendix Figure \ref{fig_neighb_FE_cty} shows that the normal prior and the posterior distributions are closer to each other than in the case of commuting zones.

\subsection{Income dynamics}

In this subsection, we consider the following permanent-transitory model of household log-income, 
$$
Y_{it}=\eta_{it}+\varepsilon_{it},
\qquad
\eta_{it}=\eta_{i,t-1}+V_{it},\qquad i=1,...,n,\quad t=1,...,T,
$$
where $\varepsilon_{it}$ and $V_{it}$ are independent at all lags and leads, and independent of $\eta_{i0}$. This process is commonly used as an input for life-cycle consumption/savings models. Researchers often estimate covariances in a first step using minimum distance, and then impose a normality assumption for further analysis. However, there is increasing evidence that income components are {not} normally distributed. Instead of using a more flexible model --- as has been done by Carlton and Hall (1978) and a large subsequent literature --- here we compute posterior average effects. The advantages of this approach are that no additional assumptions are needed, and that implementation is straightforward.

We focus on six recent waves of the PSID 1999-2009 (every other year), see Blundell \textit{et al.} (2016) for a description of the data. We use the same sample selection as in Arellano \textit{et al.} (2017), and work with a balanced panel of $n=792$ households over $T=6$ periods. $Y_{it}$ are residuals of log total pre-tax household labor earnings on a set of demographics, which include cohort interacted with education categories for both household members, race, state, and large-city dummies, a family size indicator, number of kids, a dummy for income recipient other than husband and wife, and a dummy for kids out of the household. Our aim is to estimate the distributions of $\eta_{it}$ and $\varepsilon_{it}$. To do so, we compare normal model-based estimates with posterior estimates, by plotting distribution functions as well as the implied densities. The model's structure is similar to that of the fixed-effects model (\ref{FE_mod}), and analytical expressions for posterior estimators are easy to derive.

\begin{figure}[h!]
	\caption{Distribution of income components\label{fig_psid}}
	\begin{center}
		\begin{tabular}{cc}
			Permanent component & Transitory component \\
			\multicolumn{2}{c}{Distribution functions}\\
			\includegraphics[width=80mm, height=60mm]{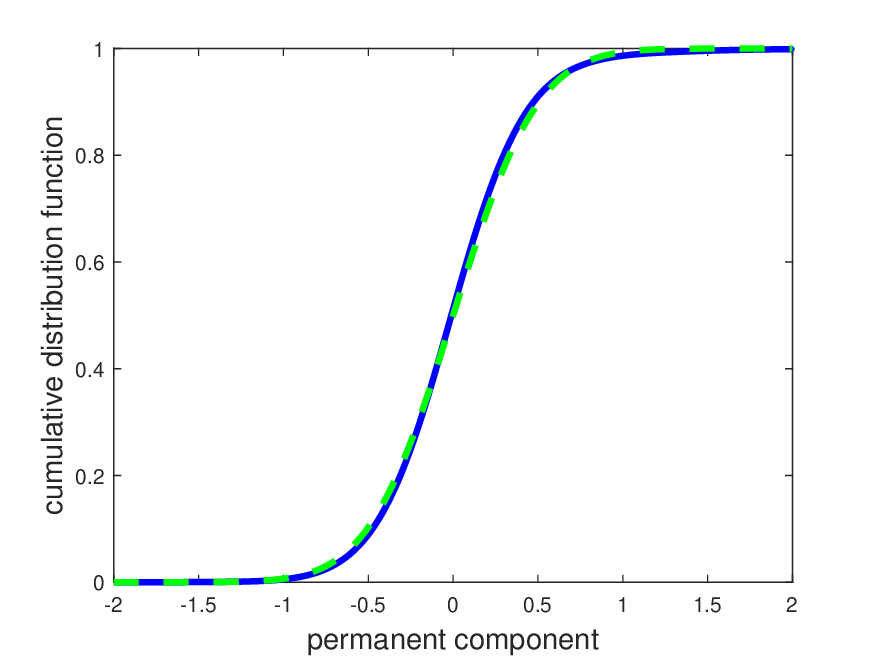}&\includegraphics[width=80mm, height=60mm]{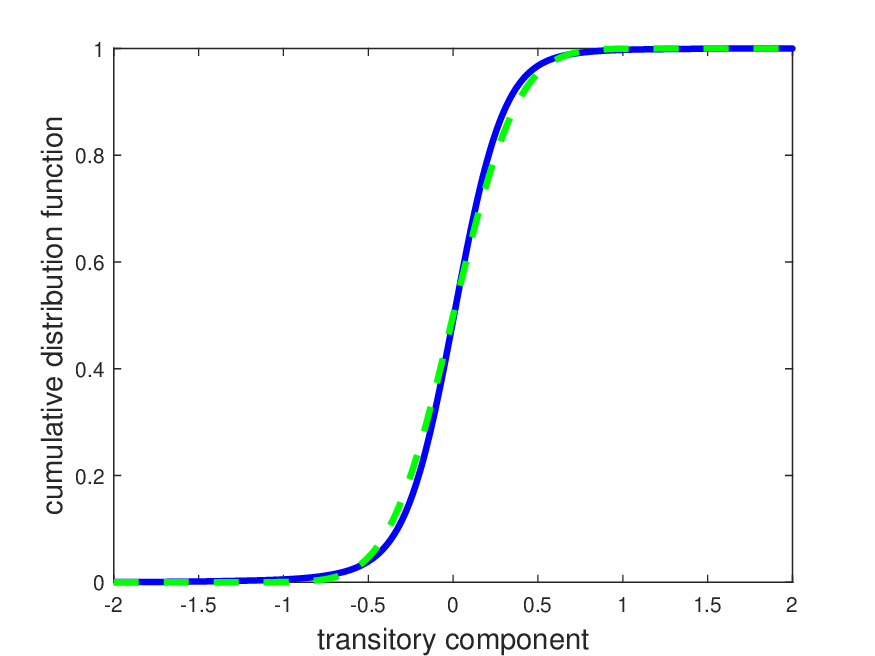} \\
				\multicolumn{2}{c}{Densities}\\
				\includegraphics[width=80mm, height=60mm]{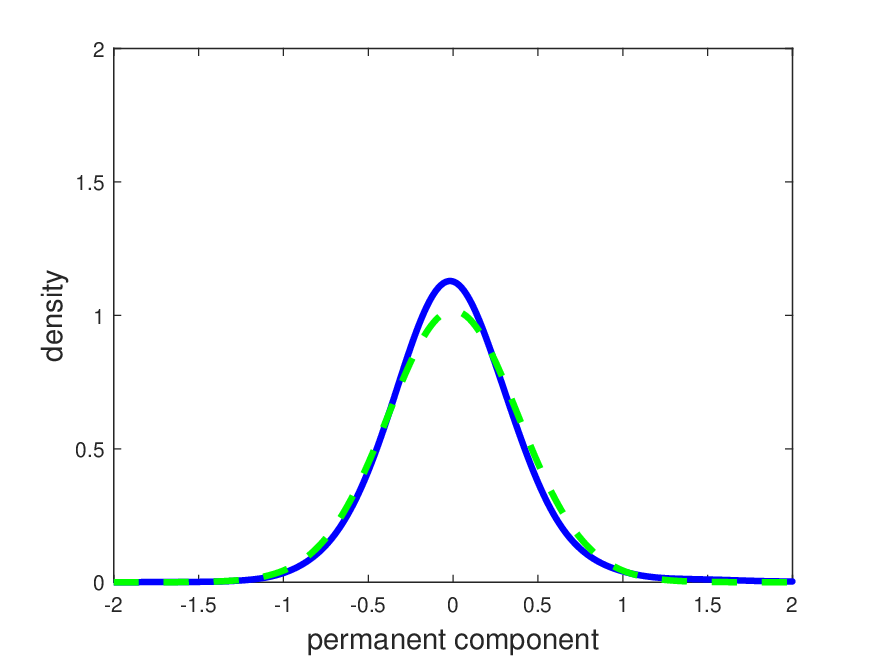}&\includegraphics[width=80mm, height=60mm]{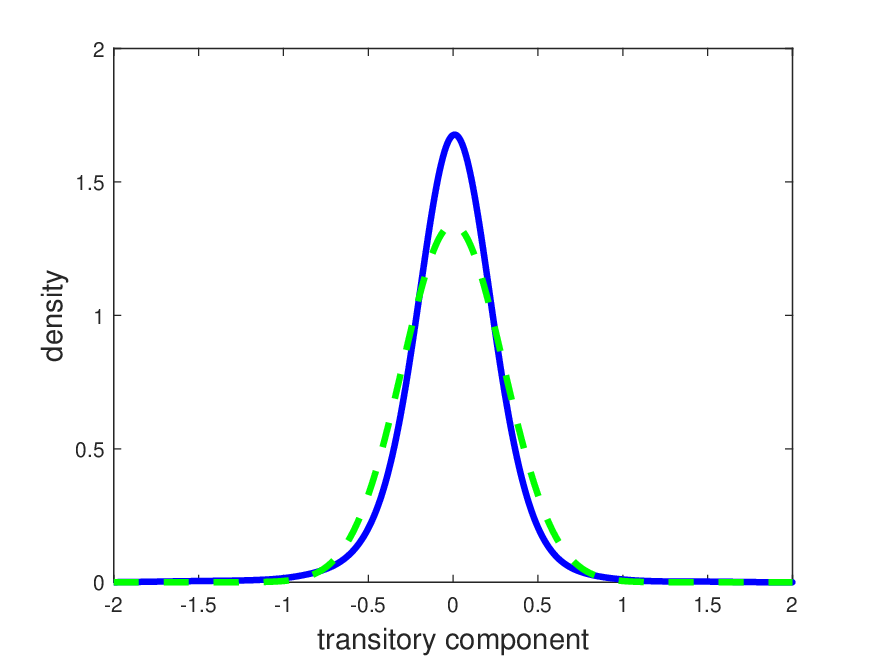} \\\end{tabular}
	\end{center}
	\par
	\textit{{\small Notes: The top panel shows PAE estimates of distribution functions (in solid), and model-based estimates (in dashed), and the bottom panel shows the associated density estimates. The left graphs correspond to the permanent income component $\eta_{it}$, the right graphs to the transitory income component $\varepsilon_{it}$. Sample from the PSID, 1999-2009.}}
\end{figure}

In the left graphs of Figure \ref{fig_psid}, we show the distribution of the permanent component $\eta_{it}$. In the right graphs, we show the distribution of the transitory component $\varepsilon_{it}$. We show PAE in solid, and model-based estimators in dashed. In the top panel we report estimates of distribution functions, and in the bottom panel we report the implied density estimates. The estimates show mild deviation from Gaussianity for the permanent component, and stronger evidence of non-Gaussianity for the transitory component. In particular, the latter shows excess kurtosis (i.e., ``peakedness'') relative to the normal.

Several papers have already documented the presence of excess kurtosis in income components, particularly in transitory innovations, using parametric or semi-parametric methods. The estimates in Figure \ref{fig_psid} share some qualitative similarities with recent findings in the literature. For example, the estimates of a flexible non-normal and non-linear model in Arellano \textit{et al.} (2017, Figure 3) are quite similar to the PAE estimates in Figure \ref{fig_psid} for permanent components. At the same time, their estimates of the distribution of transitory components show substantially more pronounced non-Gaussianity and excess kurtosis relative to PAE. This finding is in agreement with our posterior informativeness measure R$^2$, which is 12\% on average along the distribution for the permanent component, and 8\% on average for the transitory component. This degree of informativeness suggests that posterior estimates may suffer from substantial specification error when the reference distribution is misspecified.

Overall, these empirical illustrations give two examples where, starting from a normal prior, the posterior conditioning is informative about the true unknown distributions. In both settings, PAE are not normal. Yet, as indicated by the $R^2$ values we report, the signal-to-noise ratios are not high enough to be certain about the exact shapes of the distributions of interest, thus motivating further analyses using non-normal specifications. PAE should be useful in other environments where model (\ref{FE_mod}) and its extensions are widely used, for example in teacher value-added applications, where the signal-to-noise ratio is driven by the number of observations per teacher. Moreover, PAE are also applicable to other --- nonlinear --- econometric models, as we describe in the next section.

\section{Complements and extensions\label{sec_comp}}

In this section, we outline several complements and extensions that we analyze in detail in the appendix.

\subsection{PAE in other models\label{sec_other_ex}}

PAE are applicable to a variety of settings. In many econometric models, semi-parametric estimators --- i.e., robust to distributional assumptions on unobservables --- of $\beta$ parameters are available; see Powell (1994) for examples. In such models, PAE provide estimators of average effects that enjoy robustness properties when parametric assumptions are violated. In Appendix \ref{sec_other_ex2} we study static binary and ordered choice models, censored regression models, and panel data binary choice models. We also show how the White (1980) formula for robust standard errors in linear regression can be interpreted as a PAE.

\subsection{Confidence intervals and specification test}

Under correct specification of the reference model, it is easy to derive the asymptotic distributions of $\widehat{\delta}^{\rm M}$ and $\widehat{\delta}^{\rm P}$ using standard arguments. Moreover, under local misspecification, confidence intervals that account for both model uncertainty and sampling uncertainty can be constructed following Armstrong and Koles\'ar (2018) and Bonhomme and Weidner (2018). However, such confidence intervals require the researcher to set a value for the degree of misspecification $\epsilon$. In Appendix \ref{App_Ext}, we provide details on confidence intervals calculations. In addition, we explain how to construct a specification test of the reference model based on the difference $\widehat{\delta}^{\rm P}-\widehat{\delta}^{\rm M}$.

\subsection{Robustness in prediction}

In applications such as the fixed-effects model (\ref{FE_mod}) of teacher quality, researchers are often interested in {predicting} the quality $\alpha_i$ of teacher~$i$. Although our focus in this paper is on the estimation of population averages, it is interesting to see how different predictors perform under misspecification of the reference distribution. It is well-known that EB estimators minimize mean squared prediction error when the normal reference model is correctly specified. However, when normality fails, the best predictor is a different posterior mean, which does {not} generally coincide with the EB estimate based on a normal prior. Intuitively, conditioning on nonlinear functions of the data may improve prediction accuracy.

In Appendix \ref{App_Pred} we use our framework --- applied to worst-case mean squared prediction error instead of worst-case specification error of a sample average --- to provide results on the robustness of EB estimators in the presence of misspecification. We show that EB estimators have minimum worst-case mean squared prediction error, up to smaller-order terms, under local deviations from normality. In addition, we derive a fixed-$\epsilon$, non-local risk bound in the spirit of Theorem \ref{theo_global}.

\section{Conclusion\label{Sec_conclusion}}

Posterior averages are commonly used to predict individual parameters, such as teacher quality or neighborhood effects, and they play a central role in Bayesian and empirical Bayes approaches. In this paper, we have provided a frequentist justification for posterior conditioning when the goal of the researcher is to estimate a population average quantity. We have shown that posterior average effects (PAE) have minimum worst-case specification error under various forms of misspecification of parametric assumptions. PAE are simple to implement, and our analysis provides a rationale for reporting them in applications alongside other parametric and semi-parametric estimators, as well as a simple way to assess the informativeness of the posterior conditioning. As an example, Arnold \textit{et al.} (2020) recently reported PAE to document judge heterogeneity in the context of bail decisions. While we have used a linear fixed-effects model as a running example due to its popularity, there are other possible applications, some of which we discuss in the appendix. 

\vskip 1cm

\paragraph{Acknowledgments.} We would like to thank two anonymous referees, Manuel Arellano, Tim Armstrong, Raj Chetty, Tim Christensen, Nathan Hendren, Peter Hull, Max Kasy, Derek Neal, Jesse Shapiro, Xiaoxia Shi, Danny Yagan, and audiences at various places for comments. Bonhomme acknowledges support from the NSF, Grant SES-1658920. Weidner acknowledges support from the Economic and Social Research Council through the ESRC Centre for Microdata Methods and Practice grant RES-589-28-0001 and from the European Research Council grants ERC-2014-CoG-646917-ROMIA
	and ERC-2018-CoG-819086-PANEDA.

\baselineskip21pt

\clearpage

\appendix

\renewcommand{\theequation}{\thesection \arabic{equation}}

\renewcommand{\thelemma}{\thesection \arabic{lemma}}

\renewcommand{\theproposition}{\thesection \arabic{proposition}}

\renewcommand{\thecorollary}{\thesection \arabic{corollary}}

\renewcommand{\thetheorem}{\thesection \arabic{theorem}}

\renewcommand{\theassumption}{\thesection \arabic{assumption}}

\renewcommand{\thefigure}{\thesection \arabic{figure}}

\renewcommand{\thetable}{\thesection \arabic{table}}

\setcounter{equation}{0}
\setcounter{table}{0}
\setcounter{figure}{0}
\setcounter{assumption}{0}
\setcounter{proposition}{0}
\setcounter{lemma}{0}
\setcounter{corollary}{0}
\setcounter{theorem}{0}

{\small \baselineskip15pt }

\setcounter{page}{1}

\setcounter{section}{0}\renewcommand{\thesection}{S\arabic{section}}

\setcounter{figure}{0}\renewcommand{\thefigure}{S\arabic{figure}}

\setcounter{table}{0}\renewcommand{\thetable}{S\arabic{table}}

\setcounter{footnote}{0}\renewcommand{\thefootnote}{\arabic{footnote}}

\setcounter{assumption}{0}\renewcommand{\theassumption}{S\arabic{assumption}}

\setcounter{equation}{0}\renewcommand{\theequation}{S\arabic{equation}}

\setcounter{lemma}{0}\renewcommand{\thelemma}{S\arabic{lemma}}

\setcounter{proposition}{0}\renewcommand{\theproposition}{S\arabic{proposition}}

\setcounter{corollary}{0}\renewcommand{\thecorollary}{S\arabic{corollary}}

\setcounter{theorem}{0}\renewcommand{\thetheorem}{S\arabic{theorem}}

\begin{center}
	{\small {\LARGE APPENDIX } }
\end{center}

\section{Proofs of Lemma~\ref{lem_bias}, and Theorem~\ref{theo_bias} and \ref{theo_bias0}\label{App_Proofs}}

The following is an extended version of Lemma \ref{lem_bias} and Theorem~\ref{theo_bias} in the main text, which
also covers the case of unbounded functions $\gamma_{{\beta},{\sigma}_{*}}(y,x)$, $\delta_{\beta}(u,x)$
and $\psi_{\beta,\sigma_{*}}(y,x)$. In addition, we make explicit again the dependence on $\beta$ and $\sigma_*$,
which we suppressed in the main text.

\begin{lemma}\label{lem_bias_app}
        In addition to defining $\widetilde \psi(y,x) = \psi(y,x) - \mathbb{E}_{P(\beta,f_{\sigma_*})} \left[ \psi(Y,X)  \big| X=x \right]$, let $\widetilde \gamma(y,x) = \gamma(y,x) - \mathbb{E}_{P(\beta,f_{\sigma_*})} \left[ \gamma(Y,X)  \big|  X=x \right]$ and         $\widetilde  \delta(u,x) = \delta(u,x) - \mathbb{E}_{P(\beta,f_{\sigma_*})} \left[ \delta(U,X)  \big| X=x \right]$. Suppose that 
	$\phi(r) = \overline \phi(r)+\nu \, (r-1)^2$, with $\nu \geq 0$,
	and a function $\overline \phi(r)$ that is four times continuously differentiable 
	with $\overline \phi(1) =0$ and $\overline \phi''(r)>0$, for all $r \in (0,\infty)$.
	Assume
	$ \mathbb{E}_{P(\beta,f_{\sigma_*})}  \psi_{{\beta},{\sigma}_{*}}(Y,X) = 0$ and 
	$\mathbb{E}_{P(\beta,f_{\sigma_{*}})}\left[ \widetilde \psi_{\beta,\sigma_{*}}(Y,X) \, \widetilde \psi_{\beta,\sigma_{*}}(Y,X)'\right] >0$.
	Furthermore, assume that one of the following holds:
	\begin{itemize}
		\item[(i)] $\nu = 0$, and the functions
		$\left| \gamma_{{\beta},{\sigma}_{*}}(y,x) \right|$, $\left| \delta_\beta(u,x) \right| $ and $\left| \psi_{{\beta},{\sigma}_{*}}(y,x) \right|$ are bounded over the domain of $Y$, $U$, $X$.
		
		\item[(ii)] $\nu>0$, and 
		$\mathbb{E}_{P(\beta,f_{\sigma_*})} \left| \gamma_{{\beta},{\sigma}_{*}}(Y,X) - \delta_{\beta}(U,X) \right|^3 < \infty$,
		and  $\mathbb{E}_{P(\beta,f_{\sigma_*})} \left| \psi_{{\beta},{\sigma}_{*}}(Y,X)  \right|^3 < \infty$.
		
	\end{itemize}
	Then, as $\epsilon\rightarrow 0$ we have
	\begin{align*}
	b_{\epsilon}(\gamma)&=\left|\mathbb{E}_{P(\beta,f_{\sigma_*})}[\gamma_{{\beta},{\sigma}_{*}}(Y,X)]-\mathbb{E}_{f_{\sigma_*}}[\delta_{\beta}(U,X)]\right|\\
	&\quad
	 +\epsilon^{\frac{1}{2}}\left\{\frac{2}{\phi''(1)}
	 {\rm Var}_{P(\beta,f_{\sigma_{*}})}\left[ \widetilde \gamma_{{\beta},{\sigma}_{*}}(Y,X)- \widetilde \delta_{\beta}(U,X)-\lambda' \, \widetilde \psi_{\beta,\sigma_{*}}(Y,X) \right]\right\}^{\frac{1}{2}}+{\cal O}(\epsilon),
	\end{align*}
	where $$\lambda=\left\{\mathbb{E}_{P(\beta,f_{\sigma_{*}})}
	\left[\widetilde \psi_{\beta,\sigma_{*}}(Y,X)\, \widetilde \psi_{\beta,\sigma_{*}}(Y,X)'\right]\right\}^{-1}
	\mathbb{E}_{P(\beta,f_{\sigma_{*}})}\left[\left(\gamma_{{\beta},{\sigma}_{*}}(Y,X)-\delta_{\beta}(U,X)\right)\widetilde \psi_{\beta,\sigma_{*}}(Y,X)\right].$$
\end{lemma}

\begin{theorem}\label{theo_bias_app}
	Suppose that the conditions of Lemma~\ref{lem_bias_app} hold, and let 
	\begin{equation}\gamma_{\beta,\sigma_*}^{\rm P}(y,x)=\mathbb{E}_{p_{\beta,\sigma_*}}[\delta_{\beta}(U,X)\,|\, Y=y,X=x]. 
	\label{eq_gamma_P_app}
	\end{equation} 
	Then the following results hold as $\epsilon$ tends to zero.
	\begin{itemize}
	    \item[(i)] We have 
	$$b_{\epsilon}(\gamma_{\beta,\sigma_*}^{\rm P})\leq b_{\epsilon}(\gamma)+{\cal O}(\epsilon).$$
	
	     \item[(ii)] If we have $b_{\epsilon}(\gamma) =  b_{\epsilon}(\gamma_{\beta,\sigma_*}^{\rm P})+o(\epsilon^{1/2}),$
	      then there exist 
	      $\lambda \in  \mathbb{R}^{\dim \psi}$ and
	      a function $\omega : {\cal X} \rightarrow \mathbb{R}$   with  $\mathbb{E}_{f_X} [\omega(X)] = 0$ such that
	     $$
          \gamma_{{\beta},{\sigma}_{*}}(Y,X)  =     \gamma^{\rm P}_{{\beta},{\sigma}_{*}}(Y,X)   + \omega(X) + \lambda' \ \psi_{\beta,\sigma_{*}}(Y,X)     +  o_{P(\beta,f_{\sigma_{*}})}(1)   .
          $$
	\end{itemize}
\end{theorem}
\medskip
Notice that  Theorem~\ref{theo_bias_app} in the main text is a special case of part (i) of Theorem~\ref{theo_bias_app}.
Part (ii) of Theorem~\ref{theo_bias_app} is discussed in Subsection~\ref{subsec_discuss} of the main text. The proof
of Theorem~\ref{theo_bias_app}  provides explicit expressions for $\lambda$ and $\omega(X)$ that appear in part (ii),
namely  $\lambda$ is the same as in the last line of Lemma~\ref{lem_bias_app}, and 
$  \omega(x)  =  \mathbb{E}_{P(\beta,f_{\sigma_*})} \big[   \gamma_{{\beta},{\sigma}_{*}}(Y,X)-    \delta_{\beta}(U,X) 
    -\lambda' \,   \psi_{\beta,\sigma_{*}}(Y,X)   \big| X=x \big]
     -  \mathbb{E}_{P(\beta,f_{\sigma_*})} \big[   \gamma_{{\beta},{\sigma}_{*}}(Y,X)-    \delta_{\beta}(U,X) 
    -\lambda' \,   \psi_{\beta,\sigma_{*}}(Y,X)  \big]   $.

\subsection{Proof of Lemma~\ref{lem_bias_app} (containing Lemma \ref{lem_bias} as a special case)}

We first introduce some additional notation and establish some helpful intermediate results.
We write ${\cal B}$ and ${\cal S}$ for the set of possible values of the parameters $\beta$ and $\sigma$, respectively.
Lemma~\ref{lem_bias_app} is for given values $\beta \in {\cal B}$ and $\sigma_* \in {\cal S}$,
and given functions $\gamma_{\beta,\sigma_*}(y,x)$,
$\delta_{\beta}(u,x)$,
$\psi_{\beta,\sigma_*} (y,x)$,
and those values and functions are also taken as given in following two intermediate lemmas.    
Remember also that     $\Gamma_{\epsilon}$ depends on the function  $\phi : [0,\infty) \rightarrow \mathbb{R}  \cup \{\infty\}$,
which is assumed to be strictly convex in Lemma~\ref{lem_bias_app}.
We define the corresponding function $\rho : \mathbb{R} \rightarrow \mathbb{R}  \cup \{\infty\}$ by
\begin{align}
\rho(t)  :=  
\left\{ \begin{array}{ll}  \argmax_{r \geq 0} \, \left[ r\, t - \phi(r) \right]
& \text{if this ``argmax'' exists,} $$
\\[5pt]
\infty & \text{otherwise.}
\end{array}
\right.
\label{DefRho}
\end{align}
For $t = \phi'(r)$ we have $\rho(t) = r$, that is, for those values of $t$
the function $ \rho(t)$ is simply the inverse function of the first derivative $\phi'$.
For $t < \inf_{r > 0} \phi'(r)$ we have $\rho(t) = 0$,
and for $t> \sup_{r > 0} \phi'(r)$ the value of $\rho(t)$ is defined to be $\infty$.
The following lemma provides a characterization of the $\epsilon$-worst-case specification error $ b_{\epsilon}(\gamma)$
that was defined in \eqref{DefBias}.

\begin{lemma}
	\label{lemma:FOC}
	Let $\epsilon>0$.
	Assume that $\phi(r)$ is strictly convex with $\phi(1)=0$.
	Suppose that for $s \in \{-1,1\}$ and $x \in {\cal X}$ there exists $\lambda^{(1)}_{\beta,\sigma_*}(s,x) \in \mathbb{R}$, $\lambda^{(2)}_{\beta,\sigma_*}(s)>0$, $\lambda^{(3)}_{\beta,\sigma_*}(s) \in \mathbb{R}^{\dim \psi}$
	such that
	$$
	t_{\beta,\sigma_*}(u,x|s):=   \lambda^{(1)}_{\beta,\sigma_*}(s,x) + s \, \lambda^{(2)}_{\beta,\sigma_*}(s)   \left[\gamma_{{\beta},{\sigma}_{*}}(g_\beta(u,x),x)- \delta_{\beta}(u,x)   \right]
	+ \lambda^{(3) \, \prime}_{\beta,\sigma_*}(s) \, \psi_{\beta,\sigma_*}(g_\beta(u,x),x)) 
	$$
	satisfies
	\begin{align}
	\forall x \in {\cal X}: \; \;
	\mathbb{E}_{P(\beta,f_{\sigma_*})} \left\{  \rho\left[ t_{\beta,\sigma_*}(U,X|s)  \right]  \Big| \, X=x \right\}  &= 1 ,
	\nonumber  \\
	\mathbb{E}_{P(\beta,f_{\sigma_*})}   \,  \phi\left\{ \rho\left[ t_{\beta,\sigma_*}(U,X|s)  \right] \right\}  &= \epsilon ,
	\nonumber \\
	\mathbb{E}_{P(\beta,f_{\sigma_*})}   \, \left\{  \psi_{{\beta},{\sigma}_{*}}(Y,X)  \phantom{\Big|}  \rho\left[ t_{\beta,\sigma_*}(U,X|s)  \right] \right\}  &= 0 . 
	\label{RewriteConstraints}
	\end{align}
	Then the maximizer ($s=+1$) and minimizer  ($s=-1$) of
	 $\mathbb{E}_{P(\beta,f_{0})}   \left[\gamma_{{\beta},{\sigma}_{*}}(Y,X)- \delta_{\beta}(U,X) \right] $
	 over $f_0 \in \Gamma_\epsilon$ are given
	 by
	 $$
	       f^{(s)}_0(u|x) = f_{\sigma_*}(u|x) \, \rho\left[ t_{\beta,\sigma_*}(u,x|s)  \right] ,
	$$
	and for the worst-case absolute specification error we therefore have
	\begin{align*}
	b_{\epsilon}(\gamma)
	&= \max_{s \in \{-1,1\}}
	\left\{  s \; \mathbb{E}_{P(\beta,f_{\sigma_*})} \left[  \left[\gamma_{{\beta},{\sigma}_{*}}(Y,X)- \delta_{\beta}(U,X) \right] 
	\phantom{\Big|}
	\rho\left[ t_{\beta,\sigma_*}(U,X|s)  \right] \right]
	\right\}  .
	\end{align*}
	
\end{lemma}        
\noindent
The proof of Lemma~\ref{lemma:FOC} is given in Section~\ref{sec:ProofTechnical}.
Notice that for $\phi(r) = r [\log(r)-1]$, when $d(f_0,f_{\sigma_{*}}) $ 
  is the Kullback-Leibler divergence, we have $\rho(t) = \exp(t)$, and the worst
  case densities $ f^{(s)}_0(u|x) $ in Lemma~\ref{lemma:FOC} are exponentially tilted 
  versions of the reference density $ f_{\sigma_*}(u|x) $. Lemma~\ref{lemma:FOC}
  shows that, more generally, the required ``tilting function'' is given by $\rho(t)$.

We  impose $\phi(1)=0$ throughout the paper to guarantee that  $d(f_0,f_{\sigma_{*}}) \geq 0$
(by an application of Jensen's inequality). 
In addition, we now impose the normalization $\phi'(1)=0$.
This is without loss of generality, because we can always redefine 
$\phi(r) \mapsto \phi(r) -  (r-1) \, \phi'(1) $,
which has no effect on $ d(f_0,f_{\sigma_{*}})$
and guarantees $\phi'(1)=0$ for the redefined function.

The goal of the following lemma is to establish Taylor expansions of
$\rho(t)$ and $ \phi( \rho(t) )$  around $t=0$ of the form
\begin{align}
\rho(t) &= 1 + \frac{t} { \phi''(1)}+ t^2   \,  R_1(t)  ,
&
\phi( \rho(t) ) &=     \frac {t^2} {2 \, \phi''(1)} + t^3 \, R_2(t) ,
\label{ExpansionRhoPhi}
\end{align}
where the remainder terms are defined by
\begin{align*}
R_1(t) &:=   \left\{ \begin{array}{l@{\qquad}l}  t^{-2} \left[  \rho(t) - 1 - t  /  \phi''(1) \right] & \text{if $t \neq 0$,}
\\[5pt]
- \phi'''(1) / \{ 2 \, [\phi''(1)]^3 \}  & \text{if $t = 0$,}    
\end{array} \right. 
\\[10pt]  
R_2(t) &:=    \left\{ \begin{array}{l@{\quad}l}  t^{-3} \left[ \phi( \rho(t) )  -  t^2 / \{2  \phi''(1)\} \right] & \text{if $t \neq 0$,}
\\[5pt]
- \phi'''(1) / \{ 3 \, [\phi''(1)]^3 \}   & \text{if $t = 0$.}    
\end{array} \right. 
\end{align*}
Notice that the expansions \eqref{ExpansionRhoPhi} are trivially true by definition of $R_1(t) $ and $R_2(t)$,
but the  following lemma provides bounds on $R_1(t) $ and $R_2(t)$, which are useful for the proof of Lemma~\ref{lem_bias_app}
afterwards.

\begin{lemma}
	\label{lemma:ConvexDual}
	For all $r \geq  0$
	let
	$\phi(r) = \overline \phi(r)+\nu \, (r-1)^2$, for $\nu \geq 0$,
	and a function $\overline \phi : [0,\infty) \rightarrow \mathbb{R} \cup \{\infty\}$ that is four times continuously differentiable 
	with $\overline \phi(1)= \overline \phi'(1)=0$ and $\overline \phi''(r)>0$, for all $r \in (0,\infty)$.
	The lemma has two parts:
	\begin{itemize}
		\item[(i)]
		Assume in addition that $\nu=0$.
		Then, there exist constants $c_1>0$, $c_2>0$ and $\eta>0$ such that
		for all $t \in [-\eta,\eta]$ we have
		\begin{align}
		\left| R_1(t)   \right| \, \leq \, c_1  , 
		\quad \quad 
		\text{and}
		\quad \quad 
				\left|  R_2(t)   \right| \, \leq  \, c_2 ,
		\label{BoundRho}   
		\end{align}
		and the functions $R_1(t)$ and $R_2(t)$ are continuous  within $[-\eta,\eta]$.
		
		\item[(ii)] 
		Assume in addition that $\nu>0$.  Then, there  exist constants $c_1>0$ and $c_2>0$
		such that the two inequalities in \eqref{BoundRho} hold for all $t \in \mathbb{R}$,
		and  the functions $R_1(t)$ and $R_2(t)$ are  everywhere continuous.
		
	\end{itemize}
	
\end{lemma}
\noindent
The proof of Lemma~\ref{lemma:ConvexDual} is given in Section~\ref{sec:ProofTechnical}. 

\paragraph{Comment:} Part (i) and part (ii) of Lemma~\ref{lemma:ConvexDual} give the same approximations
of $\rho(t)$ and $ \phi( \rho(t) )$, but the difference is that in part (i) 
the result only holds locally in a neighborhood of $t=0$, while in part (ii)
the inequalities are established globally for all $t \in \mathbb{R}$. Notice that the result
of part (ii) cannot hold under the assumptions of part (i) only, because 
$\rho(t)$ is equal to infinity for
all $t > t_{\sup}$, where $t_{\sup} = \sup_{r \in (0,\infty)} \phi'(r)$ can be finite. The regularization  $\phi(r) = \overline \phi(r)+\nu \, (r-1)^2$,
with $\nu>0$,
guarantees that $\rho(t)$ is finite and well-defined for all $t \in \mathbb{R}$.
This property of the regularized $\phi(r) $ is key whenever the moment functions $\gamma$, $\delta$, $\psi$
are unbounded (i.e., for case (ii) of the assumptions of Lemma~\ref{lem_bias_app}).

\bigskip

Using the intermediate  Lemmas~\ref{lemma:FOC} and \ref{lemma:ConvexDual} we can now show  
Lemma~\ref{lem_bias_app}, which contains Lemma \ref{lem_bias} as a special case.
In the following proofs we again drop the arguments $\beta$ and $\sigma_*$ everywhere for ease notation,
	and we write $\mathbb{E}_{*} $ and ${\rm Var}_{*} $ for expectations and variances under the reference density $P(\beta,f_{\sigma_*})$.
	We also continue to use the normalization $\phi'(1)=0$, which is without loss of generality, as explained above.

\begin{proof}[\bf Proof of Lemma~\ref{lem_bias_app}]
	\underline{\# Additional notation and definitions:}
	Let $\lambda \in \mathbb{R}^{\dim \psi}$ be as defined in the statement of the lemma, and furthermore define
	\begin{align*}
	\kappa &=   \left\{  \frac  
	{{\rm Var}_*\left[ \widetilde \gamma(Y,X)- \widetilde \delta(U,X)-\lambda' \, \widetilde \psi(Y,X) \right]} {2 \, \phi''(1)}    \right\}^{1/2}.
	\end{align*}
	For $s \in \{-1,+1\}$ and $\epsilon>0$, let
	\begin{align*}
	t(u,x|s)  &=   \lambda^{(1)}(s,x) + s \, \lambda^{(2)}(s)   \left[\gamma(g(u,x),x)- \delta(u,x)   \right] 
	+ \lambda^{(3) \, \prime} (s) \, \psi(g(u,x),x)  ,
	\end{align*}
	with
	\begin{align*}
	\lambda^{(1)}(s,x)
	&=    - \epsilon^{1/2}    \, s \,   \kappa^{-1} \,  \mathbb{E}_*\left[ \gamma(Y,X)-\delta(U,X) - \lambda' \, \psi(Y,X) \, \big| \, X=x \right]  
	\\ & \qquad \qquad 
	+  \epsilon\,  \left\{ \lambda^{(1)}_{\rm rem}(s,x) 
	   -    s \, \lambda^{(2)}_{\rm rem}(s)  \, \mathbb{E}_*\left[ \gamma(Y,X) - \delta(U,X)  - \lambda' \, \psi(Y,X) \, \big| \, X=x \right]  \right\},
	\\
	\lambda^{(2)}(s) &= 
	\epsilon^{1/2}   
	\kappa^{-1} 
	+  \epsilon \, \lambda^{(2)}_{\rm rem}(s) ,
	\\ 
	\lambda^{(3)}(s)
	&=    - \epsilon^{1/2}    \, s \,  \kappa^{-1}  \, \lambda   
	+  \epsilon \, \left[ \lambda^{(3)}_{\rm rem}(s)  -  s\,  \lambda^{(2)}_{\rm rem}(s) \, \lambda \right].
	\end{align*}
	Here, we are explicit about the leading order terms (of order $\epsilon^{1/2}$),
	but the higher order terms (of order $\epsilon$)
	contain the coefficients
	$\lambda^{(1)}_{\rm rem}(s) \in \mathbb{R}$, $\lambda^{(2)}_{\rm rem}(s) \in \mathbb{R}$,
	and $\lambda^{(3)}_{\rm rem}(s) \in \mathbb{R}^{\dim \psi}$, which will only be specified in \eqref{ConditionsRemainder} below.
	We can rewrite 
	\begin{align}
	t(u,x|s) =      \epsilon^{1/2}    t_{(0)}(u,x|s)   +         \epsilon   \,    t_{\rm rem}(u,x|s),
	\end{align}
	with
	\begin{align*}
	t_{(0)}(u,x|s) 
	&=  s\,  \kappa^{-1}  \, \left[ \widetilde \gamma(g(u,x),x)- \widetilde \delta(u,x)   - \lambda' \, \widetilde \psi(g(u,x),x) \right]  ,
	\\
	t_{\rm rem}(u,x|s) 
	&=  \lambda_{\rm rem}^{(1)}(s,x) + \lambda_{\rm rem}^{(2)}(s)  \,  \kappa \,  t_{(0)}(u,x|s)
	+ \lambda_{\rm rem}^{(3) \, \prime}(s) \, \psi(g(u,x),x) .
	\end{align*}
	Here, $ t(u,x|s)$, $ \lambda^{(1)}(s,x)$,  $\lambda^{(2)}(s)$, etc, also depend on $\epsilon$, but we do not make this dependence explicit 
	in our notation. Our goal is to apply Lemma~\ref{lemma:FOC} with $ t_{\beta,\sigma_*}(u,x|s)$ in the lemma equal to $ t(u,x|s)$
	as defined here. However, in order to apply that lemma we need to satisfy the conditions \eqref{RewriteConstraints}, which in current notation read
	\begin{align}
	\mathbb{E}_*   \,  \left\{ \rho[t(U,X|s)] \big| X=x  \right\} &= 1 ,
	&
	\mathbb{E}_*   \,  \phi\left\{ \rho\left[t(U,X|s)  \right] \right\}  &= \epsilon ,
	&
	\mathbb{E}_*   \, \left\{  \psi(Y,X)  \phantom{\Big|}  \rho\left[ t(U,X|s)  \right] \right\}  &= 0 .
	\label{RewriteConstraintsNEW}     
	\end{align}
	The definition of $t(u,x|s)$ above is already designed to satisfy \eqref{RewriteConstraintsNEW} to leading order in $\epsilon$,
	but we still need to find $  \lambda_{\rm rem}^{(1)}(s,x) $, $ \lambda_{\rm rem}^{(2)}(s) $, $\lambda_{\rm rem}^{(3)}(s)$ 
	such that \eqref{RewriteConstraintsNEW} holds exactly.
	Plugging the expansions \eqref{ExpansionRhoPhi}
	into \eqref{RewriteConstraintsNEW}, using the definition of $t(u,x|s)$,
	as well as  $\mathbb{E}_*\left[ t_{(0)}(U,X|s) \big| X=x  \right]  = 0$, $  \mathbb{E}_* \left\{ [ t_{(0)}(U,X|s) ]^2 \right\} = 2 \, \phi''(1)$, and $\mathbb{E}_*   \psi(Y,X)  \, t_{(0)}(U,X|s) = 0$, we obtain
	\begin{align*}
	\mathbb{E}_*  \left\{   \frac{\epsilon \, t_{\rm rem}(U,X|s)} { \phi''(1)}       + \left[ t(U,X|s)  \right]^2  R_1 \left[ t(U,X|s)  \right] \Bigg| X=x \right\} &= 0 ,
	\nonumber  \\
	\mathbb{E}_*   \left\{   
	\frac{2 \, \epsilon^{3/2} \,    t_{\rm rem}(U,X|s)  \, t_{(0)}(U,X|s) 
		+ \epsilon^2 \, [  t_{\rm rem}(U,X|s)]^2} {2 \phi''(1)}     + \left[ t(U,X|s)  \right]^3  R_2 \left[ t(U,X|s)  \right] \right\}  &= 0,
	\nonumber \\
	\mathbb{E}_*  \left\{   \frac{\epsilon \,  \psi(Y,X)  \, t_{\rm rem}(U,X|s)} { \phi''(1)}       + 
	\psi(Y,X)  \, \left[ t(U,X|s)  \right]^2  R_1 \left[ t(U,X|s)  \right] \right\} &= 0 . 
	\end{align*}
	Those conditions can be rewritten as follows
	\begin{align}
	\lambda_{\rm rem}^{(1)}(s,x) &= - \phi''(1) \, \mathbb{E}_*\left\{  \left[   t_{(0)}(U,X|s)   +         \epsilon^{1/2}       t_{\rm rem}(U,X|s) \right]^2   \,  R_1\left[ t(U,X|s) \right]  \Bigg| X=x  \right\} ,
	\nonumber  \\  
	\lambda_{\rm rem}^{(2)}(s)  
	&=  - \frac{1} {2 \, \kappa}  \, \mathbb{E}_*\left\{  
	\left[   t_{(0)}(U,X|s)   +         \epsilon^{1/2}       t_{\rm rem}(U,X|s) \right]^3  \,   R_2\left[ t(U,X|s) \right]
	+  \frac{ \epsilon^{1/2} \, [  t_{\rm rem}(U,X|s)]^2} {2 \phi''(1)} 
	\right\} ,
	\nonumber \\  
	\lambda_{\rm rem}^{(3)}(s) 
	&= - \phi''(1) \, 
	\left\{ \mathbb{E}_* \left[  \psi(Y,X) \,  \psi(Y,X)' \right] \right\}^{-1}
	\nonumber  \\ & \qquad    \qquad    \quad 
	\times   \mathbb{E}_*\left\{   \psi(Y,X) \,  \left[   t_{(0)}(U,X|s)   +         \epsilon^{1/2}       t_{\rm rem}(U,X|s) \right]^2     R_1\left[ t(U,X|s) \right]   \right\} .
	\label{ConditionsRemainder}   
	\end{align}
	Thus, as $\epsilon \rightarrow 0$ we have
	\begin{align}
	\lambda_{\rm rem}^{(1)}(s,x) &= - 2 [\phi''(1)]^2 \, R_1(0) \,
	 \frac{{\rm Var}_*\left[ \widetilde \gamma(Y,X)- \widetilde \delta(U,X)-\lambda' \, \widetilde \psi(Y,X)  \, \Big| \, X=x\right]}
	  {{\rm Var}_*\left[ \widetilde \gamma(Y,X)- \widetilde \delta(U,X)-\lambda' \, \widetilde \psi(Y,X) \right]} 
	  + {\cal O}(\epsilon^{1/2}) ,
	\nonumber  \\  
	\lambda_{\rm rem}^{(2)}(s)  
	&=  - \frac{1} {2 \, \kappa}  \, \mathbb{E}_*    \left[   t_{(0)}(U,X|s)    \right]^3  \,    R_2(0) + {\cal O}(\epsilon^{1/2})
	,
	\nonumber \\  
	\lambda_{\rm rem}^{(3)}(s) 
	&= - \phi''(1) \, 
	\left\{ \mathbb{E}_* \left[  \psi(Y,X) \,  \psi(Y,X)' \right] \right\}^{-1}
	\mathbb{E}_*\left\{   \psi(Y,X) \,  \left[   t_{(0)}(U,X|s)   \right]^2     \right\} \,   R_1(0)  + {\cal O}(\epsilon^{1/2}) .
	\label{ConditionsRemainderApprox}
	\end{align}
	Notice that $  \lambda_{\rm rem}^{(1)}(s,x) $, $ \lambda_{\rm rem}^{(2)}(s) $, $\lambda_{\rm rem}^{(3)}(s)$ also appear implicitly
	on the right-hand sides of the equations \eqref{ConditionsRemainder}, because $ t_{\rm rem}(u,x|s) $ depends on those parameters,
	and  \eqref{ConditionsRemainder} is therefore a system of equations for  $  \lambda_{\rm rem}^{(1)}(s,x) $, $ \lambda_{\rm rem}^{(2)}(s) $, $\lambda_{\rm rem}^{(3)}(s)$. Our assumptions guarantee that the system \eqref{ConditionsRemainder} has a solution
	for sufficiently small $\epsilon$, as will be explained below for the two different cases distinguished in the lemma.

	\bigskip
	\noindent
	\underline{\# Proof for case (i):}
	The assumptions for this case guarantee that   $t(u,x|s)$ is uniformly 
	bounded over $u$ and $x$.
	Part (i) of Lemma~\ref{lemma:ConvexDual} guarantees existence of $c_1>0$, $c_2>0$, $\eta>0$ such that
	for all $t \in [-\eta,\eta]$ we have  $\left| R_1(t) \right|  \leq c_1$    and $\left| R_2(t) \right|  \leq c_2$.
	For sufficiently small $\epsilon$ 
	we have $ t(u,x|s) \in [-\eta,\eta]$ for all $u$ and $x$, implying that as $\epsilon \rightarrow 0$
	there exists a solution of  \eqref{ConditionsRemainder} that satisfies \eqref{ConditionsRemainderApprox}, which in particular implies
	\begin{align}
	\sup_{x \in {\cal X}} \left| \lambda^{(1)}(s,x) \right| &=  {\cal O}(1),
	&
	\lambda^{(2)}(s) &=  {\cal O}(1),
	&
	\lambda^{(3)}(s) &=  {\cal O}(1) ,
	\label{BoundeLambda123}  
	\end{align}
	and by construction the conditions \eqref{RewriteConstraintsNEW}  are satisfied for that solution.
	Thus, for sufficiently small $\epsilon$ the $  t(u,x|s) $ defined above 
	satisfies the conditions  of Lemma~\ref{lemma:FOC}. Applying that lemma we thus obtain that,
	for sufficiently small $\epsilon$, we have
	\begin{align*}
	b_{\epsilon}(\gamma)
	&= \max_{s \in \{-1,1\}}
	\left\{  s \; \mathbb{E}_* \left[  \left[\gamma(Y,X)- \delta(U,X) \right] 
	\phantom{\Big|}
	\rho\left[ t(U,X|s)  \right] \right]
	\right\}  .
	\end{align*}
	Again applying the expansion for $\rho(t)$ in \eqref{ExpansionRhoPhi},
	and part (i) of Lemma~\ref{lemma:ConvexDual} we thus obtain that    
	\begin{align}
	& b_{\epsilon}(\gamma) 
	= 
	\max_{s \in \{-1,1\}}
	\left\{ s \,
	\mathbb{E}_{*} \left[\gamma(Y,X)- \delta(U,X) \right]   
	\right\}
	\nonumber  \\ & \qquad \qquad \qquad \qquad  
	+    
	\epsilon^{1/2}  \left\{ \frac{2} {\phi''(1)}
	{\rm Var}_*\left[ \widetilde  \gamma(Y,X)- \widetilde \delta(U,X)-\lambda' \, \widetilde \psi(Y,X) \right]
	\right\}^{1/2}
	+ {\cal O}(\epsilon)   
	\nonumber  \\
	&= 
	\left|  \mathbb{E}_{*} \left[\gamma(Y,X)- \delta(U,X) \right]  \right|
	+ \epsilon^{1/2}  \left\{ \frac{2} {\phi''(1)}
	{\rm Var}_*\left[ \widetilde \gamma(Y,X)- \widetilde \delta(U,X)-\lambda' \, \widetilde \psi(Y,X) \right]
	\right\}^{1/2}
	+ {\cal O}(\epsilon) .
	\label{FinalExpansionBias}       
	\end{align}
	This is what we wanted to show.

	\bigskip
	\noindent
	\underline{\# Proof for case (ii):}
	In this case, according to part (ii) of Lemma~\ref{lemma:ConvexDual}
	the functions $R_1(t)$ and $R_2(t)$ are continuous and bounded over all $t \in \mathbb{R}$.
	In addition, we have assumed that
	$\mathbb{E}_* \left| \gamma (Y,X) - \delta (U,X) \right|^3 < \infty$,
	and  $\mathbb{E}_*  \left| \psi (Y,X)  \right|^3 < \infty$, which guarantees that
	all of the  expectations in \eqref{ConditionsRemainder} are finite.
	We therefore again conclude that for small $\epsilon$
	the equations \eqref{ConditionsRemainder} have a solution such that
	\eqref{BoundeLambda123} holds. The remainder of the proof is equivalent to the proof of part (i),
	that is, we again apply Lemma~\ref{lemma:FOC} and Lemma~\ref{lemma:ConvexDual} to obtain
	\eqref{FinalExpansionBias}.
\end{proof}

\subsection{Proof of Theorem~\ref{theo_bias_app} (containing Theorem~\ref{theo_bias} as a special case)}
  
   \underline{\# Part (i):} 
      We first want to show that $b_{\epsilon}(\gamma_{\beta,\sigma_*}^{\rm P})\leq b_{\epsilon}(\gamma)+{\cal O}(\epsilon)$.
	By applying Lemma~\ref{lem_bias_app} to both  $\gamma_{\beta,\sigma_*}(y,x)$
	and $\gamma_{\beta,\sigma_*}^{\rm P}(y,x)=\mathbb{E}_{p_{\beta,\sigma_*}}[\delta_{\beta}(U,X)\,|\, Y=y,X=x]$
	we obtain, as $\epsilon\rightarrow 0$,
	\begin{align}
	b_{\epsilon}(\gamma)&=\left|\mathbb{E}_{P(\beta,f_{\sigma_*})}[\gamma_{{\beta},{\sigma}_{*}}(Y,X)]-\mathbb{E}_{f_{\sigma_*}}[\delta_{\beta}(U,X)]\right|
	\nonumber \\
	&\quad 
	+  \epsilon^{\frac{1}{2}}\left\{\frac{2}{\phi''(1)} 
	\mathbbm{E}_{P(\beta,f_{\sigma_{*}})}\left[\left( \widetilde  \gamma_{{\beta},{\sigma}_{*}}(Y,X)- \widetilde  \delta_{\beta}(U,X)
	-\lambda' \,  \widetilde  \psi_{\beta,\sigma_{*}}(Y,X)\right)^2\right]\right\}^{\frac{1}{2}}+{\cal O}(\epsilon),
	\nonumber \\
	b_{\epsilon}(\gamma^{\rm P})&= \epsilon^{\frac{1}{2}}\left\{\frac{2}{\phi''(1)}\mathbb{E}_{P(\beta,f_{\sigma_{*}})}\left[\left(\gamma^{\rm P}_{{\beta},{\sigma}_{*}}(Y,X)-\delta_{\beta}(U,X) \right)^2\right]\right\}^{\frac{1}{2}}+{\cal O}(\epsilon) ,
	   \label{BiasExpansionsProof}
	\end{align}
	where $$\lambda=\left\{\mathbb{E}_{P(\beta,f_{\sigma_{*}})}\left[ \widetilde \psi_{\beta,\sigma_{*}}(Y,X) \,  \widetilde  \psi_{\beta,\sigma_{*}}(Y,X)'\right]\right\}^{-1}\mathbb{E}_{P(\beta,f_{\sigma_{*}})}\left[\left(\gamma_{{\beta},{\sigma}_{*}}(Y,X)-\delta_{\beta}(U,X)\right) \, \widetilde    \psi_{\beta,\sigma_{*}}(Y,X)\right],
	$$
	Here, to simplify $b_{\epsilon}(\gamma^{\rm P})$
	we used that by the law of iterated expectations	
	we have that
	$\mathbb{E}_{P(\beta,f_{\sigma_*})}[\gamma^{\rm P}_{{\beta},{\sigma}_{*}}(Y,X)]-\mathbb{E}_{f_{\sigma_*}}[\delta_{\beta}(U,X)] = 0$
	(that is, the first term in $b_{\epsilon}(\gamma)$ is not present in $b_{\epsilon}(\gamma^{\rm P})$)
	and also $\mathbb{E}_{P(\beta,f_{\sigma_{*}})}\left[\left(\gamma^{\rm P}_{{\beta},{\sigma}_{*}}(Y,X)-\delta_{\beta}(U,X)\right) \widetilde \psi_{\beta,\sigma_{*}}(Y,X)\right] = 0$ (that is, the vector $\lambda$ is equal to zero for $\gamma^{\rm P}$).	We also use that under the reference model
	 $ \widetilde \gamma_{{\beta},{\sigma}_{*}}(Y,X)- \widetilde  \delta_{\beta}(U,X)-\lambda' \, \widetilde \psi_{\beta,\sigma_{*}}(Y,X)  $
	has zero mean, implying that its variance equals its second moment.
	
	For any $\gamma_{{\beta},{\sigma}_{*}}(y,x)$
	with 
	$\mathbb{E}_{P(\beta,f_{\sigma_*})}[\gamma_{{\beta},{\sigma}_{*}}(Y,X)]-\mathbb{E}_{f_{\sigma_*}}[\delta_{\beta}(U,X)] \neq 0$
	we have 
	$b_{\epsilon}(\gamma^{\rm P}) \leq b_{\epsilon}(\gamma) $ for sufficiently small $\epsilon$, 
	and the statement of the theorem thus holds in that case.
	In the following we therefore consider the case that 
	$\mathbb{E}_{P(\beta,f_{\sigma_*})}[\gamma_{{\beta},{\sigma}_{*}}(Y,X)]-\mathbb{E}_{f_{\sigma_*}}[\delta_{\beta}(U,X)] = 0$.
	The expression for $b_{\epsilon}(\gamma)$ then simplifies to
	\begin{align*}
	b_{\epsilon}(\gamma)&= \epsilon^{\frac{1}{2}}\left\{\frac{2}{\phi''(1)} 
	\mathbbm{E}_{P(\beta,f_{\sigma_{*}})}\left[\left( \widetilde  \gamma_{{\beta},{\sigma}_{*}}(Y,X)- \widetilde  \delta_{\beta}(U,X)
	-\lambda' \,  \widetilde  \psi_{\beta,\sigma_{*}}(Y,X)\right)^2\right]\right\}^{\frac{1}{2}}+{\cal O}(\epsilon).
	\end{align*}
      Again applying  the law of iterated expectations we find that
	\begin{align*}
	& \mathbb{E}_{P(\beta,f_{\sigma_{*}})}
	\left[ \widetilde  \gamma_{{\beta},{\sigma}_{*}}(Y,X)-  \widetilde  \delta_{\beta}(U,X)-\lambda' \,  \widetilde  \psi_{\beta,\sigma_{*}}(Y,X)\right]
	\left[\gamma^{\rm P}_{{\beta},{\sigma}_{*}}(Y,X)-\delta_{\beta}(U,X) \right]
	\\ & \qquad
	=           \mathbb{E}_{P(\beta,f_{\sigma_{*}})}
	\left[ -\delta_{\beta}(U,X) \right]
	\left[\gamma^{\rm P}_{{\beta},{\sigma}_{*}}(Y,X)-\delta_{\beta}(U,X) \right] 
	\\ & \qquad
	=           \mathbb{E}_{P(\beta,f_{\sigma_{*}})}
	\left[ \gamma^{\rm P}_{{\beta},{\sigma}_{*}}(Y,X) -\delta_{\beta}(U,X) \right]
	\left[\gamma^{\rm P}_{{\beta},{\sigma}_{*}}(Y,X)-\delta_{\beta}(U,X) \right] 
	\\ & \qquad
	=           \mathbb{E}_{P(\beta,f_{\sigma_{*}})}
	\left[\gamma^{\rm P}_{{\beta},{\sigma}_{*}}(Y,X)-\delta_{\beta}(U,X) \right]^2 .
	\end{align*}
	Using this we obtain
	\begin{align}
	&  \mathbb{E}_{P(\beta,f_{\sigma_{*}})}
	\left\{ \left[  \widetilde  \gamma_{{\beta},{\sigma}_{*}}(Y,X)-  \widetilde  \delta_{\beta}(U,X)-\lambda' \,  \widetilde  \psi_{\beta,\sigma_{*}}(Y,X)\right]
	- \left[ \gamma^{\rm P}_{{\beta},{\sigma}_{*}}(Y,X)-\delta_{\beta}(U,X) \right]
	\right\}^2
	\nonumber \\
	&=  \mathbb{E}_{P(\beta,f_{\sigma_{*}})} \left[ \widetilde  \gamma_{{\beta},{\sigma}_{*}}(Y,X)- \widetilde  \delta_{\beta}(U,X)-\lambda' \,  \widetilde \psi_{\beta,\sigma_{*}}(Y,X)\right]^2
	+  \mathbb{E}_{P(\beta,f_{\sigma_{*}})} \left[\gamma^{\rm P}_{{\beta},{\sigma}_{*}}(Y,X)-\delta_{\beta}(U,X) \right]^2
	\nonumber \\
	& \quad -2 \, \mathbb{E}_{P(\beta,f_{\sigma_{*}})} 
	\left[ \widetilde  \gamma_{{\beta},{\sigma}_{*}}(Y,X)- \widetilde  \delta_{\beta}(U,X)-\lambda' \,  \widetilde  \psi_{\beta,\sigma_{*}}(Y,X)\right]
	\left[\gamma^{\rm P}_{{\beta},{\sigma}_{*}}(Y,X)-\delta_{\beta}(U,X) \right]
	\nonumber \\
	&= 
	\mathbb{E}_{P(\beta,f_{\sigma_{*}})} \left[ \widetilde  \gamma_{{\beta},{\sigma}_{*}}(Y,X)- \widetilde  \delta_{\beta}(U,X)
	-\lambda' \,  \widetilde  \psi_{\beta,\sigma_{*}}(Y,X)\right]^2
	-  \mathbb{E}_{P(\beta,f_{\sigma_{*}})} \left[\gamma^{\rm P}_{{\beta},{\sigma}_{*}}(Y,X)-\delta_{\beta}(U,X) \right]^2 .
	\label{MainConclusionLocalProof}
	\end{align}
	Since $ \mathbb{E}_{P(\beta,f_{\sigma_{*}})}
	\left\{ \left[  \widetilde  \gamma_{{\beta},{\sigma}_{*}}(Y,X)-  \widetilde  \delta_{\beta}(U,X)-\lambda' \,  \widetilde  \psi_{\beta,\sigma_{*}}(Y,X)\right]
	- \left[ \gamma^{\rm P}_{{\beta},{\sigma}_{*}}(Y,X)-\delta_{\beta}(U,X) \right]
	\right\}^2 \geq 0$ we thus conclude that
	\begin{align*}
	\mathbb{E}_{P(\beta,f_{\sigma_{*}})} \left[ \widetilde  \gamma_{{\beta},{\sigma}_{*}}(Y,X)-  \widetilde  \delta_{\beta}(U,X)-\lambda' \,  \widetilde \psi_{\beta,\sigma_{*}}(Y,X)\right]^2
	\geq  \mathbb{E}_{P(\beta,f_{\sigma_{*}})} \left[\gamma^{\rm P}_{{\beta},{\sigma}_{*}}(Y,X)-\delta_{\beta}(U,X) \right]^2 ,
	\end{align*}
        and therefore we obtain that
	$$b_{\epsilon}(\gamma_{\beta,\sigma_*}^{\rm P})\leq b_{\epsilon}(\gamma)+{\cal O}(\epsilon).$$
	This is the first statement of the theorem. This concludes the proof of part (i) of  Theorem~\ref{theo_bias_app}, of which Theorem~\ref{theo_bias} in the main text is a special case.

\medskip

   \underline{\# Part (ii):} 
   Next, let $\gamma_{\beta,\sigma_*}(y,x)$ be such that
   \begin{align}    
       b_{\epsilon}(\gamma) = b_{\epsilon}(\gamma_{\beta,\sigma_*}^{\rm P})+ o(\epsilon^{1/2}).
       \label{eq:ConditionTheoremMainII}
   \end{align}    
   Then, the specification error expansions in \eqref{BiasExpansionsProof} are still valid, and using those we conclude that we must have
   \begin{align}
      \mathbb{E}_{P(\beta,f_{\sigma_*})}[\gamma_{{\beta},{\sigma}_{*}}(Y,X) - \delta_{\beta}(U,X)] 
        &= o(1) ,
        \label{FirstImplicationBias}
   \end{align}
   because otherwise that term dominates all other terms in \eqref{eq:ConditionTheoremMainII}.
    We also conclude that we must have
   \begin{align*}
    &   \mathbbm{E}_{P(\beta,f_{\sigma_{*}})}\left[\left( \widetilde  \gamma_{{\beta},{\sigma}_{*}}(Y,X)- \widetilde  \delta_{\beta}(U,X)
	-\lambda' \,  \widetilde  \psi_{\beta,\sigma_{*}}(Y,X)\right)^2\right]
     \\
     & \qquad  \qquad  \qquad
   \leq  \mathbb{E}_{P(\beta,f_{\sigma_{*}})}\left[\left(\gamma^{\rm P}_{{\beta},{\sigma}_{*}}(Y,X)-\delta_{\beta}(U,X) \right)^2\right]
      + 	o(1) 
   \end{align*}
   for \eqref{eq:ConditionTheoremMainII} to hold.
   Furthermore,
   the  calculation in \eqref{MainConclusionLocalProof} is still valid here, and the inequality in the last display can therefore 
   equivalently be rewritten as
   \begin{align*}
          \mathbb{E}_{P(\beta,f_{\sigma_{*}})}
	\left\{ \left[  \widetilde  \gamma_{{\beta},{\sigma}_{*}}(Y,X)-  \widetilde  \delta_{\beta}(U,X)-\lambda' \,  \widetilde  \psi_{\beta,\sigma_{*}}(Y,X)\right]
	- \left[ \gamma^{\rm P}_{{\beta},{\sigma}_{*}}(Y,X)-\delta_{\beta}(U,X) \right]
	\right\}^2 
	=  o(1) ,
   \end{align*}
   where we write $=$ instead of $\leq$, because the left hand side expression is non-negative.   
   Applying Markov's inequality we thus find that
   \begin{align*}
         \widetilde  \gamma_{{\beta},{\sigma}_{*}}(Y,X)-  \widetilde  \delta_{\beta}(U,X)-\lambda' \,  \widetilde  \psi_{\beta,\sigma_{*}}(Y,X)
	&=   \gamma^{\rm P}_{{\beta},{\sigma}_{*}}(Y,X)-\delta_{\beta}(U,X)  +   o_{P(\beta,f_{\sigma_{*}})}(1)  .
   \end{align*}   
   Defining 
   \begin{align*}
    \omega(x) &:=  \mathbb{E}_{P(\beta,f_{\sigma_*})} \left[   \gamma_{{\beta},{\sigma}_{*}}(Y,X)-    \delta_{\beta}(U,X) 
    -\lambda' \,   \psi_{\beta,\sigma_{*}}(Y,X)   \big| X=x \right] 
    \\
     & \qquad  \qquad \qquad
    -  \mathbb{E}_{P(\beta,f_{\sigma_*})} \left[   \gamma_{{\beta},{\sigma}_{*}}(Y,X)-    \delta_{\beta}(U,X) 
    -\lambda' \,   \psi_{\beta,\sigma_{*}}(Y,X)  \right], 
    \end{align*}
    we therefore obtain   
   \begin{align*}
          \gamma_{{\beta},{\sigma}_{*}}(Y,X) &=   \gamma^{\rm P}_{{\beta},{\sigma}_{*}}(Y,X)   + \omega(X) + \lambda' \ \psi_{\beta,\sigma_{*}}(Y,X)
    \\ & \quad
        + \mathbb{E}_{P(\beta,f_{\sigma_*})} \left[   \gamma_{{\beta},{\sigma}_{*}}(Y,X)-    \delta_{\beta}(U,X) 
    - \lambda' \,   \psi_{\beta,\sigma_{*}}(Y,X)  \right] 
       + o_{P(\beta,f_{\sigma_{*}})}(1)
      \\
       &=    \gamma^{\rm P}_{{\beta},{\sigma}_{*}}(Y,X)   + \omega(X) + \lambda' \ \psi_{\beta,\sigma_{*}}(Y,X)     +   o_{P(\beta,f_{\sigma_{*}})}(1)   ,
   \end{align*}
   where in the last step we have used  \eqref{FirstImplicationBias} and $\mathbb{E}_{P(\beta,f_{\sigma_{*}})} [\psi_{\beta,\sigma_{*}}(Y,X)]=0$.
   Finally, notice that by construction we have
   \begin{align*}
        \mathbb{E}_{f_X} [\omega(X)] = 0 .
   \end{align*}

\subsection{Proof of Theorem~\ref{theo_bias0}}

 We are going to show Theorem \ref{theo_bias0}, which we restate here.

\begin{theorem*}
	Assume that $ \mathbb{E}_{P(\beta,f_{\sigma_*})}  \psi_{{\beta},{\sigma}_{*}}(Y,X) = 0$,  $\phi(t)=\frac 1 2 (t-1)^2$,
	and that $\gamma_{{\beta},{\sigma}_{*}}(Y,X)$ and $\delta_{{\beta}}(U,X)$   have finite second moments under the reference model.
	Then, for $0< \epsilon \leq \overline \epsilon$, we have
	$$b_{\epsilon}(\gamma_{\beta,\sigma_*}^{\rm P})\leq b_{\epsilon}(\gamma) ,$$
	with $\gamma_{\beta,\sigma_*}^{\rm P}(y,x)$ given by \eqref{eq_gamma_P_app} and $\overline{\epsilon}$ given by (\ref{MaxEpsilon}). 
\end{theorem*}

\begin{proof}[\bf Proof of Theorem~\ref{theo_bias0}]

In the following proof, we again
omit the arguments $\beta$, $\sigma_*$. By defining $Q(\gamma,f_0):= \mathbb{E}_{P(\beta,f_0)}[\gamma_{{\beta},{\sigma_*}}(Y,X)]-\mathbb{E}_{f_0}[\delta_{\beta}(U,X)] $
we can rewrite \eqref{DefBias} as
\begin{align*}
b_{\epsilon}(\gamma)=\sup_{f_0\in\Gamma_{\epsilon}}\, \left| Q(\gamma,f_0) \right|.
\end{align*}
Using that $\phi(t)=\frac 1 2 (t-1)^2$, we find that
the two worst-case distributions that maximize ($s=+$) and minimize ($s=-$) the function $Q(\gamma^{\rm P},f_0)$
over $f_0\in\Gamma_{\epsilon}$ are given by
\begin{align}
   f^{(s)}_0(u|x) = f_{\sigma_*}(u|x) \left[ 1+   s \,   (2\epsilon)^{1/2}   \,
	\frac{  \gamma^{\rm P}(g(u,x),x)-   \delta(u,x)   }  {\left\{ {\rm Var}_*\left[ \gamma^{\rm P}(Y,X)-   \delta(U,X)  \right] \right\}^{1/2}} \right] .
   \label{WorstCaseF0}	
\end{align}    
Notice that our condition $ \epsilon \leq \overline \epsilon$
guarantees that $ f^{(s)}_0(u|x) \geq 0$ for all $u$, $x$, $s$.

We refer to Lemma~\ref{lemma:FOC} above for a more general derivation of the worst-case distribution \eqref{WorstCaseF0}.
The results of Lemma~\ref{lemma:FOC} simplifies here,
because we assume that $\phi(t)=\frac 1 2 (t-1)^2$ in the current theorem, which implies that the function $\rho(t)$
   defined in \eqref{DefRho} is now given by $\rho(t)=\max(0,1+t) $, and we are also only interested in the worst-case
   distribution $  f^{(s)}_0(u|x) $ at $\gamma=\gamma^{\rm P}$.
Using \eqref{WorstCaseF0}  we find that
\begin{align}
    b_{\epsilon}(\gamma^{\rm P}) &=  Q\left(\gamma, f^{(+)}_0 r\right)  = - \, Q\left(\gamma, f^{(-)}_0 r\right)
    =  \left\{ 2 \, \epsilon\,  {\rm Var}_*\left[ \gamma^{\rm P}(Y,X)-   \delta(U,X)  \right] \right\}^{1/2} .
    \label{ResultBepsilon}
\end{align}
For a given $\gamma(y,x)$ we set
$$
     s_\gamma \, :=  \left\{ \begin{array}{ll}
         +   & \text{if }    \mathbb{E}_* [ \gamma(Y,X)-\delta(U,X)]  \geq 0, \\
         -   & \text{otherwise.}
     \end{array} \right. 
$$
Using that  $\mathbb{E}_*[\gamma^{\rm P}(Y,X) - \delta(U,X)\,|\, Y=y,X=x] = 0$,
we then find that
\begin{align}
    \mathbb{E}_* 
	\left\{  \gamma(Y,X)  
	\left[  \gamma^{\rm P}(Y,X)  -   \delta(U,X)  \right] \right\}
  =  \mathbb{E}_* 
	\left\{ \gamma^{\rm P}(Y,X)  
	\left[  \gamma^{\rm P}(Y,X)  -   \delta(U,X)  \right] \right\}	= 0.
   \label{ZeroMeanGamma}	
\end{align}
We  now  calculate
\begin{align*}
  b_{\epsilon}( \gamma) 
 &= \sup_{f_0\in\Gamma_{\epsilon}}\, \left| Q(\gamma,f_0) \right|  
 \\
 &\geq  \left| Q\left( \gamma,f^{(s_\gamma)}_0 \right) \right| 
\\ 
    &  =  \left| \mathbb{E}_*   [ \gamma(Y,X)-\delta(U,X)]   +
    s_\gamma \,   (2\epsilon)^{1/2}   \,
	\frac{\mathbb{E}_* 
	\left\{ \left[  \gamma(Y,X)  -   \delta(U,X)  \right] 
	\left[  \gamma^{\rm P}(Y,X)  -   \delta(U,X)  \right] \right\}  }  {\left\{ {\rm Var}_*\left[ \gamma^{\rm P}(Y,X)-   \delta(U,X)  \right] \right\}^{1/2}}
      \right| 
\\ 
    &  =  \left| \mathbb{E}_*   [ \gamma(Y,X)-\delta(U,X)]   +
    s_\gamma \,   (2\epsilon)^{1/2}   \,
	\frac{\mathbb{E}_* 
	\left\{ \left[ \gamma^{\rm P}(Y,X)  -   \delta(U,X)  \right] 
	\left[  \gamma^{\rm P}(Y,X)  -   \delta(U,X)  \right] \right\}  }  {\left\{ {\rm Var}_*\left[ \gamma^{\rm P}(Y,X)-   \delta(U,X)  \right] \right\}^{1/2}}
      \right| 
 \\
    &=      \left| \mathbb{E}_*   [ \gamma(Y,X)-\delta(U,X)]   +
    s_\gamma \,   b_{\epsilon}(\gamma^{\rm P}) 
      \right| 
  \\
    &=        \left| \mathbb{E}_*  [ \gamma(Y,X)-\delta(U,X)]  \right|
    + b_{\epsilon}(\gamma^{\rm P})  
  \\
    &\geq     b_{\epsilon}(\gamma^{\rm P})  ,
\end{align*}
where the first step is the definition of $b_{\epsilon}( \gamma) $,
the second step is a property of the supremum, 
the third step uses \eqref{WorstCaseF0},
the fourth step uses \eqref{ZeroMeanGamma}, 
the fifth step uses \eqref{ResultBepsilon}, and the sixth step uses that $ \mathbb{E}_*  [ \gamma(Y,X)-\delta(U,X)] $
and $ s_\gamma$ have the same sign. The result in the last display is exactly the statement of the theorem.

\end{proof}
   
   \subsection{Proof of Theorem \ref{theo_global}\label{App_global}}
   
   We are going to show Theorem \ref{theo_global}, which we restate here.

   \begin{theorem*} %
   	Let $\gamma_{\beta,\sigma_*}^{\rm P}$ as in (\ref{eq_gamma_P_app}).
   	Then, for all $\epsilon>0$,
   	$$ b_{\epsilon}(\gamma_{\beta,\sigma_*}^{\rm P})\leq 2 \,  \limfunc{inf}_{\gamma}\, b_{\epsilon}(\gamma_{\beta,\sigma_*}).$$
   \end{theorem*}
   \noindent
   The following lemma is useful for the proof of this theorem  (Theorem \ref{theo_global} in the main text).
   
   \begin{lemma}
   	\label{lemma:SupremumBoundConvex}
   	Let $\epsilon \geq 0$, $\beta \in {\cal B}$, $\sigma_* \in {\cal S}$,
   	and let $\zeta : {\cal U} \times {\cal X} \rightarrow \mathbb{R}$.
   	Then we have
   	$$
   	\sup_{f_0\in\Gamma_{\epsilon}} 
   	\left|
   	\mathbb{E}_{P(\beta,f_0)} \left\{
   	\mathbb{E}_{p_{\beta,\sigma_*}} \left[ \zeta(U,X) \,|\, Y,X \right] 
   	\right\} \right|
   	\leq 
   	\sup_{f_0\in\Gamma_{\epsilon}}  \left| \mathbb{E}_{P(\beta,f_0)} \left[ \zeta(U,X) \right] \right| .
   	$$
   \end{lemma}
   \noindent
   The proof of this lemma is given in Section~\ref{sec:ProofTechnical}. Notice that both Theorem~\ref{theo_global}
   and Lemma~\ref{lemma:SupremumBoundConvex} require that $\phi(r)$ is convex with $\phi(1)=0$,
   but they do not require $\phi''(1)>0$. For example, $\phi(r) = |r-1|/2$ is allowed here, which gives the
   total variation distance for $d(f_0,f_{\sigma_{*}})$.

   \begin{proof}[\bf Proof of Theorem~\ref{theo_global}]
   	By definition we have 
   	\begin{align*}
   	b_{\epsilon}(\gamma)&=\sup_{f_0\in\Gamma_{\epsilon}}\,
   	\left|\mathbb{E}_{P(\beta,f_0)}[\gamma_{{\beta},{\sigma}_{*}}(Y,X) - \delta_{\beta}(U,X)]\right|,
   	\\
   	b_{\epsilon}(\gamma^{\rm P})&=\sup_{f_0\in\Gamma_{\epsilon}}\,
   	\left|\mathbb{E}_{P(\beta,f_0)}[\gamma^{\rm P}_{{\beta},{\sigma}_{*}}(Y,X) - \delta_{\beta}(U,X)]\right|.
   	\end{align*} 
   	By writing 
   	$\gamma^{\rm P}_{{\beta},{\sigma}_{*}}(Y,X) - \delta_{\beta}(U,X) 
   	= \gamma_{{\beta},{\sigma}_{*}}(Y,X) - \delta_{\beta}(U,X)
   	- \left[ \gamma_{{\beta},{\sigma}_{*}}(Y,X) -  \gamma^{\rm P}_{{\beta},{\sigma}_{*}}(Y,X) \right] $
   	we obtain
   	\begin{align*}
   	b_{\epsilon}(\gamma^{\rm P})&=\sup_{f_0\in\Gamma_{\epsilon}}\,
   	\left|
   	\mathbb{E}_{P(\beta,f_0)} \left[  \gamma_{{\beta},{\sigma}_{*}}(Y,X) - \delta_{\beta}(U,X)  \right]
   	- \mathbb{E}_{P(\beta,f_0)} \left[  \gamma_{{\beta},{\sigma}_{*}}(Y,X) - \gamma^{\rm P}_{{\beta},{\sigma}_{*}}(Y,X) \right]
   	\right|
   	\\
   	&\leq  b_{\epsilon}(\gamma)
   	+
   	\sup_{f_0\in\Gamma_{\epsilon}} 
   	\left|\mathbb{E}_{P(\beta,f_0)} \left[  \gamma_{{\beta},{\sigma}_{*}}(Y,X) - \gamma^{\rm P}_{{\beta},{\sigma}_{*}}(Y,X) \right] \right| 
   	\\
   	&=     b_{\epsilon}(\gamma)
   	+
   	\sup_{f_0\in\Gamma_{\epsilon}} 
   	\left|
   	\mathbb{E}_{P(\beta,f_0)} \left\{
   	\mathbb{E}_{p_{\beta,\sigma_*}} \left[ \gamma_{{\beta},{\sigma}_{*}}(g_\beta(U,X),X) - \delta_{\beta}(U,X)\,|\, Y,X \right] 
   	\right\} \right|
   	\\
   	&\leq
   	b_{\epsilon}(\gamma)
   	+ \sup_{f_0\in\Gamma_{\epsilon}}\,
   	\left|\mathbb{E}_{P(\beta,f_0)}[\gamma_{{\beta},{\sigma}_{*}}(Y,X) - \delta_{\beta}(U,X)]\right|
   	= 2 \, b_{\epsilon}(\gamma),
   	\end{align*}  
   	where in the second-to-last step we have used Lemma~\ref{lemma:SupremumBoundConvex}
   	with $ \zeta(u,x) = \gamma_{{\beta},{\sigma}_{*}}(g_\beta(u,x),x) - \delta_{\beta}(u,x)$.
   	We have thus shown that  $   b_{\epsilon}(\gamma^{\rm P}) \leq 2 \, b_{\epsilon}(\gamma)$
   	holds for any function $\gamma_{\beta,\sigma_*}(y,x)$, which implies that
   	$$ b_{\epsilon}(\gamma^{\rm P})\leq 2 \,  \limfunc{inf}_{\gamma}\, b_{\epsilon}(\gamma).$$ 
   \end{proof}

\section{Proofs of Technical Lemmas}
\label{sec:ProofTechnical}

\begin{proof}[\bf Proof of Lemma~\ref{lemma:FOC}]
          In the following we assume that  $f_{\sigma_*}(u|x) f_X(x) > 0$ for all  $(u,x)$ in the joint domain of $(U,X)$. This is without
          loss of generality, because we can define the joint domain of $(U,X)$ such that this is the case. With a slight abuse
          of notation we continue to write ${\cal U} \times {\cal X}$ for the joint domain, even though this need not be a product set.

	To account for the absolute value in the definition of $ b_{\epsilon}(\gamma)$
	in \eqref{DefBias} we let
	\begin{align*}
	b_{\epsilon}(\gamma,s)=\limfunc{sup}_{f_0\in\Gamma_{\epsilon}}
	\left\{ s  \, \mathbb{E}_{P(\beta,f_0)} \left[ \gamma_{{\beta},{\sigma}_{*}}(Y,X) - \delta_{\beta}(U,X) \right] \right\},
	\end{align*}
	for  $s \in \{-1,1\}$.
	We then have $b_{\epsilon}(\gamma) =  \max_{s \in \{-1,1\}}  b_{\epsilon}(\gamma,s) $.
	In the following we drop the arguments $\beta$ and $\sigma_*$ everywhere,  that is,
	we simply write $g(u,x)$, $\gamma(y,x)$, $\delta(u,x)$, $f_*(u|x)$, $\psi(y,x)$, $\lambda^{(1)}(s,x)$, $\lambda^{(2)}(s)$, $\lambda^{(3)}(s)$
	instead of $g_\beta(u,x)$, $\gamma_{\beta,\sigma_*}(y,x)$,
	$\delta_{\beta}(u,x)$, $f_{\sigma_*}(u|x)$, $\psi_{\beta,\sigma_*}(y,x)$, 
	$\lambda^{(1)}_{\beta,\sigma_*}(s)$, $\lambda^{(2)}_{\beta,\sigma_*}(s)$, $\lambda^{(3)}_{\beta,\sigma_*}(s)$.
	The optimal $f_0(u|x)$ in the definition of $b_{\epsilon}(\gamma,s)$ 
	solves, for  $u,x \in {\cal U} \times {\cal X}$ almost surely under the reference distribution,
	\begin{align}
	\widetilde f_0(u|x;  s)
	&= 
	\argmax_{f_0 \in [0,\infty)} \Bigg\{
	s\,  [
	\gamma(g(u,x),x) - \delta(u,x)
	] \, f_X(x) \, f_0 
	- \mu_1(s,x)  \, f_X(x) \, f_0
	\nonumber \\ & \qquad 
	-  \mu_2(s) \, \phi\left( \frac{f_0} {f_*(u|x)} \right)
	\, f_*(u|x) \, f_X(x)
	- \mu_3'(s)  \, \psi(g(u,x),x) \,  \, f_X(x) \, f_0
	\Bigg\} ,
	\label{LagrangeMultiplierMax}  
	\end{align}
	where $\mu_1(s,x) \in \mathbb{R}$, $\mu_2(s)>0$, $\mu_3(s) \in \mathbb{R}^{\dim \psi}$
	are Lagrange multipliers, which we choose to reparameterize as follows
	\begin{align*}	
  	    \mu_1(s,x) &= - \frac{\lambda^{(1)}(s,x)} {\lambda^{(2)}(s)} ,
	    &
    	     \mu_2(s) &= \frac 1 {\lambda^{(2)}(s)} ,
	     &
	     \mu_3(s)  &= - \frac{\lambda^{(3)}(s)} {\lambda^{(2)}(s)} .
	\end{align*}
	Those (reparameterized) Lagrange multipliers 
	need to be chosen such that the constraints
	\begin{align}
	\int_{ {\cal U} \times {\cal X} } \, \widetilde f_0(u|x;s) \,  f_X(x)  \, du \, dx &= 1 ,
	\nonumber \\
	\int_{ {\cal U} \times {\cal X} }  \, \phi\left( \frac{\widetilde f_0(u|x;s)} {f_*(u|x)} \right)
	\, f_*(u|x) \, f_X(x)  \, du \, dx &= \epsilon ,
	\nonumber  \\  
	\int_{ {\cal U} \times {\cal X} }  \,  \psi(g(u,x),x) \, \widetilde f_0(u|x;s) \,  f_X(x)  \, du \, dx &= 0 
	\label{LagrangeConstraints}
	\end{align}
	are satisfied. We need $\lambda^{(2)}(s)>0$ because the second constraint here is actually an inequality constraint ($\leq \epsilon$).
	Our assumptions guarantee that $f_*(u|x)>0$ and $f_X(x)>0$.
	We can therefore rewrite \eqref{LagrangeMultiplierMax} as follows,
	\begin{align*}
	\frac{ \widetilde f_0(u|x;  s) } {f_*(u|x)}
	&= 
	\argmax_{r \geq 0} \left\{  r\, t(u,x|s) - \phi(r) \right\} ,
	\end{align*}
	where $r = f_0 \, f_*(u|x)$, the objective function was multiplied with $f_{\sigma_*}(u|x) f_X(x)$
	(which does not change the value of the $\argmax$),
	and $t(u,x|s) = t_{\beta,\sigma_*}(u,x|s)$ is defined in the statement of the lemma.
	Comparing the last display with  the definition of $\rho(t)$ in \eqref{DefRho} we find that
	if $ \rho\left[ t(u,x|s)  \right] < \infty$, then
	\begin{align*}
	\widetilde f_0(u|x;  s) = f_*(u|x) \, \rho\left[ t(u,x|s)  \right].
	\end{align*}
	The condition $ \rho\left[ t(u,x|s)  \right] < \infty$ is implicitly imposed in the statement of the lemma, because otherwise we could not have 
	$ \mathbb{E}_{P(\beta,f_{\sigma_*})}   \,  \rho\left[ t_{\beta,\sigma_*}(U,X|s)  \right]  = 1$.
	Using the result in the last display we find that the constraints \eqref{LagrangeConstraints}
	are exactly the conditions \eqref{RewriteConstraints} imposed in the lemma.
	Under the conditions of the lemma we therefore have
	\begin{align*}
	b_{\epsilon}(\gamma,s)
	&=\limfunc{sup}_{f_0\in\Gamma_{\epsilon}}
	\left\{ s  \, \mathbb{E}_{P(\beta,f_0)} \left[ \gamma(Y,X) - \delta(U,X) \right] \right\}
	\\
	&=  \int_{ {\cal U} \times {\cal X} }  \,  \left[ \gamma(g(u,x),x) - \delta (u,x) \right]  \, \widetilde f_0(u|x;s) \,  f_X(x)  \, du \, dx 
	\\
	&=   s \; \mathbb{E}_{P(\beta,f_{\sigma_*})} \left\{  \left[\gamma(Y,X)- \delta(U,X) \right] 
	\phantom{\Big|}
	\rho\left[ t(U,X|s)  \right] \right\} ,
	\end{align*}
	and from $b_{\epsilon}(\gamma) =  \max_{s \in \{-1,1\}}  b_{\epsilon}(\gamma,s) $ we thus obtain the statement of the lemma.
\end{proof}

\begin{proof}[\bf Proof of Lemma~\ref{lemma:ConvexDual}]
	\# \underline{Part (i):}
	For $\nu = 0$
	we have $\phi = \overline \phi$.
	Our assumptions imply that there exists $\tau>0$ such that 
	$\phi'(r)$, $\phi''(r)$, $\phi'''(r)$ and $\phi''''(r)$ are all uniformly bounded over $r \in [1-\tau,1+\tau]$.
	We can choose $\eta>0$ such that $[\rho(-\eta), \rho(\eta)] \subset  [1-\tau,1+\tau]$.
	The conjugate of the convex function $\phi: \mathbb{R} \rightarrow \mathbb{R}$
	is given by
	\begin{align}
	\phi_*(t)  &=  \max_{r \geq 0} \, \left[ r\, t - \phi(r) \right] 
	= \rho(t)\, t - \phi(\rho(t))  .
	\label{ConvexDual}
	\end{align}
	We  have $\rho(t)  =  \phi_*'(t)$,
	which is the inverse function of $\phi'(r)$; that is,
	$\phi'(\rho(t)  )=t$.
	We can  express all derivatives of $\phi_*$ in terms of derivatives of $\phi$,
	for example, $ \phi_*''(t) = 1/  \phi''( \rho(t) )$
	and $ \phi_*'''(t) = - \phi'''( \rho(t) ) / [\phi''( \rho(t) )]^3$.
	A Taylor expansion of $  \rho(t)   =  \phi_*'(t) $ around $t=0=\phi'(1)$ reads
	\begin{align*}
	\rho(t)
	&=  1 + \frac {t} {\phi''(1)}  + t^2 \, R_1(t) ,
	\end{align*} 
	where by the mean-value formula for the remainder term we have
	\begin{align*}
	\left| R_1(t) \right| \leq  \frac {1} 2 \, \sup_{t' \in [-\eta,\eta]} \,  \left| \phi_*'''(t')  \right|
	\leq   \underbrace{
		\frac 1 2  \, \sup_{r \in [1-\tau, 1+\tau]} \,  \left| \frac{ \phi'''(r) } {[\phi''(  r )]^3}  \right| 
	}_{=:c_1 < \infty}   .
	\end{align*}
	Similarly, a Taylor expansion of
	$\phi( \rho(t) ) = t \, \rho(t) - \phi_*(t) $ around $t=0$ reads
	\begin{align*}
	\phi( \rho(t) )
	&=  \frac {t^2} {2 \, \phi''(1)} + t^3 \, R_2(t) ,
	\end{align*}
	where again by the mean-value formula for the remainder we have
	\begin{align*}
	\left| R_2(t) \right| \leq  
	\underbrace{
		\frac 1 6 \, \sup_{r \in [1-\tau, 1+ \tau]} \, 
		\left| - \frac{2 \phi'''(r) } {[\phi''(  r )]^3} + \frac{3 \phi'(r) [\phi'''(r)]^2 } {[\phi''(  r )]^5}
		-  \frac{\phi'(r) \phi''''(r) } {[\phi''(  r )]^4} \right| 
	}_{=:c_2 < \infty} .
	\end{align*}
	Continuity of $R_1(t)$ and $R_2(t)$ in a neighborhood of $t=0$ is also guaranteed by $\phi'(r)$ being four times continuously differentiable
	in neighborhood around $r=1$. This concludes the proof of part (i).
	
	\medskip
	\noindent
	\# \underline{Part (ii):}
	For $\nu>0$
	the function $\phi(r) = \overline \phi(r)+\nu \, (r-1)^2$ still satisfies all the assumptions of part (i)
	of the lemma, that is, we can apply part (i) to find that there exists 
	$\widetilde c_1>0$, $\widetilde  c_2>0$ and $\eta>0$ such that
	for all $t \in [-\eta,\eta]$ we have
	\begin{align}
	\left|  R_1(t)  \right| \, \leq \, \widetilde c_1 \, t^2, 
	\quad
	\text{and}
	\quad
	\left| R_2(t)  \right| \, \leq  \, \widetilde c_2 \, t^3.
	\label{BoundRhoTilde}            
	\end{align}
	What is left to show here is that there exists constant $c_1>0$ and $c_2>0$
	such that \eqref{BoundRho} also holds for $t<-\eta$ and for $t>\eta$.
	
	We have $\phi'(r) = \overline \phi'(r) + \nu (r-1)$. Plugging in $r = \rho(t)$ we have $\phi'(\rho(t)) =t$, and therefore
	$ t = \overline \phi'(\rho(t)) + \nu [\rho(t)-1]$.
	Our assumptions  imply that $ \overline \phi'(\rho(t)) > 0$ for $t>0$
	and $ \overline \phi'(\rho(t)) < 0$ for $t<0$. We therefore find that
	\begin{align}
	\left| \rho(t)-1 \right| =  \frac{ \left|  t  -  \overline \phi'(\rho(t)) \right| } \nu 
	\leq  \frac {|t|} \nu  .
	\label{CrudeBoundRho}
	\end{align}
	Using   \eqref{BoundRhoTilde} and \eqref{CrudeBoundRho}, and choosing $c_1 = \max\left\{  \widetilde c_1 , \, [1/\nu + 1/ \phi''(1)] / \eta \right\}$, we obtain
	$$
	\left| \rho(t) - 1 - \frac{t} { \phi''(1)}  \right| \,  \leq \,  c_1 \, t^2 ,
	$$
	for all $t \in \mathbb{R}$. This is the first inequality that we wanted to show.
	
	Using again the convex conjugate defined in \eqref{ConvexDual} we have
	\begin{align*}
	\phi( \rho(t) ) =  t \, \rho(t) - \phi_*(t) 
	=  t \, \rho(t) -  \max_{r \geq 0} \, \left[ r\, t - \phi(r) \right] 
	\leq  t [ \rho(t) - 1] = |t| \, \left|  \rho(t) - 1 \right| \,  ,
	\end{align*}
	where in the second to last step we used that $r=1$ is one possible choice for $r \geq 0$,
	and we have $\phi(1)=0$,
	and in the last step we used that ${\rm sign}[ \rho(t) - 1 ] = {\rm sign}(t)$.
	Our assumptions imply that $\phi(r) \geq 0$,
	that is,  $\left| \phi(r) \right| =  \phi(r)$.
	The result in the last display together with
	\eqref{CrudeBoundRho} therefore give
	\begin{align*}
	\left| \phi( \rho(t) ) \right| \leq   \frac {t^2} \nu ,
	\end{align*} 
	for all $t \in \mathbb{R}$.
	Using this and  \eqref{BoundRhoTilde},
	and choosing $c_2 =  \max\left\{  \widetilde c_2 , \, [1/\nu + 1/ \{2\phi''(1)\}] / \eta \right\}$,
	we thus obtain
	$$
	\left|   \phi( \rho(t) )  -  \frac {t^2} {2 \, \phi''(1)}  \right| \,  \leq  \,  c_2 \, t^3 ,
	$$    
	for all $t \in \mathbb{R}$,
	which is the second  inequality  that we wanted to show. 
	Continuity of $R_1(t)$ and $R_2(t)$ in $\mathbb{R}$ is also guaranteed by $\phi'(r)$ being four times continuously differentiable
	in $r \in (0, \infty)$.
	This concludes the proof of part~(ii).
\end{proof}

\begin{proof}[\bf Proof of Lemma~\ref{lemma:SupremumBoundConvex}]
	Let $f_0\in\Gamma_{\epsilon}$.
	Remember the definition of the posterior density $p_{\beta,\sigma_*}(u\,|\, y,x)$ in \eqref{eq_posterior}.
	Define
	\begin{align*}
	\widetilde f_0(u|x) :=  \mathbb{E}_{P(\beta,f_0)} \left[ p_{\beta,\sigma_*}(u\,|\, Y,x) \right] 
	= \int_{\cal U}  \, p_{\beta,\sigma_*}(u\,|\, g_\beta(\tilde u,x) ,x) \,  f_0(\tilde u|x)  \, d \tilde u.
	\end{align*}
	Then, for any $x \in {\cal X}$ we have $\widetilde f_0(u|x) \geq 0$, for all $u \in {\cal U}$,
	and $\int_{\cal U} \widetilde f_0(u|x) du =1$; that is,
	$\widetilde f_0(u|x)$ is a probability density over ${\cal U}$.
	Furthermore, by construction we have   
	\begin{align}
	\mathbb{E}_{P(\beta,f_0)} \left\{
	\mathbb{E}_{p_{\beta,\sigma_*}} \left[ \zeta(U,X) \,|\, Y,X \right] 
	\right\}  
	=     \mathbb{E}_{P(\beta,\widetilde f_0)} \left[ \zeta(U,X) \right]   .
	\label{PropertyTildeF0}
	\end{align}
	We also find that
	\begin{align}
	\mathbb{E}_{P(\beta,\widetilde f_0)} [\psi_{\beta,\sigma_{*}}(Y,X)]= 
	\mathbb{E}_{P(\beta,f_0)} \left\{
	\mathbb{E}_{p_{\beta,\sigma_*}} \left[ \psi_{\beta,\sigma_{*}}(Y,X) \,|\, Y,X \right] 
	\right\}  
	=    \mathbb{E}_{P(\beta, f_0)} [\psi_{\beta,\sigma_{*}}(Y,X)] = 0 .
	\label{ConditionGammaEpsilon1}
	\end{align}
	Furthermore, we have
	\begin{align*}
	d(\widetilde f_0,f_{\sigma_{*}})
	&=\int_{{\cal{X}}}\int_{{\cal{U}}} \phi \left(\frac{\widetilde f_0(u\,|\, x)}{f_{\sigma_{*}}(u\,|\, x)}\right)f_{\sigma_*}(u\,|\, x)f_X(x) \,du\,dx
	\nonumber   \\
	&=\int_{{\cal{X}}}\int_{{\cal{U}}} \phi \left(
	\frac{\int_{\cal U}  \, p_{\beta,\sigma_*}(u\,|\, g_\beta(\tilde u,x) ,x) \,  f_0(\tilde u|x)  \, d \tilde u}
	{f_{\sigma_{*}}(u\,|\, x)}
	\right)f_{\sigma_*}(u\,|\, x)f_X(x) \,du \,dx
	\nonumber   \\
	&=\int_{{\cal{X}}}\int_{{\cal{U}}} \phi \left(
	\int_{\cal U}  \,  \frac{ f_0(\tilde u|x) } {f_{\sigma_{*}}(\tilde u\,|\, x)}  \,
	K_{\beta,\sigma_*}(\tilde u | u,x) \, d \tilde u
	\right)f_{\sigma_*}(u\,|\, x)f_X(x) \,du \,dx ,
	\end{align*}
	where we defined
	\begin{align*}
	K_{\beta,\sigma_*}(\tilde u | u,x) =  \frac{f_{\sigma_{*}}(\tilde u\,|\, x) \,p_{\beta,\sigma_*}(u\,|\, g_\beta(\tilde u,x) ,x) }  {f_{\sigma_{*}}(u\,|\, x)} .
	\end{align*}
	Using the definition of $p_{\beta,\sigma_*}(u\,|\, y,x)$ one can verify that
	$K_{\beta,\sigma_*}(\tilde u | u,x)  \geq 0$, for all $\tilde u \in {\cal U}$,
	and $\int_{\cal U} K_{\beta,\sigma_*}(\tilde u | u,x)  d\tilde u = \frac{ \mathbb{E}_{P(\beta,f_{\sigma_*})} \left[ p_{\beta,\sigma_*}(u\,|\, Y,x) \right]  }  {f_{\sigma_{*}}(u\,|\, x)}=1$, 
	almost surely (under $P(\beta,f_{\sigma_*})$) for $u \in {\cal U}$ and $x \in {\cal X}$.
	Thus,
	$K_{\beta,\sigma_*}(\tilde u | u,x) $ is a probability density over $\tilde u \in {\cal U}$, for all $u,x$.
	Also using that $\phi(r)$ is convex, we can therefore  apply Jensen's inequality to obtain
	\begin{align}
	d(\widetilde f_0,f_{\sigma_{*}}) 
	&\leq
	\int_{{\cal{X}}}\int_{{\cal{U}}} 
	\int_{\cal U}  \,  \phi \left( \frac{ f_0(\tilde u|x) } {f_{\sigma_{*}}(\tilde u\,|\, x)} \right)  \,
	K_{\beta,\sigma_*}(\tilde u | u,x) \, d \tilde u \,
	f_{\sigma_*}(u\,|\, x)f_X(x) \,du \,dx
	\nonumber   \\
	&=
	\int_{{\cal{X}}}\int_{{\cal{U}}}   \phi \left( \frac{ f_0(\tilde u|x) } {f_{\sigma_{*}}(\tilde u\,|\, x)} \right)
	\underbrace{
		\left[   \int_{\cal U}  \,    f_{\sigma_*}(u\,|\, x)  \,
		K_{\beta,\sigma_*}(\tilde u | u,x) \, d   u \right] }_{= f_{\sigma_{*}}(\tilde u\,|\, x)}
	f_X(x) \,d\tilde u \,dx
	\nonumber   \\
	&=   d( f_0,f_{\sigma_{*}}) \leq \epsilon .
	\label{ConditionGammaEpsilon2}         
	\end{align}
	Because $\widetilde f_0 $ satisfies \eqref{ConditionGammaEpsilon1}
	and \eqref{ConditionGammaEpsilon2} we thus have $\widetilde f_0 \in \Gamma_{\epsilon}$.
	We have thus shown that for every $f_0\in\Gamma_{\epsilon}$
	there exists $\widetilde f_0 \in \Gamma_{\epsilon}$ such that \eqref{PropertyTildeF0} holds.
	Let $\widetilde \Gamma_{\epsilon}$  be the set of all such $\widetilde f_0$  obtained for an $f_0\in\Gamma_{\epsilon}$.
	Since $\widetilde \Gamma_{\epsilon} \subset  \Gamma_{\epsilon}$ we find that
	\begin{align*}
	\sup_{f_0\in\Gamma_{\epsilon}}   \left|
	\mathbb{E}_{P(\beta,f_0)} \left\{
	\mathbb{E}_{p_{\beta,\sigma_*}} \left[ \zeta(U,X) \,|\, Y,X \right] 
	\right\} \right|
	=  \sup_{\widetilde f_0\in \widetilde\Gamma_{\epsilon}}   \left| \mathbb{E}_{P(\beta,\widetilde f_0)} \left[ \zeta(U,X) \right] \right|
	\leq 
	\sup_{f_0\in\Gamma_{\epsilon}}  \left| \mathbb{E}_{P(\beta,f_0)} \left[ \zeta(U,X) \right] \right| .
	\end{align*}
\end{proof}

\section{Robustness in prediction\label{App_Pred}}

Under squared loss, we wish to find a predictor $\gamma_{\widehat{\beta},\widehat{\sigma}}(Y_i,X_i)$, for some function $\gamma$, such that the worst-case mean squared prediction error is minimum. That is, our goal is to minimize
$$e_{\epsilon}(\gamma)=\underset{f_0\in\Gamma_{\epsilon}} \sup\, \mathbb{E}_{P(\beta,f_0)}\, [(\delta_{\beta}(U,X)-\gamma_{\beta,\sigma_*}(Y,X))^2]
$$ with respect to $\gamma$. Similarly to our measure of worst-case specification error, here the mean squared  prediction error is asymptotic, hence well-suited for settings with a large cross-section (e.g., settings with many teachers).

We first state the following local result, which is a direct generalization of Lemma \ref{lem_bias}.

\begin{lemma}\label{lem_prediction_expansion}
	In addition to defining $\widetilde \psi(y,x) = \psi(y,x) - \mathbb{E}_* \left[ \psi(Y,X)  \big| X=x \right]$, let $\widetilde \gamma(y,x) = \gamma(y,x) - \mathbb{E}_* \left[ \gamma(Y,X)  \big|  X=x \right]$ and $\widetilde  \delta(u,x) = \delta(u,x) - \mathbb{E}_* \left[ \delta(U,X)  \big| X=x \right]$. Suppose that 
	$\phi(r) = \overline \phi(r)+\nu \, (r-1)^2$, with $\nu \geq 0$,
	and a function $\overline \phi(r)$ that is four times continuously differentiable 
	with $\overline \phi(1) =0$ and $\overline \phi''(r)>0$, for all $r \in (0,\infty)$.
	Assume
	$ \mathbb{E}_{P(\beta,f_{\sigma_*})}   \psi_{{\beta},{\sigma}_{*}}(Y,X) = 0$ and 
	$\mathbb{E}_{P(\beta,f_{\sigma_{*}})}\left[\widetilde \psi_{\beta,\sigma_{*}}(Y,X) \widetilde \psi_{\beta,\sigma_{*}}(Y,X)'\right] >0$.
	Furthermore, assume that one of the following holds:
	\begin{itemize}
		\item[(i)] $\nu = 0$, and the functions
		$\left| \gamma_{{\beta},{\sigma}_{*}}(y,x) \right|$, $\left| \delta_\beta(u,x) \right| $ and $\left| \psi_{{\beta},{\sigma}_{*}}(y,x) \right|$ are bounded over the domain of $Y$, $U$, $X$.
		
		\item[(ii)] $\nu>0$, and 
		$\mathbb{E}_{P(\beta,f_{\sigma_*})} \left| \gamma_{{\beta},{\sigma}_{*}}(Y,X) - \delta_{\beta}(U,X) \right|^6 < \infty$,
		and  $\mathbb{E}_{P(\beta,f_{\sigma_*})} \left| \psi_{{\beta},{\sigma}_{*}}(Y,X)  \right|^3 < \infty$.
		
	\end{itemize}
	Then, as $\epsilon\rightarrow 0$ we have
	\begin{align*}
	e_{\epsilon}(\gamma)&=
	\mathbb{E}_{P(\beta,f_{\sigma_*})} \left[ \left( \gamma_{{\beta},{\sigma}_{*}}(Y,X) - \delta_{\beta}(U,X) \right)^2  \right]\\
	&\quad
	+\epsilon^{\frac{1}{2}}\Bigg(\frac{2}{\phi''(1)}
	{\rm Var}_{P(\beta,f_{\sigma_{*}})}\Bigg\{ \left( \gamma_{{\beta},{\sigma}_{*}}(Y,X) - \delta_{\beta}(U,X) \right)^2  
	\\ & \qquad \qquad \quad
	- \mathbb{E}_{P(\beta,f_{\sigma_{*}})}
	\left[ \left(  \gamma_{{\beta},{\sigma}_{*}}(Y,X) -  \delta_{\beta}(U,X) \right)^2  \bigg| X \right]
	- \lambda'\widetilde\psi_{\beta,\sigma_{*}}(Y,X) \Bigg\} \Bigg)^{\frac{1}{2}}
		+{\cal O}(\epsilon),
	\end{align*}
	where $$\lambda{=}\left\{\mathbb{E}_{P(\beta,f_{\sigma_{*}})}\left[\widetilde \psi_{\beta,\sigma_{*}}(Y,X) \widetilde \psi_{\beta,\sigma_{*}}(Y,X)'\right]\right\}^{-1}\mathbb{E}_{P(\beta,f_{\sigma_{*}})}\left[\left(\gamma_{{\beta},{\sigma}_{*}}(Y,X){-}\delta_{\beta}(U,X)\right)^2 \, \widetilde \psi_{\beta,\sigma_{*}}(Y,X)\right].$$
\end{lemma}

\vskip .3cm

Let $\gamma^{\rm P}$ as in (\ref{eq_gamma_P}), so $\gamma^{\rm P}_{\widehat{\beta},\widehat{\sigma}}(Y_i,X_i)$ is the empirical Bayes estimate of $\delta_{\beta}(U_i,X_i)$. Under correct specification of the reference density $f_{\sigma}$, the posterior mean $\gamma^{\rm P}_{\beta,\sigma_*}(Y_i,X_i)$ is the minimum mean squared error predictor of $\delta_{\beta}(U_i,X_i)$ under squared loss. Under misspecification of $f_{\sigma}$, Lemma \ref{lem_prediction_expansion} implies that the leading term of the worst-case mean squared error is minimized at $\gamma=\gamma^{\rm P}$. Moreover, the lemma also implies the stronger result that the first-order term in the expansion of the worst-case mean squared prediction  error (which is a multiple of $\epsilon^{\frac{1}{2}}$) is also minimized at $\gamma^{\rm P}$, provided the following condition holds almost surely:
\begin{equation}
\label{cond_skew}\mathbb{E}_{p_{\beta,\sigma_{*}}}\, \left[({\delta}_{\beta}(U,X)-{\gamma}^{\rm P}_{\beta,\sigma_*}(Y,X))^3\,|\, Y,X\right]=0.
\end{equation}

While (\ref{cond_skew}) is restrictive in general, it is satisfied in the fixed-effects model (\ref{FE_mod}), when the researcher wishes to predict the quality $\alpha_i$ of teacher $i$. Indeed, in that case (\ref{cond_skew}) is equivalent to the posterior skewness of $\alpha_i$ being zero, when using the normal reference model as the prior. Since the normal distribution is symmetric, (\ref{cond_skew}) is satisfied, and the empirical Bayes estimator $\gamma^{\rm P}_{\widehat{\beta},\widehat{\sigma}}(Y_i,X_i)=\widehat{\mu}_{\alpha}+\widehat{\rho}(\overline{Y}_i-\widehat{\mu}_{\alpha})$ has minimum worst-case mean squared  prediction  error, up to second-order terms in $\epsilon^{\frac{1}{2}}$. 

We also have a fixed-$\epsilon$ bound in the spirit of Theorem \ref{theo_global}.

\begin{theorem}
	\label{theo_prediction_global}
	Let $\gamma_{\beta,\sigma_*}^{\rm P}$ as in (\ref{eq_gamma_P_app}).
	Then, for all $\epsilon>0$,
	$$ e_{\epsilon}(\gamma_{\beta,\sigma_*}^{\rm P})\leq 4 \,  \inf_{\gamma}\, e_{\epsilon}(\gamma_{\beta,\sigma_*}).$$
\end{theorem}

\vskip .3cm

Theorem~\ref{theo_prediction_global} shows that EB estimators are optimal, up to a factor of at most four, 
in terms of worst-case mean squared  prediction  error. In model (\ref{FE_mod}), when ${\varepsilon}_{1},...,{\varepsilon}_{J}$ are normally distributed and $\alpha_1,...,\alpha_N$ are parameters belonging to an $L^2$ ball, empirical Bayes James-Stein estimators are known to be optimal in terms of asymptotic  minimax mean squared error since they achieve the Pinsker bound (see Wasserman, 2006, Chapter 7). Here, by contrast, we consider a worst case computed in a set of unrestricted, possibly non-normal joint distributions of $\alpha,\varepsilon_1,...,\varepsilon_{J}$.

\begin{proof}[\bf Proof of Lemma~\ref{lem_prediction_expansion}]
	This statement of the lemma is obtained from Lemma~\ref{lem_bias_app} by replacing $ ( \gamma_{{\beta},{\sigma}_{*}}(Y,X) - \delta_{\beta}(U,X) )$ by
	$ \left( \gamma_{{\beta},{\sigma}_{*}}(Y,X) - \delta_{\beta}(U,X) \right)^2 $.
	The proof is obtained by the same replacement from the proof of Lemma~\ref{lem_bias_app}.
\end{proof}

\begin{proof}[\bf Proof of Theorem~\ref{theo_prediction_global}]
	By definition we have 
	\begin{align*}
	e_{\epsilon}(\gamma) &=\underset{f_0\in\Gamma_{\epsilon}} \sup\, \mathbb{E}_{P(\beta,f_0)}\, [(\delta_{\beta}(U,X)-\gamma_{\beta,\sigma_*}(Y,X))^2],
	\\
	e_{\epsilon}(\gamma^{\rm P})&=
	\underset{f_0\in\Gamma_{\epsilon}} \sup\, \mathbb{E}_{P(\beta,f_0)}\, [(\delta_{\beta}(U,X)-\gamma^{\rm P}_{\beta,\sigma_*}(Y,X))^2].
	\end{align*} 
	Using that $(a-b)^2 \leq 2 (a^2+b^2)$ 
	with $a=\delta_{\beta}(U,X) - \gamma_{{\beta},{\sigma}_{*}}(Y,X) $
	and $b=\gamma^{\rm P}_{{\beta},{\sigma}_{*}}(Y,X) - \gamma_{{\beta},{\sigma}_{*}}(Y,X)$
	we obtain
	\begin{align*}
	e_{\epsilon}(\gamma^{\rm P}) &\leq 
	2 \,
	\sup_{f_0\in\Gamma_{\epsilon}}\,
	\left|
	\mathbb{E}_{P(\beta,f_0)} \left[  \left( \delta_{\beta}(U,X) - \gamma_{{\beta},{\sigma}_{*}}(Y,X)    \right)^2  \right]
	+  \mathbb{E}_{P(\beta,f_0)} \left[  \left(  \gamma^{\rm P}_{{\beta},{\sigma}_{*}}(Y,X)  - \gamma_{{\beta},{\sigma}_{*}}(Y,X) \right)^2 \right] \right|
	\\
	&\leq
	2 e_{\epsilon}(\gamma)
	+ 2\, \sup_{f_0\in\Gamma_{\epsilon}}
	\left|
	\mathbb{E}_{P(\beta,f_0)} \left[  \left( \gamma_{{\beta},{\sigma}_{*}}(Y,X) - \gamma^{\rm P}_{{\beta},{\sigma}_{*}}(Y,X) \right)^2 \right]
	\right| .
	\end{align*}
	We furthermore have
	\begin{align*}
	\sup_{f_0\in\Gamma_{\epsilon}} &
	\left|
	\mathbb{E}_{P(\beta,f_0)} \left[  \left( \gamma_{{\beta},{\sigma}_{*}}(Y,X) - \gamma^{\rm P}_{{\beta},{\sigma}_{*}}(Y,X) \right)^2 \right]
	\right| 
	\\
	&=  \sup_{f_0\in\Gamma_{\epsilon}}
	\left|
	\mathbb{E}_{P(\beta,f_0)} \left\{  \left[
	\mathbb{E}_{p_{\beta,\sigma_*}} \left( \gamma_{{\beta},{\sigma}_{*}}(Y,X) - \delta_{\beta}(U,X)\,|\, Y, X \right) \right]^2 \right\}
	\right| 
	\\
	&\leq  \sup_{f_0\in\Gamma_{\epsilon}}
	\left|
	\mathbb{E}_{P(\beta,f_0)} \left\{  
	\mathbb{E}_{p_{\beta,\sigma_*}}\left[ \left( \gamma_{{\beta},{\sigma}_{*}}(Y,X) - \delta_{\beta}(U,X) \right)^2 \,|\, Y, X  \right] \right\}
	\right| 
	\\
	&\leq 
	\sup_{f_0\in\Gamma_{\epsilon}}
	\left|
	\mathbb{E}_{P(\beta,f_0)}  \left[ \left( \gamma_{{\beta},{\sigma}_{*}}(Y,X) - \delta_{\beta}(U,X) \right)^2  \right]    \right| 
	= e_{\epsilon}(\gamma) ,
	\end{align*}
	where in the first step we used the definition of $ \gamma^{\rm P}_{{\beta},{\sigma}_{*}}(y,x)$,
	in the second step we applied the Cauchy-Schwarz inequality,
	and in the last line we used Lemma~\ref{lemma:SupremumBoundConvex}
	and the definition of $e_{\epsilon}(\gamma)$.
	Combining the results of the last two displays we obtain that
	$$ e_{\epsilon}(\gamma_{\beta,\sigma_*}^{\rm P})\leq 4 \,  \inf_{\gamma}\, e_{\epsilon}(\gamma_{\beta,\sigma_*}).$$
\end{proof}

\section{Simulations\label{sec_App_sim}}

Here we provide details on the simulations summarized in Subsection \ref{subsec_sim}. We consider four data generating processes (DGP) based on model (\ref{FE_mod}). We use a standard normal as the (possibly misspecified) reference model for $\alpha_i$. In the first two DGP, we draw $\alpha_i$ from a standard normal distribution. Hence, the reference normal model for $\alpha_i$ is correctly specified in this case. We compare two specification for $\varepsilon_{ij}$. In DGP 1, $\varepsilon_{ij}$ are i.i.d. standard normal. In DGP 2, we model heteroskedastic errors as ${\cal{N}}(0,s(\alpha_i)^2)$, where $s(\alpha_i)=\boldsymbol{1}\{\alpha_i>0\}\times .1+ \boldsymbol{1}\{\alpha_i\leq 0\}\times 1.41$, so the variance of errors is the same as in DGP 1. In the next two DGP, we draw $\alpha_i$ from a Beta distribution with parameters $(11,1)$, shifted and rescaled such that $\alpha_i$ has mean 0 and variance 1. This distribution is skewed to the left, and the reference normal model for $\alpha_i$ is thus misspecified. In DGP 3, $\varepsilon_{ij}$ are i.i.d. standard normal. In DGP 4, we model heteroskedastic errors as ${\cal{N}}(0,s(\alpha_i)^2)$, where $s(\alpha_i)=\boldsymbol{1}\{\alpha_i>0\}\times .1+ \boldsymbol{1}\{\alpha_i\leq 0\}\times 1.61$, so the variance of errors is the same as in the other DGP.

In all DGP, we compare the performance of four estimators: the fixed-effects estimator given by (\ref{eq_cdf_FE}), the PAE given by (\ref{posterior_FE}), the model-based estimator given by (\ref{model_based_FE}), and a nonparametric kernel deconvolution estimator with normal errors (Stefanski and Carroll, 1990). Unlike the other three estimators, the nonparametric kernel deconvolution estimator requires choosing a tuning parameter. We use the MISE-minimization approach of Delaigle and Gijbels (2004). To implement bandwidth selection and estimator, we use the codes available on Aurore Delaigle's page: https://researchers.ms.unimelb.edu.au/$\sim$aurored/links.html\#Code. To estimate the variance components $s_{\alpha}^2$ and $s_{\varepsilon}^2$, we use a minimum-distance estimator based on the empirical covariance matrix of the $Y_{ij}$. We set $\mu_{\alpha}=0$ and do not estimate it.

We run 1000 simulations for $n=1000$, and report estimates of $F_{\alpha}(a)=\mathbb{E}[\boldsymbol{1}\{\alpha \leq a\}]$ for a grid of $a$ values. We focus on bias (in absolute value), standard deviation, and root-MSE. We vary the number of measurements by taking $J\in\{2,20\}$. When $J=2$, the shrinkage factor is $\rho=.67$. In addition, the average value of the informativeness measure (\ref{eq_R2_FE}), over simulations and $a$ values, is $R^2=27\%$. When $J=20$, the shrinkage factor is $\rho=.95$. In addition, the average informativeness value is $R^2=69\%$.

In Figure \ref{fig_mc}, we show the results of the simulations in the cases where $\alpha_i$ is normally distributed. In this case, we expect the model-based estimator (in dashed lines) to be particularly well-behaved, since it is consistent and, for a suitable weighting of the minimum-distance estimator of $s_{\alpha}^2$ and $s_{\varepsilon}^2$, efficient as well. In DGP 1, $\varepsilon_{ij}$ are normal homoskedastic, so both the PAE (in solid) and the deconvolution estimator (in dotted) are consistent for fixed $J$ as $n$ tends to infinity. In the top two panels of Figure \ref{fig_mc}, we see that the bias of these two estimators is indeed small. Yet, the model-based estimator has the smallest bias, as well as the smallest variance and root-MSE.

\begin{figure}[tbp]
	\caption{Monte Carlo results for $F_{\alpha}(a)=\mathbb{E}[\boldsymbol{1}\{\alpha \leq a\}]$ in the fixed-effects model $Y_{ij}=\alpha_i+{\varepsilon}_{ij}$, when the distribution of $\alpha$ is correctly specified\label{fig_mc}}
	\begin{center}
		\begin{tabular}{ccc}
			Bias&Std&RMSE \\
			\multicolumn{3}{c}{(1) $\alpha_i\sim {iid}{\cal{N}}\left(0,1\right)$, ${\varepsilon}_{ij}\sim{iid}{\cal{N}}\left(0,1\right)$, $J=2$}\\
			\includegraphics[width=50mm, height=40mm]{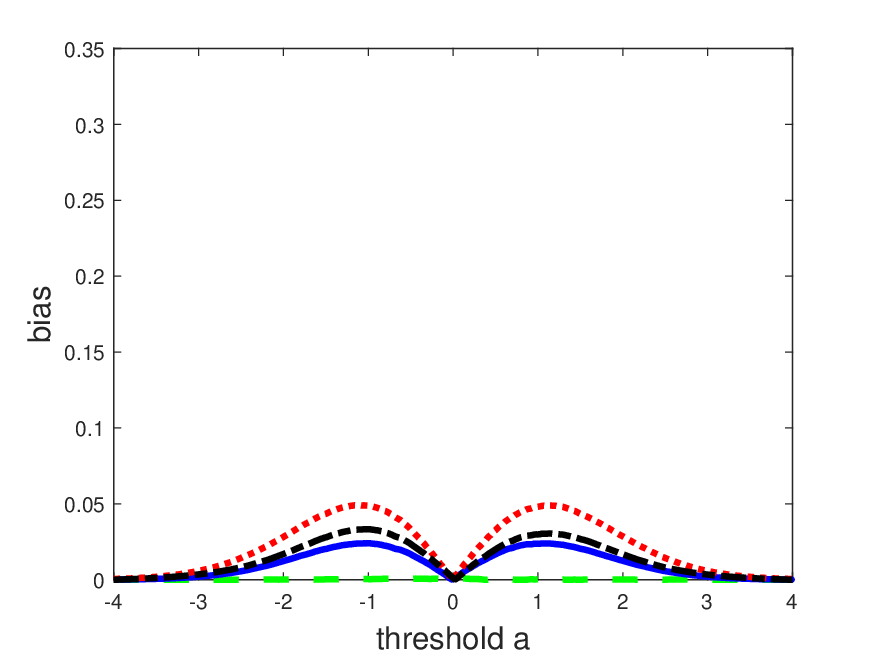}&	\includegraphics[width=50mm, height=40mm]{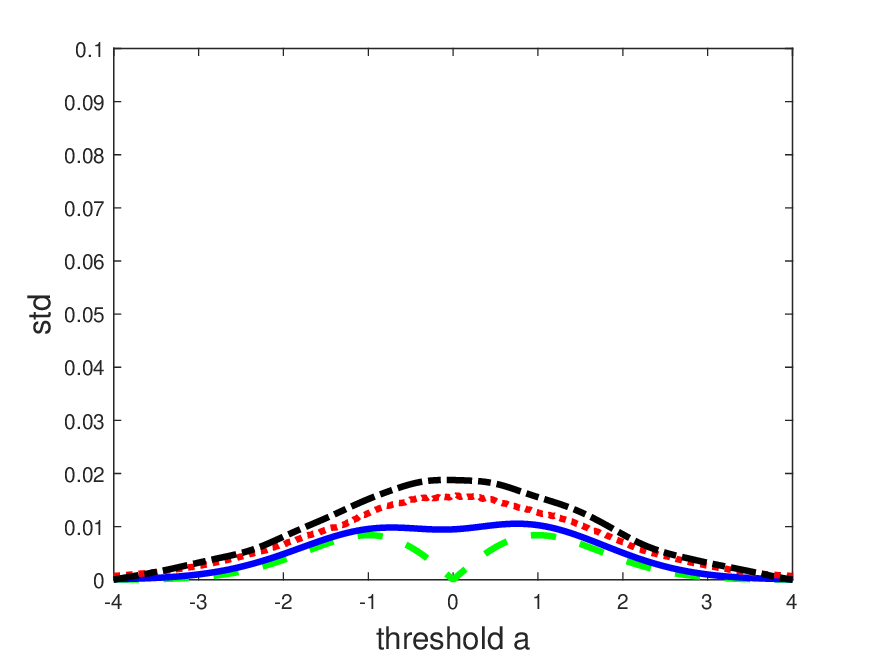}&	\includegraphics[width=50mm, height=40mm]{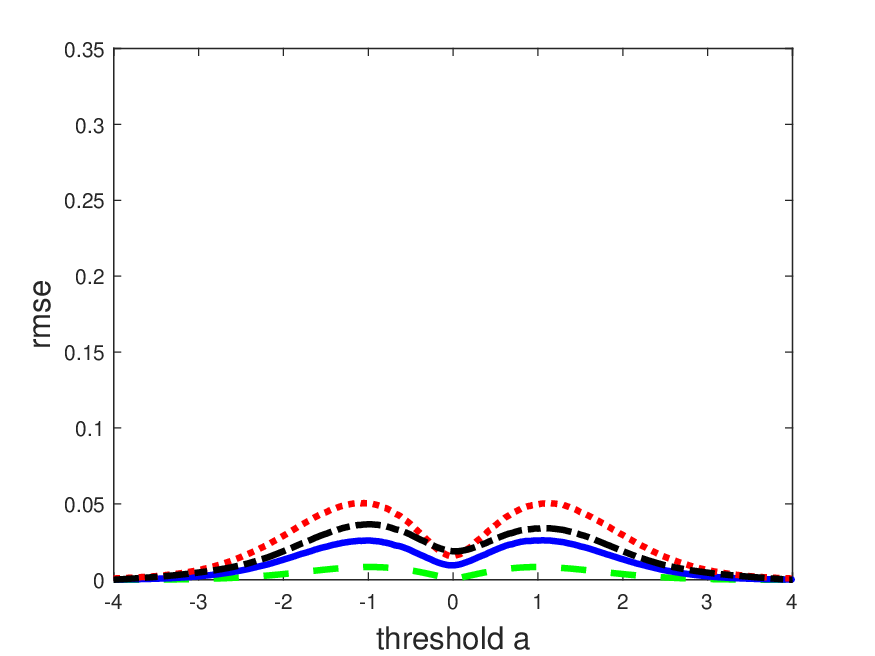}\\
			\multicolumn{3}{c}{(2) $\alpha_i\sim {iid}{\cal{N}}\left(0,1\right)$, ${\varepsilon}_{ij}\sim{iid}{\cal{N}}\left(0,1\right)$, $J=20$}\\
			\includegraphics[width=50mm, height=40mm]{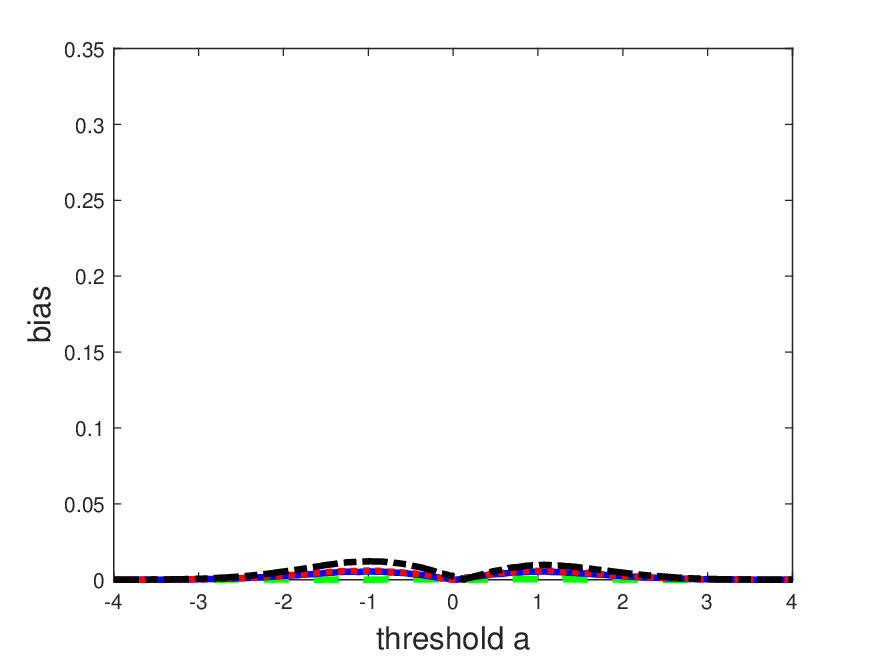}&	\includegraphics[width=50mm, height=40mm]{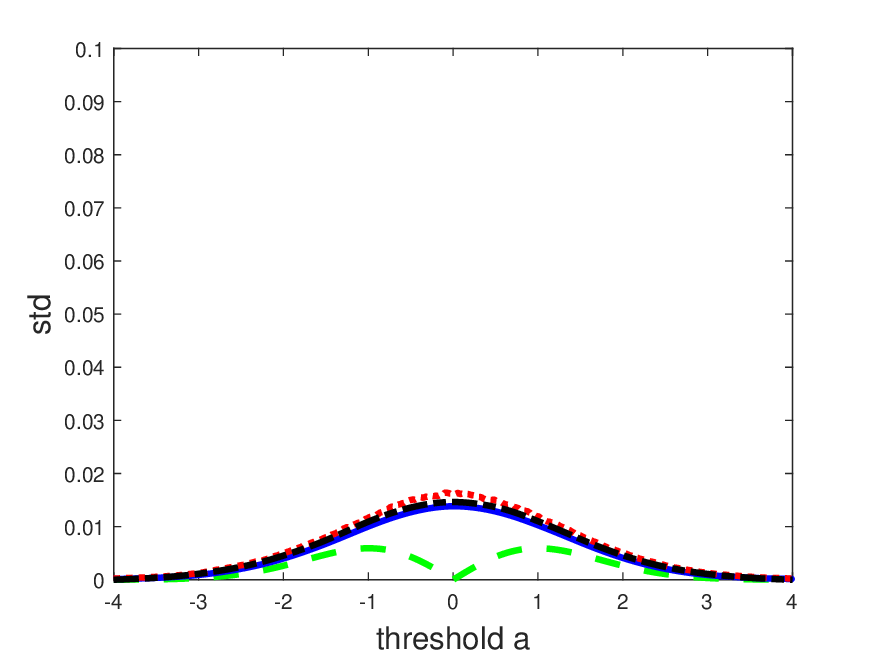}&	\includegraphics[width=50mm, height=40mm]{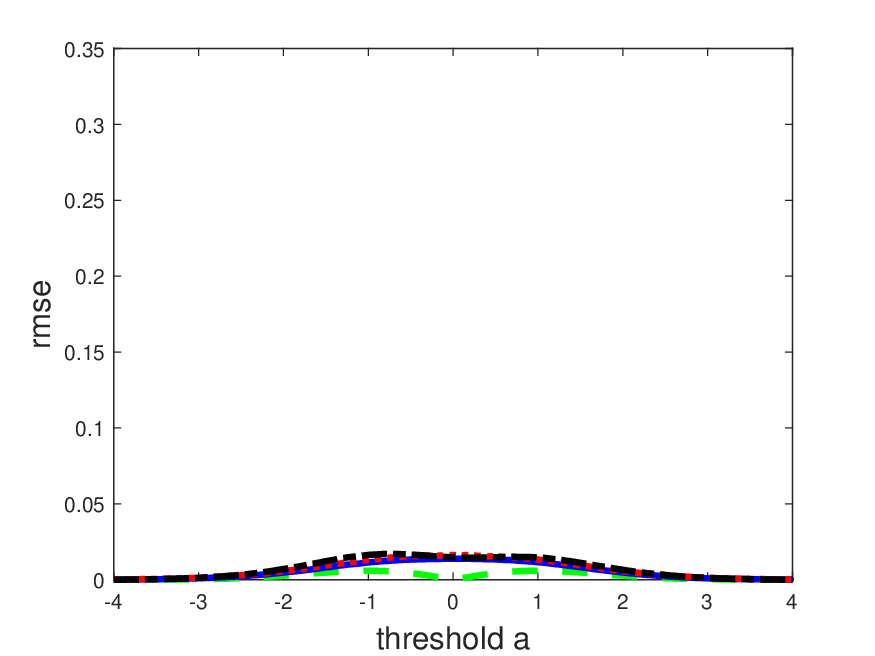}\\
			\multicolumn{3}{c}{(3) $\alpha_i\sim{iid}{\cal{N}}\left(0,1\right)$, $\frac{{\varepsilon}_{ij}}{s(\alpha_i)}\sim{iid}{\cal{N}}\left(0,1\right)$, $s(\alpha_i)\in\{.1,1.41\}$, $J=2$}\\
			\includegraphics[width=50mm, height=40mm]{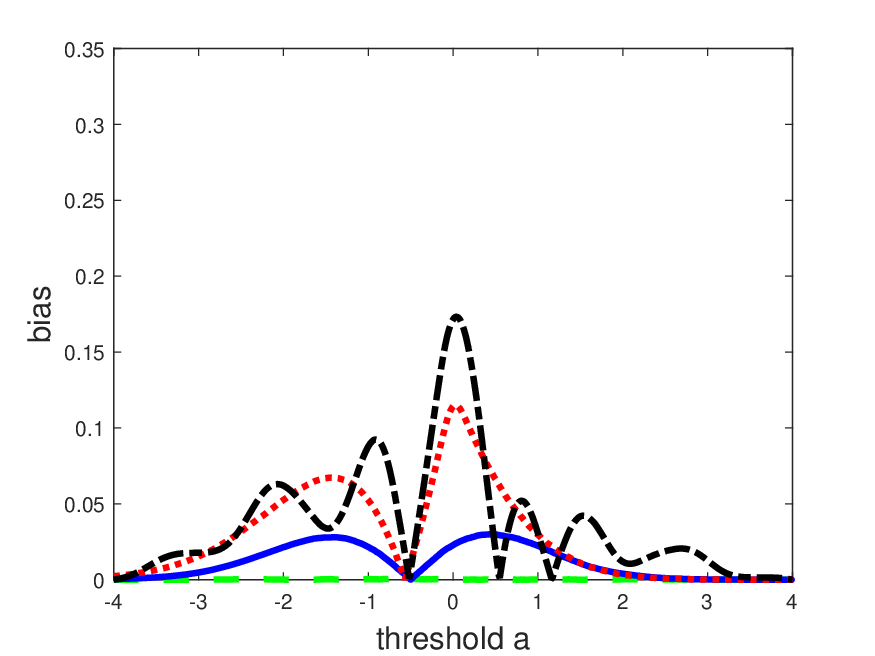}&	\includegraphics[width=50mm, height=40mm]{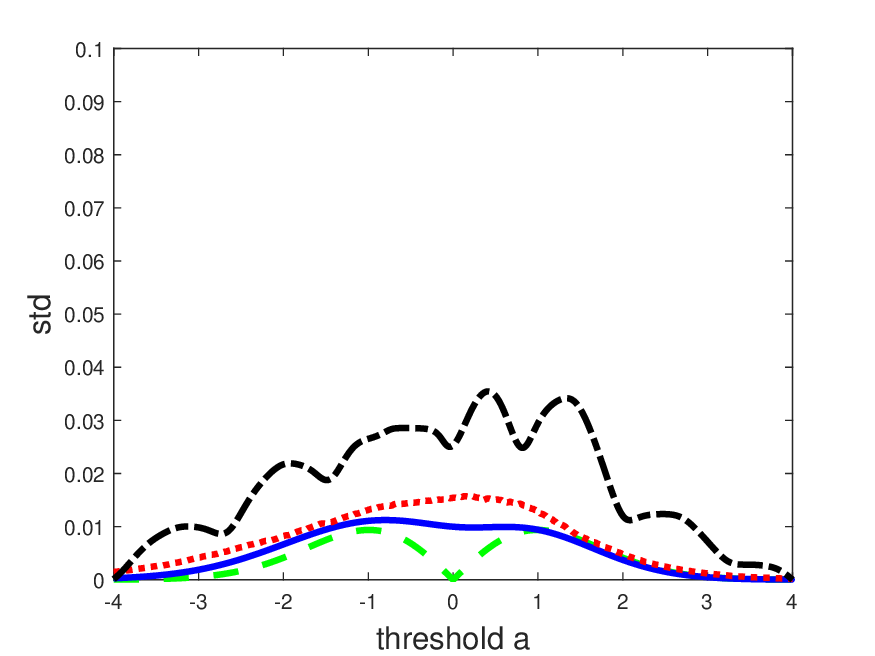}&	\includegraphics[width=50mm, height=40mm]{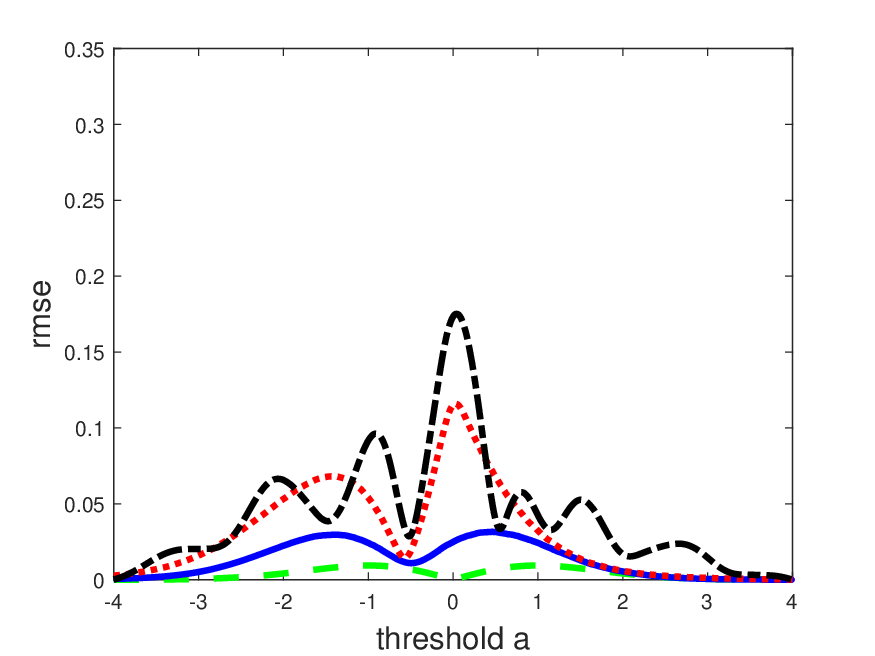}\\
			\multicolumn{3}{c}{(4) $\alpha_i\sim{iid}{\cal{N}}\left(0,1\right)$, $\frac{{\varepsilon}_{ij}}{s(\alpha_i)}\sim{iid}{\cal{N}}\left(0,1\right)$, $s(\alpha_i)\in\{.1,1.41\}$, $J=20$}\\
			\includegraphics[width=50mm, height=40mm]{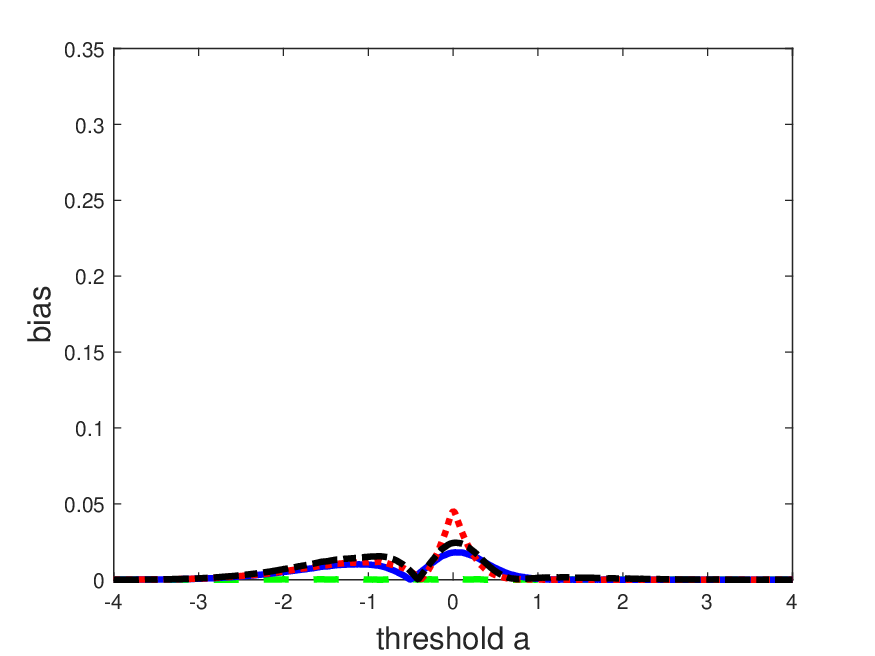}&	\includegraphics[width=50mm, height=40mm]{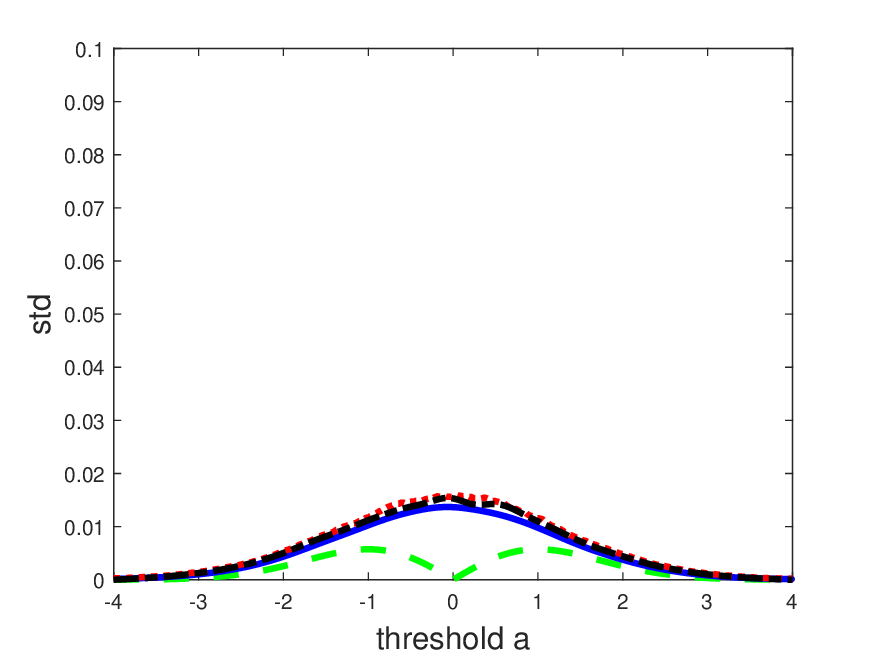}&	\includegraphics[width=50mm, height=40mm]{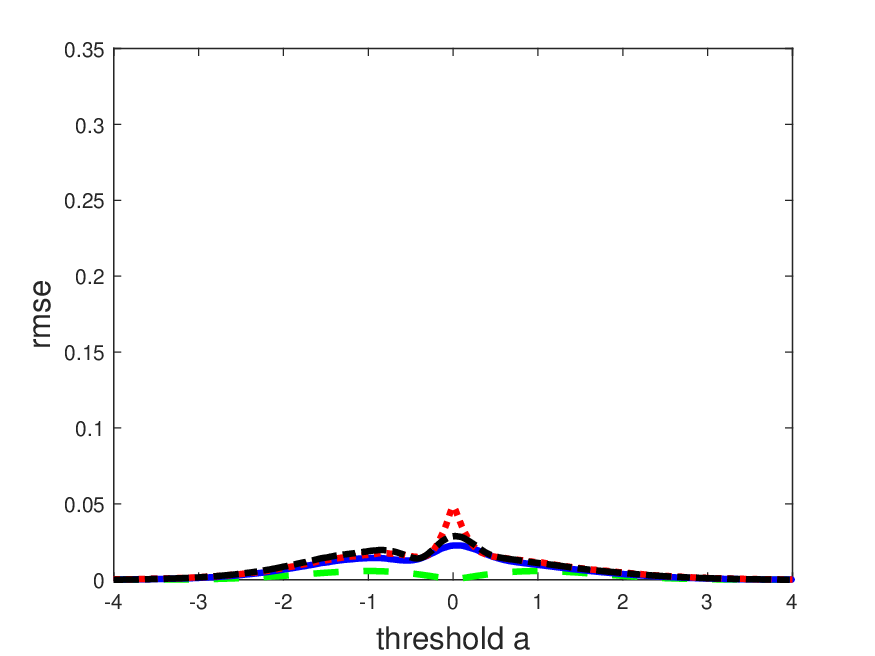}\end{tabular}
	\end{center}
	\par
	\textit{{\small Notes: 1000 simulations based on the fixed-effects model (\ref{FE_mod}). The left column shows the absolute bias, the middle column shows the standard deviation, and the right column shows the root-MSE, for four estimators: model-based (in {\color{green}{dashed}}), PAE (in {\color{blue}{solid}}), fixed-effects (in  {\color{red}{dotted}}), and nonparametric deconvolution (in \textbf{dash-dotted}).}}
\end{figure}

\begin{figure}[tbp]
	\caption{Monte Carlo results for $F_{\alpha}(a)=\mathbb{E}[\boldsymbol{1}\{\alpha \leq a\}]$ in the fixed-effects model $Y_{ij}=\alpha_i+{\varepsilon}_{ij}$, when the distribution of $\alpha$ is misspecified\label{fig_mc_misp}}
	\begin{center}
		\begin{tabular}{ccc}
			Bias&Std&RMSE \\
			\multicolumn{3}{c}{(1) $\alpha_i\sim {iid}\limfunc{Beta}\left(11,1\right)$ (rescaled), ${\varepsilon}_{ij}\sim{iid}{\cal{N}}\left(0,1\right)$, $J=2$}\\
			\includegraphics[width=50mm, height=40mm]{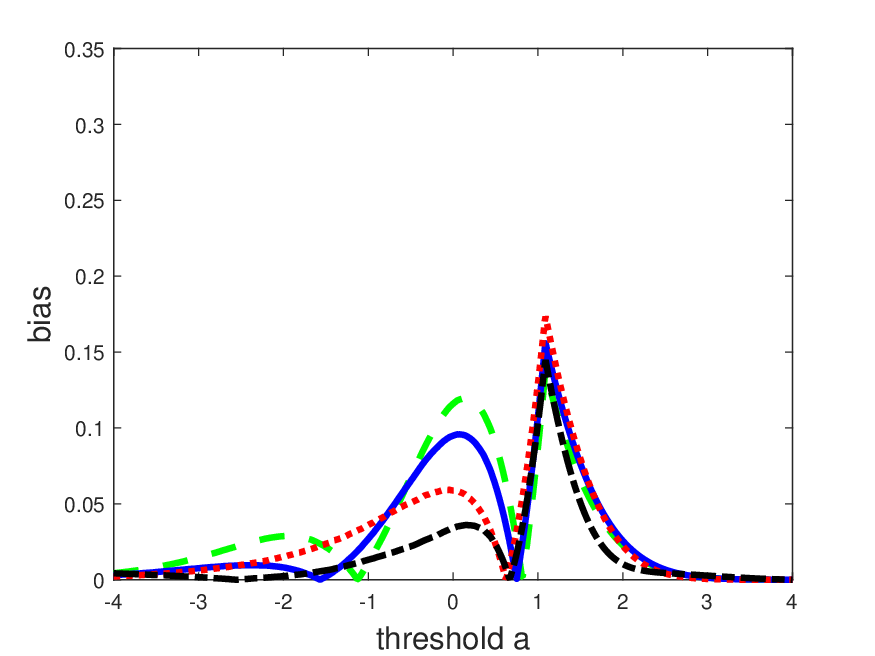}&	\includegraphics[width=50mm, height=40mm]{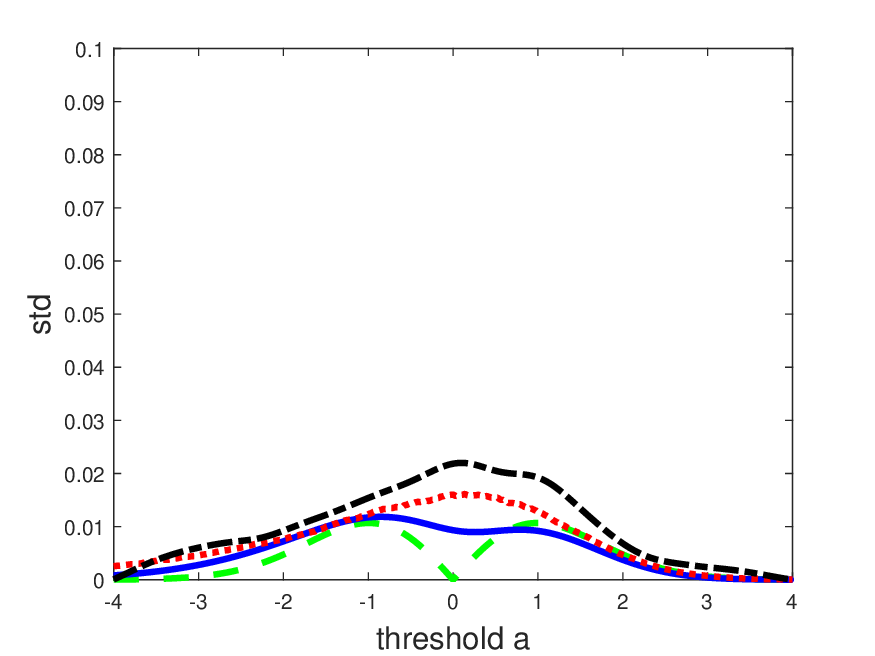}&	\includegraphics[width=50mm, height=40mm]{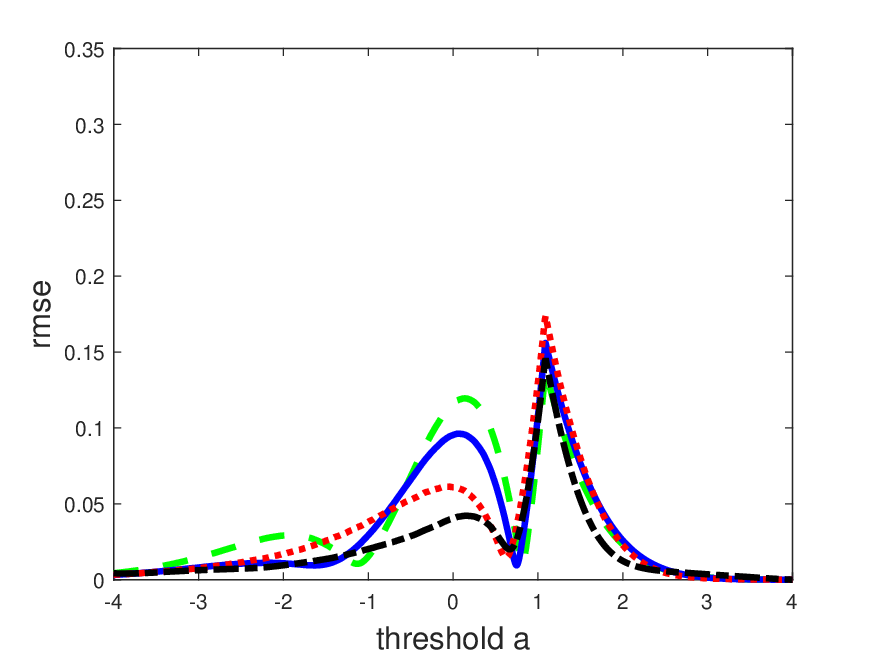}\\\multicolumn{3}{c}{(2) $\alpha_i\sim {iid}\limfunc{Beta}\left(11,1\right)$ (rescaled), ${\varepsilon}_{ij}\sim{iid}{\cal{N}}\left(0,1\right)$, $J=20$}\\
			\includegraphics[width=50mm, height=40mm]{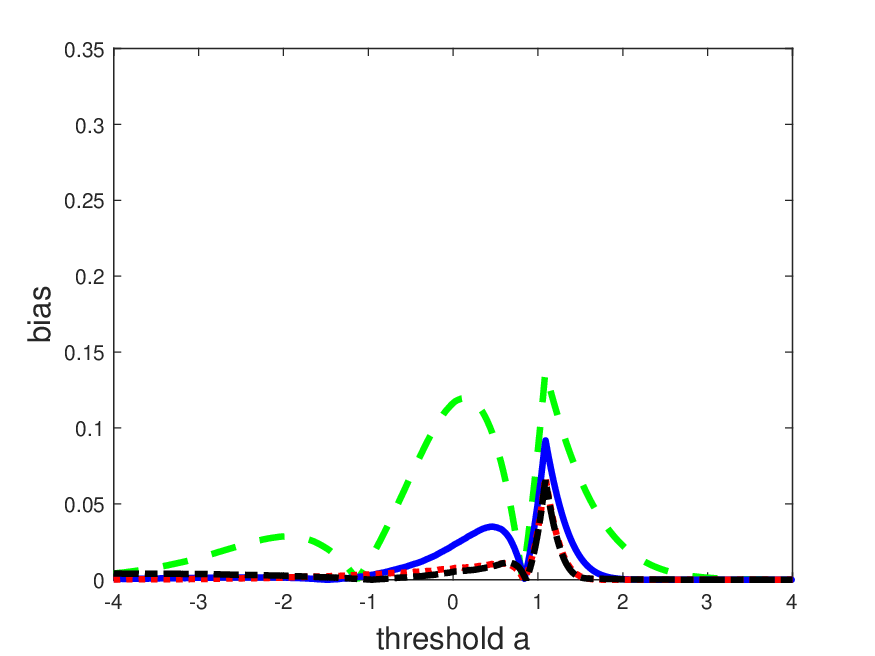}&	\includegraphics[width=50mm, height=40mm]{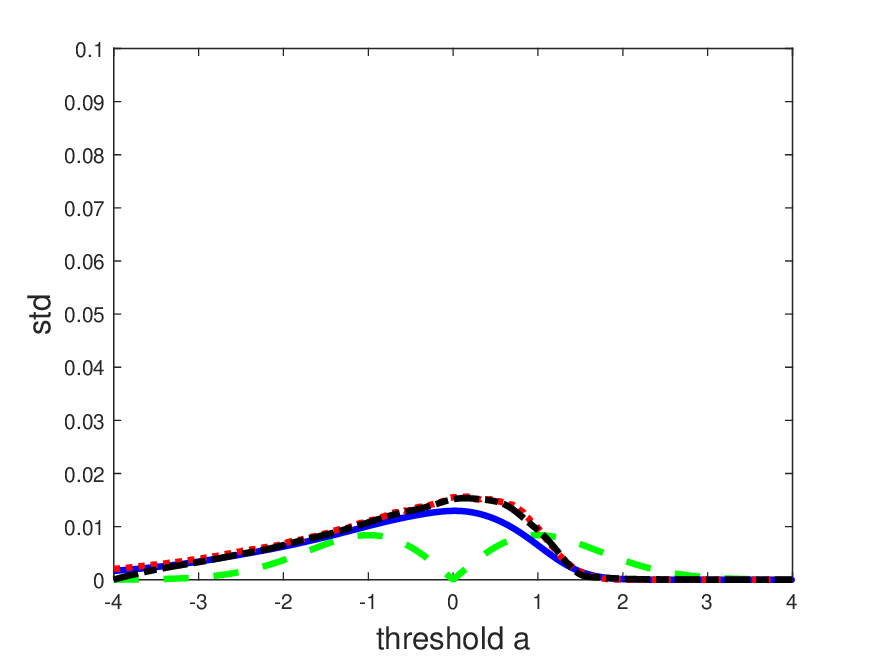}&	\includegraphics[width=50mm, height=40mm]{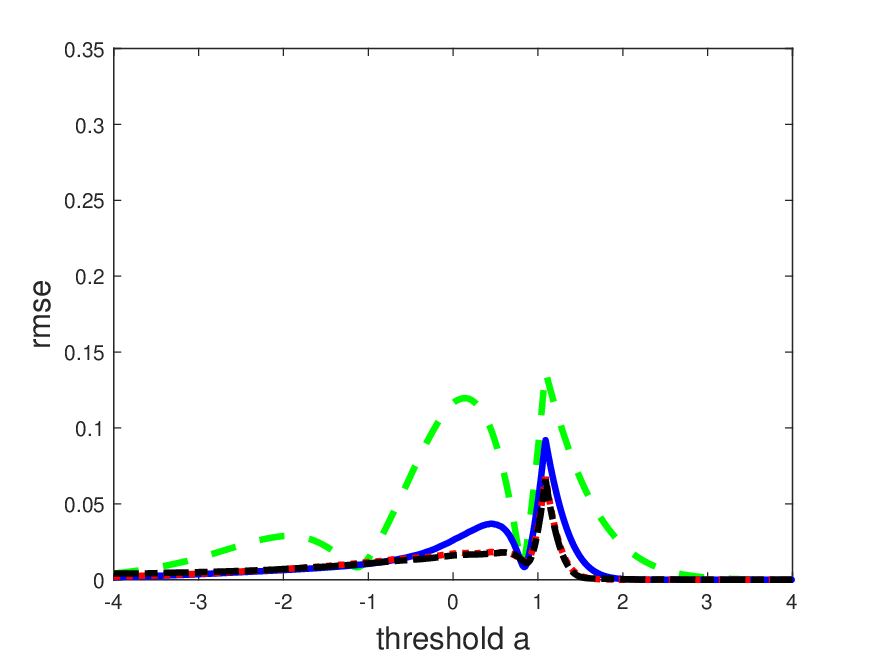}\\
			\multicolumn{3}{c}{(3) $\alpha_i\sim{iid} \limfunc{Beta}\left(11,1\right)$ (rescaled), $\frac{{\varepsilon}_{ij}}{s(\alpha_i)}\sim{iid}{\cal{N}}\left(0,1\right)$, $s(\alpha_i)\in\{.1,1.61\}$, $J=2$}\\
			\includegraphics[width=50mm, height=40mm]{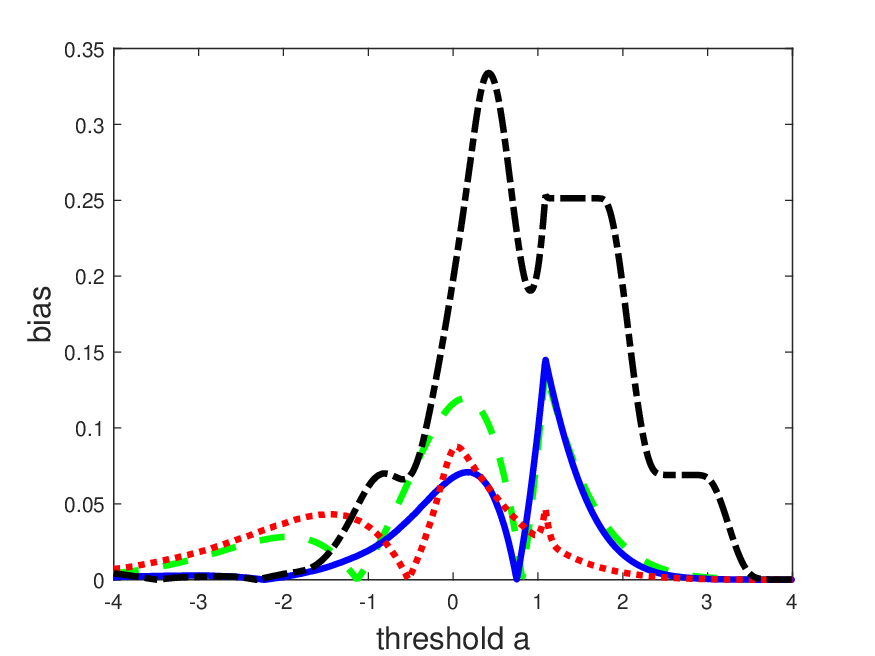}&	\includegraphics[width=50mm, height=40mm]{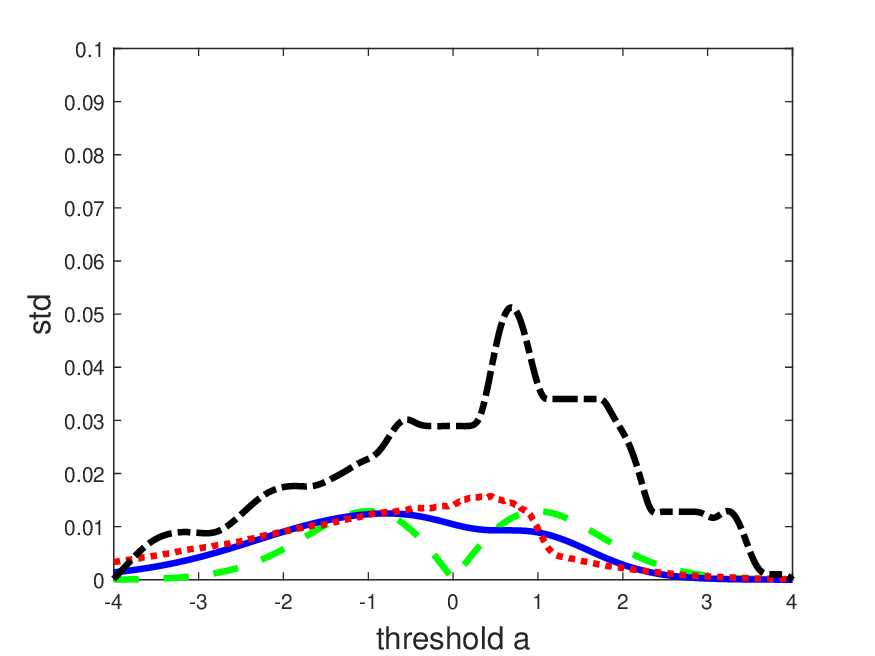}&	\includegraphics[width=50mm, height=40mm]{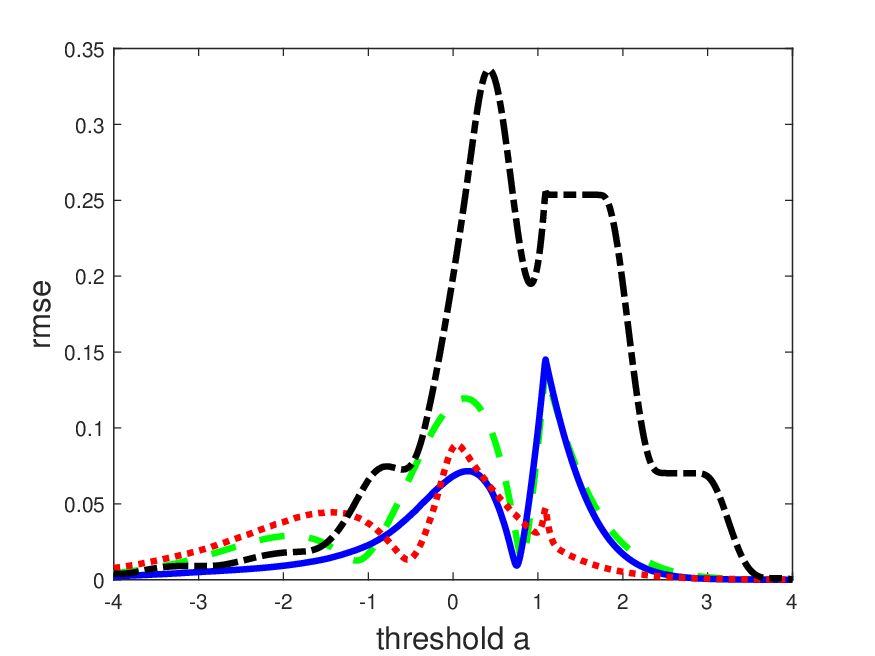}\\
			\multicolumn{3}{c}{(4) $\alpha_i\sim{iid} \limfunc{Beta}\left(11,1\right)$ (rescaled), $\frac{{\varepsilon}_{ij}}{s(\alpha_i)}\sim{iid}{\cal{N}}\left(0,1\right)$, $s(\alpha_i)\in\{.1,1.61\}$, $J=20$}\\
			\includegraphics[width=50mm, height=40mm]{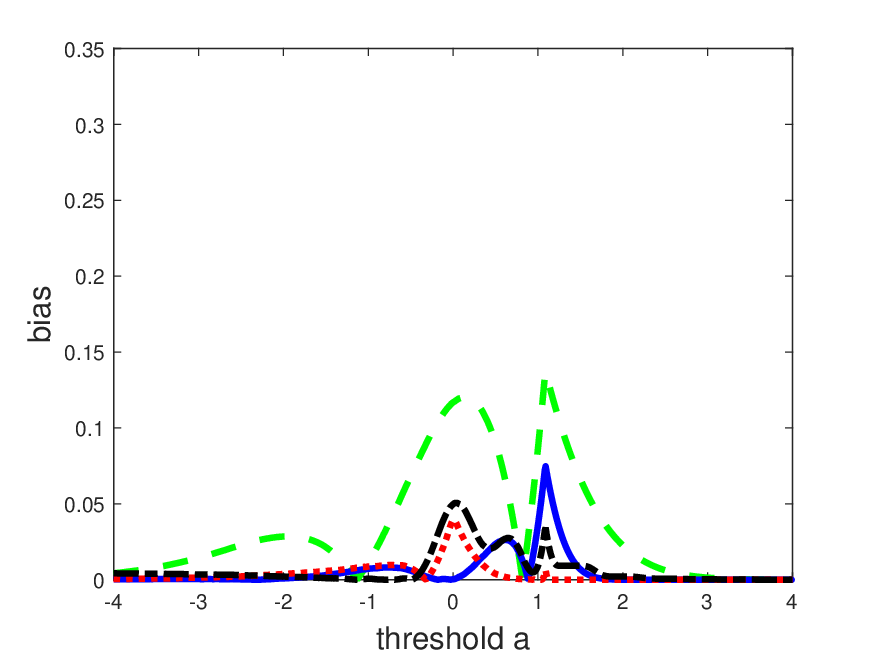}&	\includegraphics[width=50mm, height=40mm]{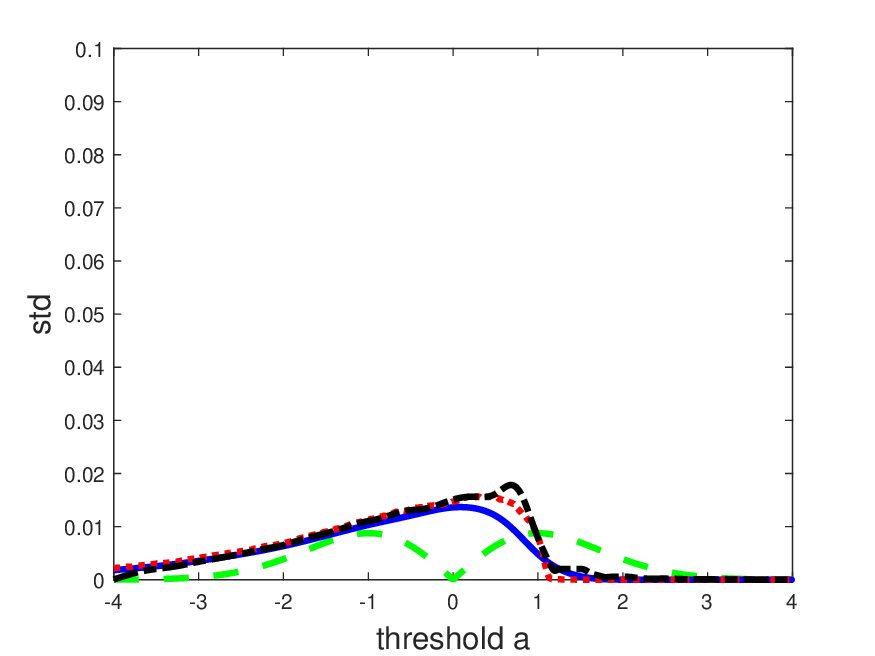}&	\includegraphics[width=50mm, height=40mm]{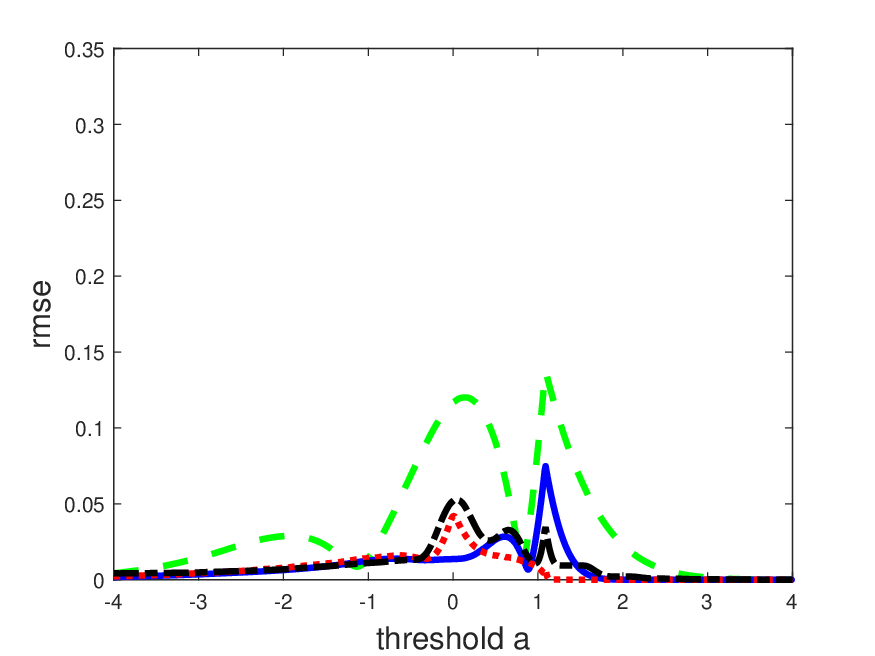}\end{tabular}
	\end{center}
	\par
	\textit{{\small Notes: 1000 simulations based on the fixed-effects model (\ref{FE_mod}). The left column shows the absolute bias, the middle column shows the standard deviation, and the right column shows the root-MSE, for four estimators: model-based (in {\color{green}{dashed}}), PAE (in {\color{blue}{solid}}), fixed-effects (in  {\color{red}{dotted}}), and nonparametric deconvolution (in \textbf{dash-dotted}).}}
\end{figure}

In the bottom two panels of Figure \ref{fig_mc}, $\varepsilon_{ij}$ are heteroskedastic, so for fixed $J$ as $n$ tends to infinity neither the PAE nor the deconvolution and fixed-effects estimators are consistent, yet the model-based estimator remains consistent in this case. We see that the PAE and the model-based estimator behave comparably to the case of DGP 2, with rather small biases and root-MSE, with smallest bias and root-MSE achieved by the model-based estimator. However, the performance of the deconvolution and fixed-effects estimators worsens relative to the homoskedastic case, especially when $J=2$. This suggests that the PAE is less sensitive to this particular form of misspecification than the deconvolution and fixed-effects estimators. 

We next turn to cases where the reference distribution of $\alpha_i$ is misspecified, which is the focus of our theory. In Figure \ref{fig_mc_misp}, we show the results of the simulations when $\alpha_i$ is distributed as a shifted and rescaled Beta with parameters $(11,1)$. In this case, the model-based estimator is substantially biased, as shown by the left column in Figure \ref{fig_mc_misp}, and the bias remains similar when varying $J$ and the distribution of $\varepsilon_{ij}$. In the top two panels of Figure \ref{fig_mc_misp}, $\varepsilon_{ij}$ are normal homoskedastic, so the deconvolution estimator is consistent as $n$ tends to infinity for fixed $J$. When $J=2$, the deconvolution estimator is biased, which is likely to reflect the ill-posedness of the estimation problem. When $J=20$, the bias and root-MSE of the deconvolution estimator are small. Interestingly, although it has no consistency guarantees in this DGP, the PAE performs relatively well. Indeed, while the PAE is biased for $J=2$, its variance is smaller than the one of the deconvolution estimator. In addition, when $J=20$ and the posterior conditioning is more informative, the performance of the PAE and deconvolution estimator improves, the latter still dominating the former. 

In the bottom two panels of Figure \ref{fig_mc_misp}, $\varepsilon_{ij}$ are heteroskedastic. We see that this form of misspecification has large effects on the performance of the deconvolution estimator, especially when $J=2$ so the  signal-to-noise ratio is lower. However, the performance of the PAE is very similar to the homoskedastic case: it is slightly less biased than the model-based estimator when $J=2$, and it is substantially less biased and has small root-MSE when $J=20$. In these designs the fixed-effects estimator and the PAE perform similarly.

\section{Extensions\label{App_Ext}}

In this section of the appendix we consider a number of issues in turn: how to compute PAE when they are not available in closed form, how to estimate quantities of interest that are nonlinear in $f_0$, whether the constant two appearing in Theorem \ref{theo_global} can be improved upon, how to construct confidence intervals, how to perform specification tests, how to derive the form of minimum-MSE estimators, and how to interpret PAE as Bayesian estimators in models where $U$ has finite support. 

\subsection{Computation}

$\widehat{\delta}^{\rm P}$ can be computed in closed form in simple models, such as all the examples in this paper. However, in complex models such as structural models, the likelihood function or posterior density may not be available in closed form. A simple approach in such cases is to proceed by simulation. 

Specifically, for all $i=1,...,n$ we first draw $U_i^{(s)}$, $s=1,...,S$ according to $f_{\widehat{\sigma}}(\cdot\,|\, X_i)$, and compute $Y_i^{(s)}=g_{\widehat{\beta}}(U_i^{(s)},X_i)$. Then, we regress $\delta_{\widehat{\beta}}(U_i^{(s)},X_i)$ on $Y_i^{(s)}$, for $s=1,...,S$. Any nonparametric/machine learning regression estimator can be used for this purpose. This procedure requires virtually no additional coding given simulation codes for outcomes and counterfactuals.

\subsection{Nonlinear effects}

The researcher may be interested in a nonlinear function of $f_0$. Specifically, here we abstract from covariates $X$ and focus on $\overline{\delta}=\varphi_{\beta}(f_0)$, for some functional $\varphi_{\beta}$. As an example, in the fixed-effects model (\ref{FE_mod}), $\overline{\delta}$ may be the Gini coefficient of $\alpha$. The analysis in the linear case applies verbatim to this case, since under regularity conditions
\begin{equation} \varphi_{\beta}(f_0)=\varphi_{\beta}(f_{\sigma_*})+\nabla \varphi_{\beta}(f_{\sigma_*})[f_0-f_{\sigma^*}]+o(\epsilon^{\frac{1}{2}}),\label{nonlinear_approx}\end{equation}
which is linear in $f_0$, up to smaller-order terms. Here $\nabla \varphi_{\beta}$ denotes the gradient of $\varphi_{\beta}(f)$ with respect to $f$. In Appendix \ref{sec_other_ex2} we report model-based and posterior estimates of Gini coefficients based on simulated data.

\subsection{The constant in Theorem \ref{theo_global}}

The binary choice model that we describe in Section \ref{sec_other_ex2} below is helpful to see that the global bound in Theorem \ref{theo_global}, which depends on the constant two, cannot be improved upon in general. To see this, consider the binary choice model (\ref{eq_binary_choice}) of Section \ref{sec_other_ex2} with three simplifications: $X$ consists of a single value, $b$ is known, and $\sigma_*=1$ is fixed. We assume that $x'b> X'b$. 

In this example, for $\epsilon$ large enough the worst-case blue specification errors of $\widehat{\delta}^{\rm M}$ and $\widehat{\delta}^{\rm P}$ are 
$${\limfunc{Bias}}_{\rm M}=\limfunc{max}(\Phi(x'b),1-\Phi(x'b)),$$
and
$${\limfunc{Bias}}_{\rm P}=\frac{\limfunc{max}(\Phi(x'b)-\Phi(X'b),1-\Phi(x'b))}{1-\Phi(X'b)},$$
respectively. 

From this, we first see that the specification error of the posterior estimator is smaller than twice that of the model-based estimator. In addition, taking $X'b=0$ and $x'b=\eta$, we have, for small $\eta$,
$$\frac{{\limfunc{Bias}}_{\rm P}}{{\limfunc{Bias}}_{\rm M}}=\frac{2(1-\Phi(\eta))}{\Phi(\eta)}\overset{\eta\rightarrow 0}{\rightarrow} 2.$$
This shows that two is indeed the smallest possible constant in Theorem \ref{theo_global}.

\subsection{Confidence intervals}

Consider first the correctly specified case. Suppose that $\widehat{\beta}$ and $\widehat{\sigma}$ are asymptotically linear in the sense that, for some mean-zero function $h$, we have
$$\left(\begin{array}{c}\widehat{\beta}\\\widehat{\sigma}\end{array}\right)=\left(\begin{array}{c}{\beta}\\{\sigma}_*\end{array}\right) +\frac{1}{n}\sum_{i=1}^n h(Y_i,X_i)+o_P(n^{-\frac{1}{2}}).$$
Then, under standard conditions (e.g., Newey and McFadden, 1994), we have
\begin{equation}
\label{eq_asympt}n^{\frac{1}{2}}\left(\begin{array}{c}\widehat{\delta}^{\rm M}-\overline{\delta}\\\widehat{\delta}^{\rm P}-\overline{\delta}\end{array}\right)\overset{d}{\rightarrow}{\cal{N}}\left(\left(\begin{array}{c} 0\\0\end{array}\right),\left(\begin{array}{cc}\Sigma_{11} & \Sigma_{12}\\\Sigma_{21} & \Sigma_{22}\end{array}\right)\right).
\end{equation}
Here,   $\Sigma_{11}=\limfunc{Var}_{*}\left(G_1' h(Y,X)+\mathbb{E}_{*}[\delta(U,X)\,|\, X]\right)$, 
$\Sigma_{12}=\limfunc{Cov}_{*}\big(G_1' h(Y,X)+\mathbb{E}_{*}[\delta(U,X)\,|\, X], \allowbreak
G_2' h(Y,X)+\mathbb{E}_{*}[\delta(U,X)\,|\,Y, X]\big)$, $\Sigma_{21}=\Sigma_{12}$, and 
$\Sigma_{22}=\limfunc{Var}_{*}\left(G_2' h(Y,X)+\mathbb{E}_{*}[\delta(U,X)\,|\, Y, X]\right)$,
for $G_1=\partial_{\beta,\sigma}\mathbb{E}_{\beta,\sigma_*}\,[ \delta_{\beta}(U,X)]$
and $G_2 = \mathbb{E}_{\beta,\sigma_*}\left\{ \partial_{\beta,\sigma}\mathbb{E}_{p_{\beta,\sigma_*}}\,[ \delta_{\beta}(U,X)\,|\,Y,X] \right\}$, where $\partial_\theta g(\theta_1)$ denotes the gradient of $g(\theta)$ at $\theta=\theta_1$. Note that in (\ref{eq_asympt}) we allow $\delta_{\beta}$ to be non-smooth in $\beta$ (e.g., an indicator function).

Consider next the locally misspecified case. A simple possibility to ensure uniform coverage within an $\epsilon$-neighborhood is to add $b_{\epsilon}(\gamma)$ on both sides of a standard confidence interval of $\overline{\delta}$. For example, one may construct the 95\% interval
$$\left[\widehat{\delta}^{\rm P}\pm \left(\epsilon^{{\frac{1}{2}}}\left\{\frac{2}{\phi''(1)}{\limfunc{Var}}_{*}\left(\delta(U,X)-\mathbb{E}_{*}[\delta(U,X)\,|\,Y,X]\right)\right\}^{\frac{1}{2}}+1.96n^{-\frac{1}{2}}\widehat{\Sigma}_{22}^{\frac{1}{2}}\right)\right],$$
for $\widehat{\Sigma}_{22}=\limfunc{Var}_{*}\left(G_2' h(Y,X)+\mathbb{E}_{*}[\delta(U,X)\,|\, Y, X]\right)$, where expectations and variances are taken with respect to $P(\widehat{\beta},f_{\widehat{\sigma}})$, and $\delta$, $G_2$, and $h$ are evaluated at $\widehat{\beta}$ and $\widehat{\sigma}$. Note that this confidence interval requires setting a value for $\epsilon$. Building on Hansen and Sargent (2008), Bonhomme and Weidner (2018) propose to interpret $\epsilon$ by relating it to the local power of a specification test.

\subsection{Specification test\label{App_spec}}

Using the asymptotic distribution of $(\widehat{\delta}^{\rm M},\widehat{\delta}^{\rm P})$ under correct specification of $f_{\sigma}$, we obtain
$$n^{\frac{1}{2}}\left(\widehat{\delta}^{\rm P}-\widehat{\delta}^{\rm M}\right)\overset{d}{\rightarrow}{\cal{N}}\left(0,\widetilde{\Sigma}\right),$$
where $\widetilde{\Sigma}=\limfunc{Var}_{*}\left(\mathbb{E}_{*}[\delta(U,X)\,|\, Y,X]-\mathbb{E}_{*}[\delta(U,X)\,|\, X] + (G_2-G_1)' h(Y,X)\right)$. Hence, under correct specification, $$n\left(\widehat{\delta}^{\rm P}-\widehat{\delta}^{\rm M}\right)'\widetilde{\Sigma}^{-1}\left(\widehat{\delta}^{\rm P}-\widehat{\delta}^{\rm M}\right)\overset{d}{\rightarrow} \chi^2_1.$$
Plugging-in a consistent empirical counterpart for $\widetilde{\Sigma}$ in this expression, we obtain a simple test of correct specification of the parametric density $f_{\sigma}$.

\subsection{Minimum local worst-case MSE estimator}
\label{app:MSE}

Here we explain why
$\widehat{\delta}^{\rm MMSE}$ in \eqref{deltaMMSE}
gives the estimator with minimum worst-case MSE, in a local asymptotic framework where $n$ tends to infinity, $\epsilon$ tends to zero, and $n\epsilon$ tends to a positive constant.
We only consider the case where $\beta$ and $\sigma_*$ are known and not estimated; that is, we have $\psi(y,x)=0$.
Then, finding $\gamma^{\rm MMSE}(y,x) $ such that $\widehat{\delta}^{\rm MMSE} $ minimizes worst-case MSE
 over $f_0 \in \Gamma_\epsilon$ can, to leading order in $\epsilon$ and $n^{-1}$, be shown 
to be equivalent to minimizing 
$$
   [b_{\epsilon}(\gamma)]^2 \, + \, \frac 1 n \, {\rm Var}_*[\gamma(Y,X)].
$$
See Bonhomme and Weidner (2018) for details.

Next, applying Lemma~\ref{lem_bias} and noting that $\mathbb{E}_{*}[\gamma(Y,X)-\delta(U,X)]=0$
is required for MSE minimization (since
adding a constant to $\gamma(y,x)$ such that  $\mathbb{E}_{*}[\gamma(Y,X)-\delta(U,X)]=0$
has no effect on the higher order bias terms in Lemma~\ref{lem_bias}, nor on ${\rm Var}_*[\gamma(Y,X)]$, it is optimal to eliminate the leading bias term $\mathbb{E}_{*}[\gamma(Y,X)-\delta(U,X)]$ in this way), we find that to leading order in $\epsilon$ and $n^{-1}$ the worst-case MSE reads
$$
   \frac{2 \,   \epsilon}{\phi''(1)} \,
   \mathbb{E}_{*}\left\{ {\rm Var}_*\left[  \gamma(Y,X)-  \delta(U,X) \, \big| \, X \right] \right\}
   +\frac 1 n \, \mathbb{E}_{*} \left\{ \gamma(Y,X) - \mathbb{E}_{*}[ \delta(U,X)] \right\}^2 .
$$
This expression for the approximate worst-case MSE depends on the distribution of $X$, which is unknown.
For the minimum local worst-case specification error result in Theorem~\ref{theo_bias}, it does not matter that the distribution of $X$
is unknown, because that distribution is identified from the sample as $n \rightarrow \infty$. However, for the MSE
result here we have to take a stand on how to deal with the randomness in the observed covariates.
In the following we {condition on the observed sample of covariates}, and replace all population
expectations over $X$ by sample averages over $X_i$, $i=1,\ldots,n$. We write $\widehat {\mathbb{E}}_X$ for those sample
averages. The worst-case MSE objective function in the last display then reads
$$
   \frac{2 \,   \epsilon}{\phi''(1)} \,
   \widehat {\mathbb{E}}_X    {\rm Var}_*\left[  \gamma(Y,X)-  \delta(U,X) \, \big| \, X \right]  
   +\frac 1 n \,  \widehat {\mathbb{E}}_X \, \mathbb{E}_{*}\left( \left\{ \gamma(Y,X) -  \widehat {\mathbb{E}}_X \mathbb{E}_{*}[ \delta(U,X)|X] \right\}^2 \, \Big| \, X\right) .
$$
By the law of total variance we have
\begin{align*}
  &  {\rm Var}_*\left[  \gamma(Y,X)-  \delta(U,X) \, \big| \, X \right] 
  \\
  &=  \mathbb{E}_{*}\left\{ {\rm Var}_*\left[  \gamma(Y,X)-  \delta(U,X) \, \big| \, Y, X \right]  \big| \, X \right\}
    +   {\rm Var}_*\left\{ \mathbb{E}_*\left[  \gamma(Y,X)-  \delta(U,X) \, \big| \, Y, X \right] \, \big| \, X \right\} 
  \\ 
  &=  \mathbb{E}_{*}\left\{ {\rm Var}_*\left[   \delta(U,X) \, \big| \, Y, X \right]  \big| \, X \right\}
    +   {\rm Var}_*\left\{ \mathbb{E}_*\left[  \gamma(Y,X)-  \delta(U,X) \, \big| \, Y, X \right] \, \big| \, X \right\} .
\end{align*}

In the following we can ignore the term $ \mathbb{E}_{*}\left\{ {\rm Var}_*\left[   \delta(U,X) \, \big| \, Y, X \right]  \big| \, X \right\}$, because 
it does not depend on $\gamma(y,x)$. Then, the leading approximation to the worst-case MSE is given by the sample average over $X$ of 
$$
   \frac{2 \,   \epsilon}{\phi''(1)} \,
  {\rm Var}_*\left\{  \gamma(Y,X)-  \mathbb{E}_*\left[  \delta(U,X) \, \big| \, Y, X \right] \, \big| \, X \right\} 
   +\frac 1 n \,  \mathbb{E}_{*}\left( \big\{ \gamma(Y,X) -  \widehat {\mathbb{E}}_X \mathbb{E}_{*}[ \delta(U,X)|X]  \big\}^2  \, \Big| \, X \right) .
$$

Clearly, if for any given $X=x$ we find $\gamma(y,x)$ that minimizes this objective function, then its expected value
over the sample distribution of $X$ is also minimized. 
The corresponding first-order condition for $ \gamma^{\rm MMSE}(Y,X)$ reads
\begin{align*}
 & \frac 1 n   \left\{  \gamma^{\rm MMSE}(y,x)  -  \widehat {\mathbb{E}}_X \mathbb{E}_{*}[ \delta(U,X)|X]  \right\} +    \frac{2 \,   \epsilon}{\phi''(1)} \,
    \Bigg\{ \gamma^{\rm MMSE}(y,x)  -   \mathbb{E}_*\left[  \delta(U,X) \, \big| \, Y=y, \, X=x \right]  
   \\ &   \qquad \qquad\qquad\qquad\qquad \qquad \qquad
    -  \mathbb{E}_*\left[ \gamma^{\rm MMSE}(Y,x) \, \big| \, X=x \right] 
    +   \mathbb{E}_*\left[  \delta(U,X) \, \big| \, X=x \right]  
    \Bigg\} = 0.
\end{align*}
The solution to this first-order condition is
\begin{align*}
\gamma^{\rm MMSE}(y,x) \,&= \, \frac{1}{n} \sum_{i=1}^n \mathbb{E}_{*}[ \delta(U,X)|X=X_i]  
\nonumber \\ & \quad
 +\left(1+\frac{\phi''(1)}{2n\epsilon}\right)^{-1}\Big\{
\mathbb{E}_*[\delta(U,X)\,|\, Y=y,X=x]-\mathbb{E}_*[\delta(U,X)\, |\, X=x]
\Big\},
\end{align*}
where we have now written $ \widehat {\mathbb{E}}_X$ as $\frac 1 n \sum_{i=1}^n$.

The corresponding minimum local MSE estimator for
$\overline \delta = \mathbb{E}_* \left[\delta_{\beta} (U,X)\right]$ is then given by
\begin{align*}
    \widehat{\delta}^{\rm MMSE}
     =\frac{1}{n} \sum_{i=1}^n \gamma^{\rm MMSE}(Y_i,X_i)  
    &= \left[ 1 - \left(1+\frac{\phi''(1)}{2n\epsilon}\right)^{-1} \right] \frac{1}{n} \sum_{i=1}^n \mathbb{E}_{*}[ \delta(U,X)|X_i]  
    \\ & \qquad
       + \left(1+\frac{\phi''(1)}{2n\epsilon}\right)^{-1}  \frac{1}{n} \sum_{i=1}^n 
       \mathbb{E}_*[\delta(U,X)\,|\, Y=Y_i,X=X_i] ,
\end{align*}
which is   the result stated in equation \eqref{deltaMMSE} of the main text.

\subsection{Finite support\label{App_Bayes}}

Here we consider the case where $U$ has finite support and takes the values $u_1,u_2,...,u_K$ with probability $\omega_1^0,...,\omega_K^0$. Here we abstract away from $\beta$, $\sigma$, and covariates $X$.

\paragraph{Injective and non-injective models.}

Let $\delta_k=\delta(u_k)$, and denote $g_k=g(u_k)$ where $Y=g(U)$. Let $\overline{g}_1,...,\overline{g}_L$ denote the $L\leq K$ equivalence classes of $g_1,...,g_K$. We will denote as $\ell(k)\in\{1,...,L\}$ the index corresponding to the equivalence class of $g_k$, for all $k$. In addition, let $n_{\ell}=\sum_{i=1}^n\boldsymbol{1}\{Y_i=\overline{g}_{\ell}\}$ for all $\ell$, and denote $\omega_k^U=f(u_k)$ for all $k$.

It is useful to distinguish two cases. When $g$ is {injective}, $K=L$ and $\mathbb{E}_{p(f)}[\delta(U)\,|\, g(U)=g_k]=\delta_k$. So we have $\widehat{\delta}^{\rm P}=\frac{1}{n}\sum_{k=1}^Kn_k\delta_k$. This estimator does not depend on the assumed $f$. Moreover, as $\limfunc{min}_{k=1,...,K} \, n_k$ tends to infinity we have
$$\widehat{\delta}^{\rm P}\overset{p}{\rightarrow} \sum_{k=1}^K\omega_k^0\delta_k=\overline{\delta}.$$
Hence $\widehat{\delta}^{\rm P}$ is consistent for $\overline{\delta}$, irrespective of the choice of the reference density $f$, provided $\omega_k^U>0$ for all $k$.

When $g$ is {not injective}, $K\neq L$ and we have
\begin{align*}\widehat{\delta}^{\rm P}=&\frac{1}{n}\sum_{i=1}^n\sum_{\ell=1}^L\boldsymbol{1}\{Y_i=\overline{g}_{\ell}\}\mathbb{E}_{p(f)}[\delta(U)\,|\, g(U)=\overline{g}_{\ell}]=\frac{1}{n}\sum_{\ell=1}^Ln_{\ell}\mathbb{E}_{p(f)}[\delta(U)\,|\, g(U)=\overline{g}_{\ell}].\end{align*}
Moreover,
\begin{eqnarray*}\mathbb{E}_{p(f)}[\delta(U)\,|\, g(U)=\overline{g}_{\ell}]&=\sum_{k=1}^{K}{\Pr}_{p(f)}(U=U_k\,|\, g(U)=\overline{g}_{\ell})\delta_{k}\\
	&=\sum_{k=1}^{K}\frac{\omega_{k}^U\boldsymbol{1}\{\ell(k)=\ell\}}{\sum_{k'=1}^{K}\omega_{k'}^U\boldsymbol{1}\{\ell(k')=\ell\}}\delta_{k}
	=: \overline{\delta}_{\ell}^{U}.\end{eqnarray*}
Hence,
$$\widehat{\delta}^{\rm P}=\frac{1}{n}\sum_{\ell=1}^Ln_{\ell}\overline{\delta}_{\ell}^{U}.$$
Through $\overline{\delta}_{\ell}^{U}$, $\widehat{\delta}^{\rm P}$ depends on the prior $\omega^U$ in general, even as $\limfunc{min}_{\ell=1,...,L} \, n_{\ell}$ tends to infinity.

\paragraph{Bayesian interpretation.}

From a Bayesian perspective, one may view $\omega^0$ as a parameter, and put a prior on it. A simple conjugate prior specification is a Dirichlet distribution $\omega \sim \text{Dir}(K,\alpha)$, where $\alpha_k>0$ for $k=1,...,K$. We will focus on the posterior mean
$$\widehat{\delta}^{\rm D}=\mathbb{E}\left[\sum_{k=1}^K\delta_k \omega_k \,|\, Y\right]=\sum_{k=1}^K\delta_k\mathbb{E}\left[\omega_k \,|\, Y\right],$$
for a Dirichlet prior with $\alpha_k=M\omega_k^{U}$ for all $k$, where $M>0$ is a constant.

For all $\ell$, let $\overline{\alpha}_{\ell}=\sum_{k=1}^K\boldsymbol{1}\{\ell(k)=\ell\}\alpha_k$, and $\overline{\omega}_{\ell}=\sum_{k=1}^K\boldsymbol{1}\{\ell(k)=\ell\}\omega_k$. $(\overline{\omega}_{1},...,\overline{\omega}_{L})$ follows the Dirichlet distribution $\text{Dir}(L,\overline{\alpha})$. Moreover, for all $k$, $\omega_k/\overline{\omega}_{\ell(k)}$ is a component of a Dirichlet distribution with mean $\alpha_k/\overline{\alpha}_{\ell(k)}$. 

Unlike the $\overline{\omega}_{\ell}$'s, the $\omega_k/\overline{\omega}_{\ell(k)}$'s are not updated in light of the data since they do not enter the likelihood. Notice the link with the Bayesian analysis of partially identified models in Moon and Schorfheide (2012): here the $\overline{\omega}_{\ell}$'s are identified but the $\omega_k$'s are not, since for identical $g_k$'s the data provides no information to discriminate across $\omega_k$'s.

As a result, we have
\begin{align*}\mathbb{E}[\omega_k \,|\, Y]&=\mathbb{E}\left[\frac{\omega_k}{\overline{\omega}_{\ell(k)}} \overline{\omega}_{\ell(k)}\,|\, Y\right]=\mathbb{E}\left[\frac{\omega_k}{\overline{\omega}_{\ell(k)}} \right]\mathbb{E}\left[ \overline{\omega}_{\ell(k)}\,|\, Y\right]\\&=\frac{\alpha_k}{\overline{\alpha}_{\ell(k)}}\frac{n_{\ell}+\overline{\alpha}_{\ell}}{n+M}\overset{M\rightarrow 0}{\rightarrow} \frac{\omega_k^U}{\sum_{k'=1}^{K}\omega_{k'}^U\boldsymbol{1}\{\ell(k')=\ell(k)\}}\frac{n_{\ell(k)}}{n}.\end{align*}
It thus follows that
\begin{align*}\widehat{\delta}^{\rm D}\overset{M\rightarrow 0}{\rightarrow} \sum_{k=1}^K\delta_k\frac{\omega_k^U}{\sum_{k'=1}^{K}\omega_{k'}^U\boldsymbol{1}\{\ell(k')=\ell(k)\}}\frac{n_{\ell(k)}}{n}=\widehat{\delta}^{\rm P}.\end{align*}
Hence, under a diffuse Dirichlet prior centered around $\omega^U$, the Bayesian posterior mean coincides with the PAE.

\section{Posterior average effects in various settings\label{sec_other_ex2}}

In this section, we provide additional examples of models where PAE may be of interest, and we show illustrative simulations for two models.

\subsection{Models}

\paragraph{Linear regression.}

Consider the linear regression
$$Y_i=X_i'b+U_i.$$
Suppose that $\mathbb{E}[XU]=0$, and that the OLS estimator $\widehat{b}$ is consistent for $b$. Suppose also that the researcher is interested in the average effect $\overline{\delta}=\mathbb{E}_{f_0}[U^2XX']$. In this example $\overline{\delta}$ is multi-dimensional; see Appendix \ref{App_Ext}.

In this context, a model-based approach consists in modeling $U\,|\, X$, say, as a normal with zero mean and variance $s^2$, and computing 
$$\widehat{\delta}^{\rm M}=\widehat{s}^2\frac{1}{n}\sum_{i=1}^nX_iX_i',$$
where $\widehat{s}^2=\frac{1}{n}\sum_{i=1}^n(Y_i-X_i'\widehat{b})^2$ is the maximum likelihood estimator of $s^2$ under normality.

By contrast, a PAE is 
\begin{align*}
\widehat{\delta}^{\rm P}&=\frac{1}{n}\sum_{i=1}^n\mathbb{E}_{p_{\widehat{b},\widehat{s}}} \left[U^2XX'\,\big|\, Y=Y_i,X=X_i\right]\\
&=\frac{1}{n}\sum_{i=1}^n(Y_i-X_i'\widehat{b})^2X_iX_i'.
\end{align*}
This is the central piece in the White (1980) variance formula. $\widehat{\delta}^{\rm P}$ remains consistent for $\overline{\delta}$ absent normality or homoskedasticity of $U$. In this very special case, $\widehat{\delta}^{\rm P}$ is thus fully robust to misspecification of $f_{s}$, since $U_i$ is a deterministic
function of $Y_i$, $X_i$ and $b$.

\paragraph{Censored regression.}

Consider next the censored regression model
\begin{equation}Y_i=\limfunc{max}(Y_i^*,0),\text{  where } Y_i^*=X_i'b+U_i.\label{eq_censored}\end{equation}
In this model, $b$ can be consistently estimated under weak conditions. For example, Powell's (1986) symmetrically trimmed least-squares estimator is consistent for $b$ when $U\,|\, X$ is symmetric around zero, under suitable regularity conditions. In this setting, suppose that we are interested in a moment of the potential outcomes $Y_i^*$, such as $\overline{\delta}=\mathbb{E}_{f_0}[h(Y^*)]$ for some function $h$. As an example, the researcher may wish to estimate a feature of the distribution of wages using a sample affected by top- or bottom-coding. 

Following a model-based approach, let us assume that $U\,|\, X\sim {\cal{N}}(0,s^2)$, and estimate ${s}^2$ using maximum likelihood. A model-based estimator is then $\widehat{\delta}^{\rm M}=\frac{1}{n}\sum_{i=1}^n \mathbb{E}_{f_{\widehat{s}}} [h(X_i'\widehat{b}+U)]$. By contrast, a PAE is 
\begin{align*}
\widehat{\delta}^{\rm P}
& =\frac{1}{n}\sum_{i=1}^n\underset{\text{uncensored}}{\underbrace{\boldsymbol{1}\{Y_i>0\}h(Y_i)}}+\frac{1}{n}\sum_{i=1}^n\underset{\text{censored}}{\underbrace{\boldsymbol{1}\{Y_i=0\}\mathbb{E}_{p_{\widehat{b},\widehat{s}}}\left[h(X_i'\widehat{b}+U)\,\big|\,X_i'\widehat{b}+U\leq 0\right]}}.
\end{align*}
This estimator relies on actual $Y$'s for uncensored observations, and on imputed $Y$'s for censored ones.

The censored regression model illustrates an aspect related to the class of neighborhoods that our theoretical characterizations rely on. In model (\ref{eq_censored}), the researcher might want to impose that $U\,|\,X $ be symmetric around zero, which is the main assumption for consistency of the Powell (1986) estimator. It is possible to construct estimators that minimize local worst-case specification error in an $\epsilon$-neighborhood that only consists of symmetric densities $f_0$. However, PAE may no longer have minimum specification error in this class. More generally, the assumptions that justify the use of a particular estimator $\widehat{\beta}$ may suggest further restrictions on the neighborhood. Our worst-case specification error results are based on a class where such restrictions are not imposed. Indeed, the only additional restriction on $f_0$, beyond belonging to an $\epsilon$-neighborhood around $f_{\sigma_*}$, is that the population moment condition $\mathbb{E}_{P(\beta,f_0)} [\psi_{\beta,\sigma_{*}}(Y,X)]=0$ is assumed to hold,  and we do not impose further restrictions that might be natural in order to justify the validity of this moment condition.

\paragraph{Binary choice.}

Consider now the binary choice model
\begin{equation}
\label{eq_binary_choice}Y_i=\boldsymbol{1}\{X_i'b+U_i>0\}.
\end{equation}
In this model, Manski (1975, 1985) shows that $b$ is identified up to scale as soon as the median of $U\,|\, X$ is zero, under sufficiently large support of $X$. In addition, he provides conditions for consistency of the maximum score estimator $\widehat{b}$, again up to scale. Manski's conditions, however, are not sufficient to consistently estimate the average structural function (ASF, Blundell and Powell, 2004)
$$\overline{\delta}(x)=\mathbb{E}_{f_0}[\boldsymbol{1}\{x'b+U>0\}].$$

Let us take as reference parametric distribution for $U\,|\, X$ a normal with zero mean and variance ${s}^2$, and let $\widehat{s}^2$ denote the maximum likelihood estimator of $s^2$ given $\widehat{b}$, based on normality. Specifically, $\widehat{s}$ maximizes the probit log-likelihood $\sum_{i=1}^nY_i\log \Phi\left(\frac{X_i'\widehat{b}}{s}\right)+(1-Y_i)\log \Phi\left(-\frac{X_i'\widehat{b}}{s}\right)$. A model-based estimator of the ASF is $\widehat{\delta}^{\rm M}(x)=\Phi\left(\frac{x'\widehat{b}}{\widehat{s}}\right)$, and a posterior estimator is
\begin{align*}\widehat{\delta}^{\rm P}(x)=
& \, \frac{1}{n}\sum_{i=1}^n \left[ Y_i \, \frac{\min\left(\Phi\left(\frac{x'\, \widehat{b}}{\widehat{s}}\right),\Phi\left(\frac{X_i'\, \widehat{b}}{\widehat{s}}\right)\right)}{\Phi\left(\frac{X_i'\, \widehat{b}}{\widehat{s}}\right)}+(1-Y_i)\, \frac{\limfunc{max}\left(\Phi\left(\frac{x'\,\widehat{b}}{\widehat{s}}\right)-\Phi\left(\frac{X_i'\,\widehat{b}}{\widehat{s}}\right),0\right)}{1-\Phi\left(\frac{X_i'\,\widehat{b}}{\widehat{s}}\right)} \right] . 
\end{align*}
Unlike $\widehat{\delta}^{\rm M}(x)$, the posterior ASF estimator $\widehat{\delta}^{\rm P}(x)$ depends directly on the observations of the binary $Y_i$'s, in addition to the indirect data dependence through $\widehat{b}$ and $\widehat{s}^2$. In the next subsection we present simulations from an ordered choice model, which suggest that the informativeness of the posterior conditioning --- and the robustness properties of posterior estimators compared to model-based estimators --- depend crucially on the support of the dependent variable.

\paragraph{Panel data discrete choice.}

Our last example is the panel data model
$$Y_{it}=\boldsymbol{1}\{X_{it}'b+\alpha_i+\varepsilon_{it}>0\},\quad i=1,...,n,\quad t=1,...,T.$$
When $\varepsilon_{it}$ are i.i.d. standard logistic, $b$ can be consistently estimated using the conditional logit estimator (Andersen, 1970, Chamberlain, 1984). However, additional assumptions are needed to consistently estimate average partial effects such as the effect of a discrete shift of $\Delta $ along the $k$-th component of $X$,
$$\overline{\delta}=(\mathbb{E}_{f_0}[\boldsymbol{1}\{(X_t+\Delta  \cdot e_k)'b+\alpha+\varepsilon_{t}>0\}]-\mathbb{E}_{f_0}[\boldsymbol{1}\{X_t'b+\alpha+\varepsilon_{t}>0\}])/\Delta  ,$$
where $e_k$ is a vector of zeros with a one in the $k$-th position.  

The standard approach is to postulate a parametric random-effects specification for the conditional distribution of $\alpha$ given $X_1,...,X_T$, and to compute an average effect $\widehat{\delta}^{\rm M}$ with respect to that distribution. By contrast, a posterior estimator is computed conditional on the observations $Y_{i1},...,Y_{iT}$, for every individual $i$. As $T$ tends to infinity, such estimators are robust to misspecification of $\alpha$, provided $\varepsilon_t$ is correctly specified (Arellano and Bonhomme, 2009). Our analysis shows that they also have robustness properties when $T$ is fixed and $n$ tends to infinity.

Aguirregabiria \textit{et al.} (2018) show that conditional logit-like estimators can also be used to consistently estimate parameters in structural dynamic discrete choice settings. As an example, they study the Rust (1987) model of bus engine replacement in the presence of unobserved heterogeneity in maintenance and replacement costs. In such structural models, estimating average welfare effects of policies requires averaging with respect to the distribution of unobservables. PAE provide an alternative to the standard parametric model-based approach in this context.

\subsection{Simulations\label{App_Simu}}

Here we report the results of two simulation exercises, based on the fixed-effects model (\ref{FE_mod}), and on an ordered choice model.

\subsubsection{Fixed-effects model}

\paragraph{Skewness.}

Let us consider the fixed-effects model (\ref{FE_mod}). Suppose the parameter of interest is the skewness of $\alpha$ 
$$\overline{\delta}=\mathbb{E}_{f_0} \left[\alpha^3-3\frac{{\mu}_{\alpha}}{{s}_{\alpha}}-\left(\frac{{\mu}_{\alpha}}{{s}_{\alpha}}\right)^3\right].$$
For example, it is of interest to estimate the skewnesses of income components and how they evolve over time (Guvenen \textit{et al.}, 2014). Since the normal distribution is symmetric, the model-based normal estimator of skewness is simply $\widehat{\delta}^{\rm M}=0$, irrespective of the observations $Y_{ij}$. Hence, $\widehat{\delta}^{\rm M}$ is not informed by the data, even when the empirical distribution of the fixed-effects $\overline{Y}_i=\frac{1}{J}\sum_{j=1}^JY_{ij}$ indicates strong asymmetry. 

By contrast, a PAE based on a normal reference distribution is
$$\widehat{\delta}^{\rm P}=\frac{1}{\widehat{s}_{\alpha}^3}\frac{1}{n}\sum_{i=1}^n \mathbb{E}_{p(f_{\widehat{s}})}\left[\alpha^3\,\big|\, Y=Y_i\right]-3\frac{\widehat{\mu}_{\alpha}}{\widehat{s}_{\alpha}}-\left(\frac{\widehat{\mu}_{\alpha}}{\widehat{s}_{\alpha}}\right)^3.$$
It can be verified that
$$\widehat{\delta}^{\rm P}=\widehat{\rho}^3\frac{1}{\widehat{s}_{\alpha}^3}\frac{1}{n}\sum_{i=1}^n \left(\overline{Y}_i-\overline{Y}\right)^3,$$
where $\widehat{\rho}=\frac{\widehat{s}_{\alpha}^2}{\widehat{s}_{\alpha}^2+\widehat{s}_{\varepsilon}^2/J}$. Under mild conditions, and in contrast with $\widehat{\delta}^{\rm M}$, the posterior estimator $\widehat{\delta}^{\rm P}$ is consistent for the true skewness of $\alpha$ as $J$ tends to infinity. However, $\widehat{\delta}^{\rm P}$ is biased for small $J$ in general.

To provide intuition about the magnitude of the bias, we simulate data where all latent components are independent, $\varepsilon_j$ are standard normal, and $\alpha$ follows a skew-normal distribution (e.g., Azzalini, 2013) with zero mean, variance 1, and skewness $\approx .47$ corresponding to the skew-normal parameter $\delta=.99$. We take $n=1000$, and run $100$ simulations varying $J$ from 1 to 30. We estimate means and variances using minimum-distance based on first and second moment restrictions.

In the left panel of Figure \ref{fig_skew} we show the results. We see that the model-based estimator is equal to zero irrespective of the number $J$ of individual measurements. By contrast, the posterior estimator converges to the true skewness of $\alpha$ as $J$ increases, although it is biased for small $J$.

\begin{figure}
	\caption{Skewness and Gini estimates in the fixed-effects model\label{fig_skew}}
	\begin{center}
		\begin{tabular}{cc}
			Skewness & Gini\\
			\includegraphics[width=70mm, height=50mm]{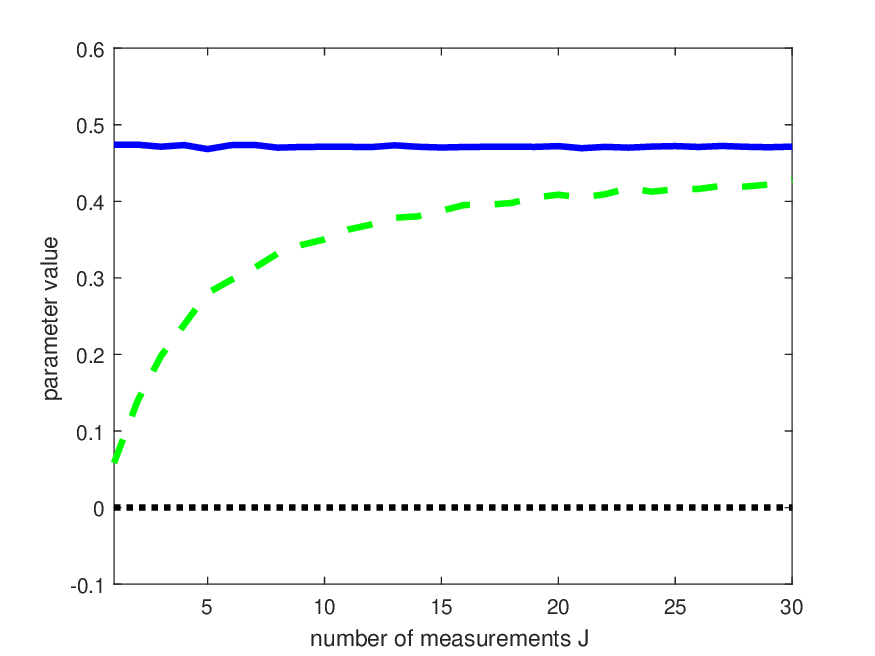} & 	\includegraphics[width=70mm, height=50mm]{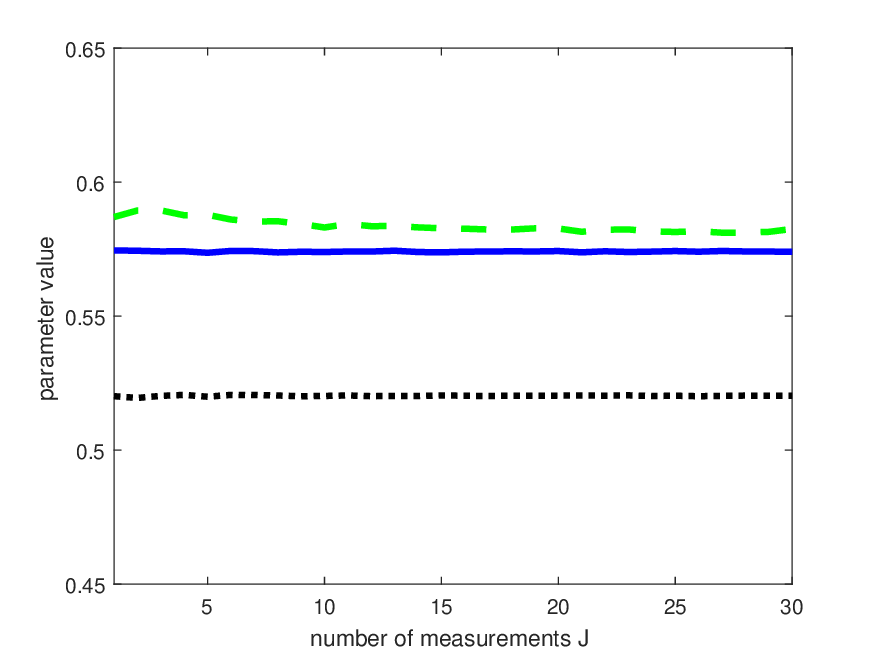}		
		\end{tabular}%
	\end{center}
	\par
	\textit{{\small Notes: true (solid), posterior (dashed), model-based (dotted). $n=1000$, $100$ simulations.}}
\end{figure}

\paragraph{Gini coefficient.}

We next focus on the Gini coefficient of $\alpha$: 
$$G=\frac{1}{2\mathbb{E}_{f_0}[\exp\left(\alpha\right)]}\iint |\exp(\alpha')-\exp(\alpha)|f_0(\alpha)f_0(\alpha')d\alpha d\alpha'.$$ 
In this case, a model-based estimator is
$$\widehat{G}^{\rm M}=2\Phi(\widehat{s}_{\alpha}/\sqrt{2})-1,$$
while a PAE is, following (\ref{nonlinear_approx}), 
$$\widehat{G}^{\rm P}=\widehat{G}^{\rm M}+\frac{1}{n}\sum_{i=1}^n\left(\mathbb{E}[\nabla \widehat{G}(\alpha)\,|\, Y_i]-\mathbb{E}[\nabla \widehat{G}(\alpha)]\right),$$
where
\begin{align*} \nabla\widehat{G}(\alpha)&=-\exp\left(\alpha-\widehat{\mu}_{\alpha}-\frac{1}{2}\widehat{s}_{\alpha}^2\right)\left(\widehat{G}^{\rm M}+1-2\Phi\left(\frac{\alpha-\widehat{\mu}_{\alpha}}{\widehat{s}_{\alpha}}\right)\right) +\left(1-2\Phi\left(\frac{\alpha-\widehat{\mu}_{\alpha}}{\widehat{s}_{\alpha}}-\widehat{s}_{\alpha}\right)\right).\end{align*}

In the right panel of Figure \ref{fig_skew} we show the simulation results. We see that in this case also the model-based estimator is insensitive to $J$. The posterior estimator has a lower bias, especially for larger $J$.

\subsubsection{Ordered choice model}

We next consider the ordered choice model
$$Y_i=\sum_{j=1}^J j\boldsymbol{1}\{\mu_{j-1}\leq Y_i^*\leq \mu_j\},\text{  where }Y_i^*=X_i'b+U_i,$$
for a sequence of {known} thresholds $-\infty=\mu_0<\mu_1<...<\mu_{J-1}<\mu_J=+\infty$. This model may be of interest to analyze data on wealth or income, say, where only a bracket containing the true observation is recorded. We focus on the average structural function 
$$ \overline{\delta}(x)=\mathbb{E}_{f_0}\left[\sum_{j=1}^J j\boldsymbol{1}\{\mu_{j-1}\leq x'b+U\leq \mu_j\}\right].$$

\begin{figure}
	\caption{Average structural function in the ordered choice model\label{fig_ordered}}
	\begin{center}
		\begin{tabular}{cc}
			$J=3$ & $J=10$\\
			\includegraphics[width=70mm, height=50mm]{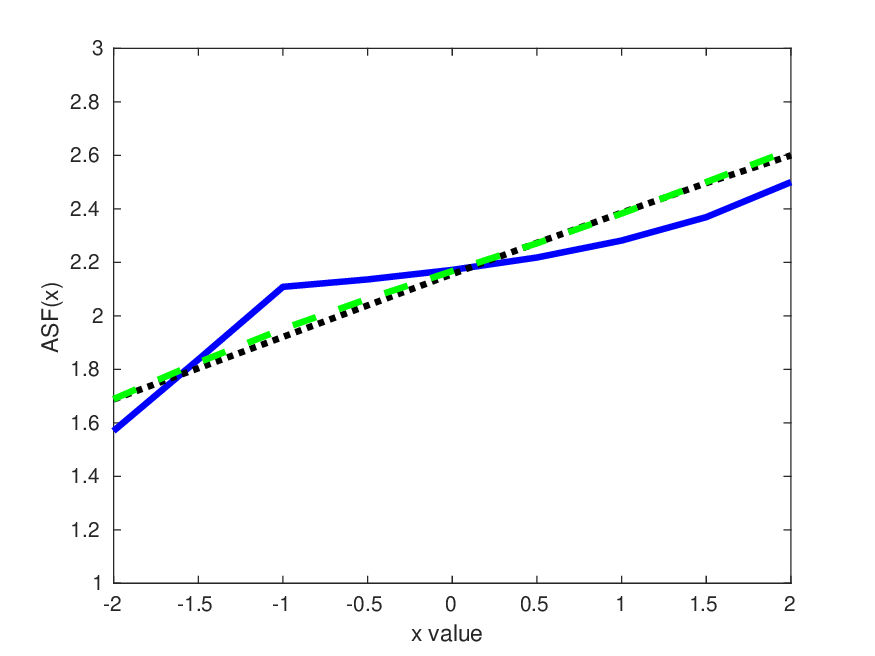} &	\includegraphics[width=70mm, height=50mm]{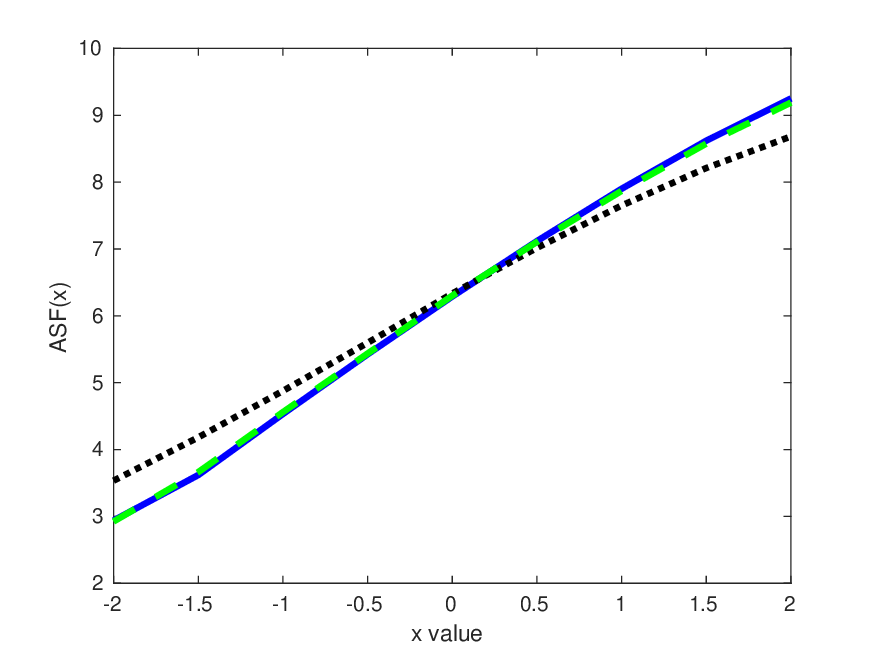}\\\end{tabular}
	\end{center}
	\par
	\textit{{\small Notes: true (solid), posterior (dashed), model-based (dotted). $n=1000$, $100$ simulations.}}
\end{figure}

We take as reference distribution $U\,|\, X \sim {\cal{N}}(0,s^2)$. In the simulated data generating process, $U$ is independent of $X$, distributed as a re-centered $\chi^2$ with mean zero and variance one. We simulate a scalar standard normal $X$. We set $n=1000$, $b_1=.5$, $b_0=0$, $s=1$, and $\mu$ as uniformly distributed between $-2$ and $2$. We estimate $b$ up to scale using maximum score (Manski, 1985). Specifically, using maximum score we regress $\boldsymbol{1}\{Y_i\leq j\}$ on $X_i$ and a constant, for all $j$, imposing that the coefficient of $X_i$ is one. We then regress the $J$ estimates on a common constant and the $\mu_j$, and obtain the implied estimate for $b$ by rescaling. For computation of maximum score, we use the mixed integer linear programming algorithm of Florios and Skouras (2008).

In Figure \ref{fig_ordered} we report the results for $J=3$ (left) and $J=10$ (right). We see that, when $J=3$, model-based and posterior estimators are similarly biased. By contrast, when $J=10$, the posterior estimator aligns well with the true average structural function, even though the model-based estimator is substantially biased. 

\clearpage
\section{Additional empirical results \label{App_Appli}}

\begin{figure}[h!]
	\caption{Distribution of posterior means of neighborhood effects\label{fig_neighb_PM}}
	\begin{center}
		\begin{tabular}{cc}
			Distribution function & Density\\
		\includegraphics[width=80mm, height=60mm]{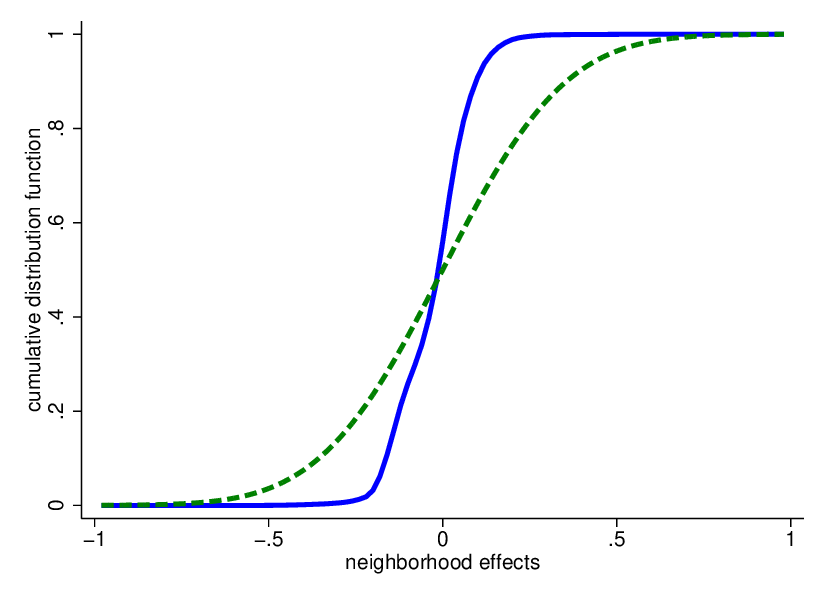}&	\includegraphics[width=80mm, height=60mm]{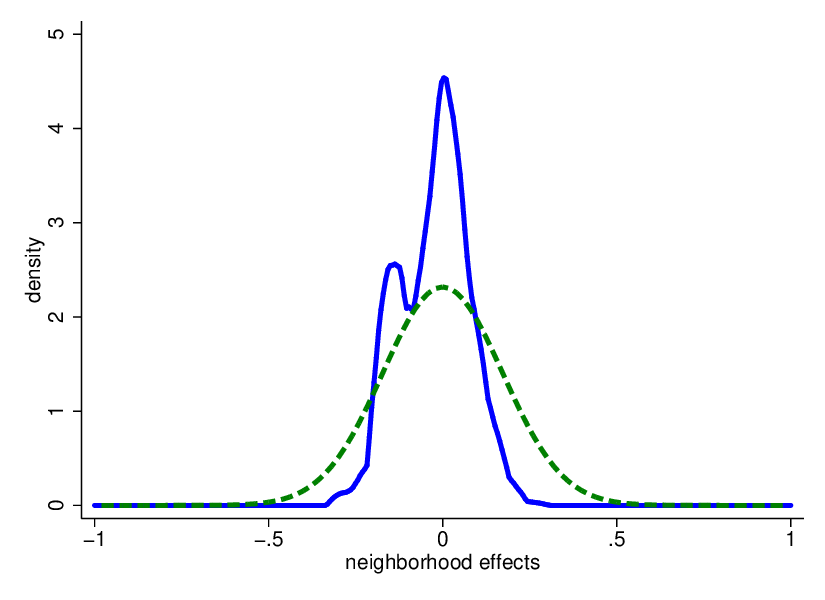}\\\end{tabular}
	\end{center}
	\par
	\textit{{\small Notes: Distribution of posterior means of $\mu_c$ (solid) and prior distribution (dashed). The distribution function is shown in the left graph, the density is shown in the right graph. Calculations are based on statistics available on the Equality of Opportunity website.}}
\end{figure}

\begin{figure}[h!]
	\caption{Posterior distribution of neighborhood effects, correlated random-effects specification\label{fig_neighb_FE_correl}}
	\begin{center}
		\begin{tabular}{cc}
				Distribution function & Density\\
			\includegraphics[width=80mm, height=60mm]{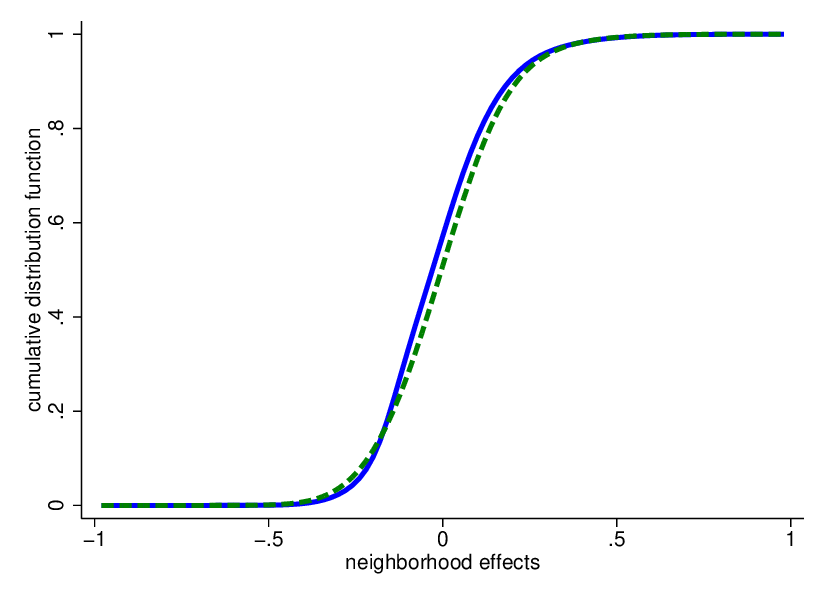}&	\includegraphics[width=80mm, height=60mm]{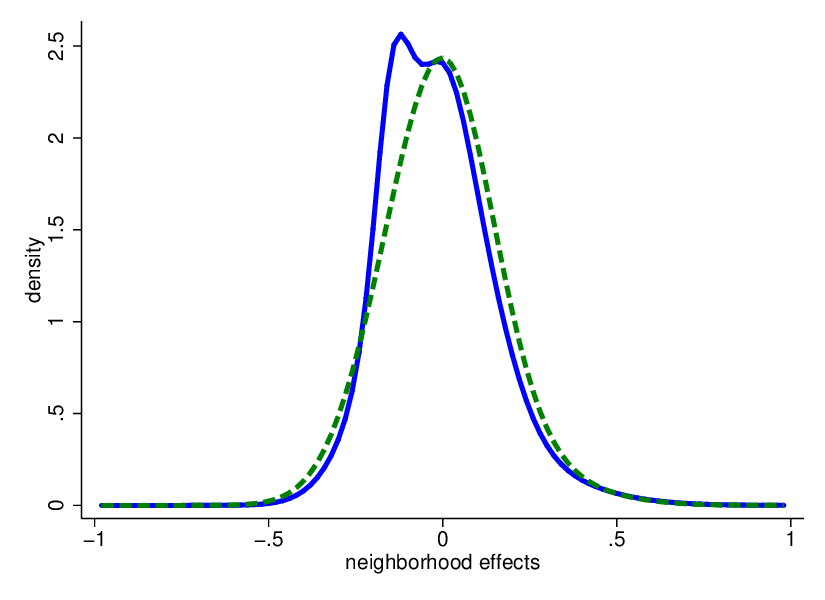}\\\end{tabular}
	\end{center}
	\par
	\textit{{\small Notes: Posterior distribution of $\mu_c$ (solid) and prior distribution (dashed), based on a correlated random-effects specification allowing for correlation between the place effects $\mu_c$ and the mean income of permanent residents $\overline{y}_c$. The distribution function is shown in the left graph, the density is shown in the right graph. Calculations are based on statistics available on the Equality of Opportunity website.}}
\end{figure}

\begin{figure}
	\caption{Distribution of neighborhood effects at the county level\label{fig_neighb_FE_cty}}
	\begin{center}
		\begin{tabular}{cc}
			Fixed-effects estimates & PAE\\
			\multicolumn{2}{c}{Distribution functions}\\
			\includegraphics[width=80mm, height=60mm]{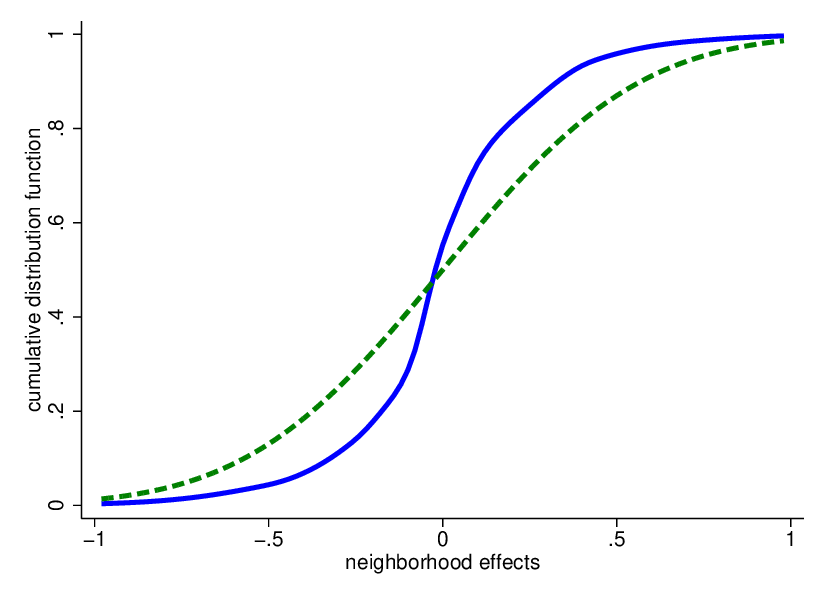}&	\includegraphics[width=80mm, height=60mm]{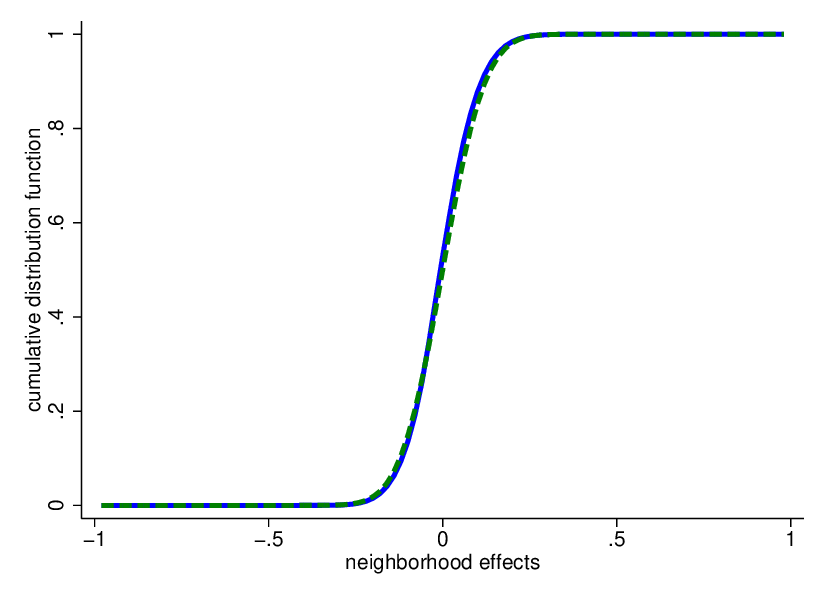}\\
				\multicolumn{2}{c}{Densities}\\
			\includegraphics[width=80mm, height=60mm]{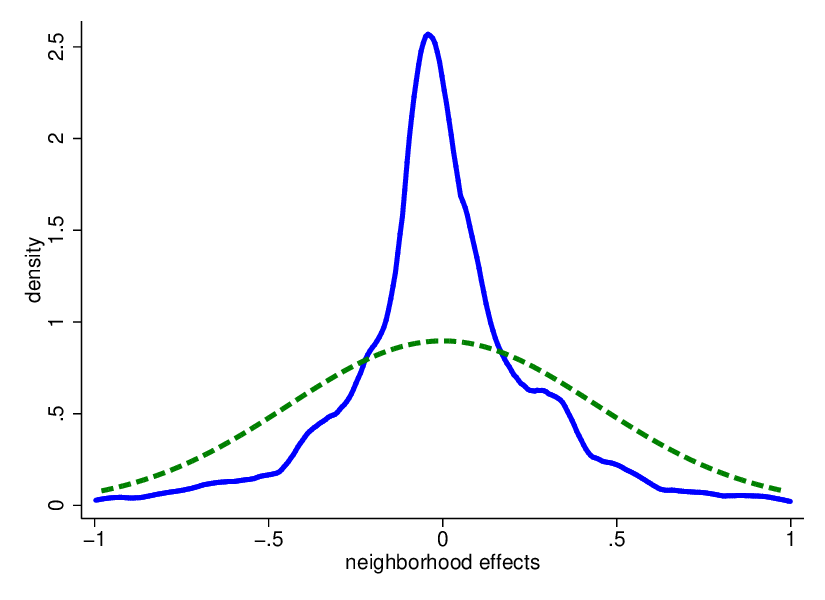}&	\includegraphics[width=80mm, height=60mm]{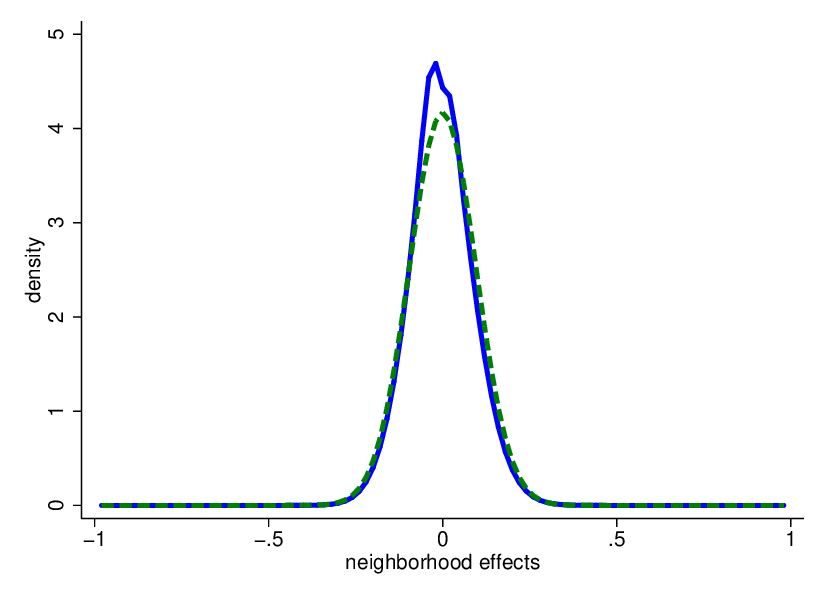}\\\end{tabular}
	\end{center}
	\par
	\textit{{\small Notes: In the left graphs, we show the distribution of fixed-effects estimates $\widehat{\mu}^{\rm county}_c$ (solid) and normal fit (dashed). In the right graphs, we show the posterior distribution of $\mu^{\rm county}_c$ (solid) and prior distribution (dashed). The distribution functions are shown in the top panel, the implied densities are shown in the bottom panel. Calculations are based on statistics available on the Equality of Opportunity website.}}
\end{figure}

\end{document}